\documentclass[mnsc,nonblindrev,copyedit]{informs3a} 

\usepackage{setspace, color, graphics,booktabs,enumerate}
\usepackage{graphicx}
\usepackage{subfigure}
\usepackage{amssymb,amsmath}
\usepackage{mathrsfs}
\usepackage{threeparttable}
\usepackage{ulem} 
\definecolor{DarkBlue}{rgb}{0,0.08,0.45}
\usepackage[backref = false, bookmarks, colorlinks = true, plainpages = false, citecolor = DarkBlue , urlcolor = DarkBlue, filecolor = DarkBlue, linkcolor = DarkBlue]{hyperref}
\usepackage[dvipsnames,table]{xcolor}

\renewcommand{\theequation}{\arabic{equation}}

\usepackage{natbib}
 \bibpunct[, ]{(}{)}{,}{a}{}{,}%
 \def\bibfont{\small\renewcommand{\baselinestretch}{1.3}\selectfont}%
\theoremstyle{TH}%
\newtheorem{theorem}{Theorem}
\newtheorem{lemma}{Lemma}
\newtheorem{proposition}{Proposition}

\newtheorem{assumption}{Assumption}

\theoremstyle{EX}

\newtheorem{definition}{Definition}

\ECRepeatTheorems

\EquationsNumberedThrough    

\MANUSCRIPTNO{}

\begin{document}




\TITLE{{\Large Dynamic Matching Under Patience Imbalance}
}


\ARTICLEAUTHORS{
\AUTHOR
{Zhiyuan Chen}
\AFF{School of Economics and Management, Wuhan University, Hubei, China 430072,\\\href{mailto:zhiyuanchen@whu.edu.cn}{zhiyuanchen@whu.edu.cn}}

\AUTHOR
{David Chen}
\AFF{School of Management and Economics, The Chinese University of Hong Kong, Shenzhen, China 518172,\\\href{mailto:davidchen@cuhk.edu.cn}{davidchen@cuhk.edu.cn}}

\AUTHOR
{Ming Hu}
\AFF{Rotman School of Management, University of Toronto, Toronto, Ontario, Canada M5S 3E6, \\\href{mailto:ming.hu@rotman.utoronto.ca}{ming.hu@rotman.utoronto.ca}}

\AUTHOR
{Yun Zhou}
\AFF{DeGroote School of Business, McMaster University, Hamilton, Ontario, Canada L8S 4E8, \\\href{mailto:zhouy185@mcmaster.ca}{zhouy185@mcmaster.ca}}
}

\ABSTRACT
{
We study a dynamic matching problem on a two-sided platform with unbalanced patience, motivated by ride-hailing and other on-demand services, in which long-lived supply accumulates over time with a unit waiting cost per period, while short-lived demand departs if not matched promptly. High- or low-quality agents arrive sequentially with one supply agent and one demand agent arriving in each period, and matching payoffs are supermodular. In the centralized benchmark, the optimal policy follows a threshold-based rule that rations high-quality supply, preserving it for future high-quality demand. In the decentralized system, where self-interested agents decide whether to match under an exogenously specified payoff allocation proportion, we characterize a welfare-maximizing Markov perfect equilibrium. We show that, unlike outcomes in the centralized benchmark or in full-backlog markets, the equilibrium exhibits distinct matching patterns in which low-type demand may match with high-type supply even when low-type supply is available. We further show that, unlike settings in which both sides have long-lived agents and perfect coordination is impossible, the decentralized system can always be perfectly aligned with the centralized optimum by appropriately adjusting the allocation of matching payoffs across agents on both sides. Finally, when the arrival probabilities for H- and L-type arrivals are identical on both sides, we compare social welfare across systems with different patience levels: full backlog on both sides, one-sided backlog, and no backlog. In the centralized setting, social welfare is weakly ordered across systems, with incremental social welfare gains from making both sides patient weakly smaller than the gain from making just one side patient. However, in the decentralized setting, the social welfare ranking across the three systems depends on the matching payoff allocation rule and the unit waiting cost, and enabling patience can either increase or decrease social welfare.
}

\maketitle

%

\section{Introduction}\label{sec:introduction}

Many modern platforms operate as two-sided dynamic matching markets in which participants arrive over time, are vertically differentiated by type, and exhibit a patience imbalance between the two sides. On ride-hailing platforms (e.g., Uber, Lyft, and Didi), riders and drivers arrive sequentially and stochastically; riders may vary by destination length (e.g., an airport trip versus a nearby grocery run), while drivers may vary by vehicle class (e.g., UberX versus Black/Premier), which carry different per-mile rates. A natural specification is that the matching payoff takes a multiplicative form in rider and driver types, implying that matching payoffs are supermodular. Matching is implemented either centrally by the platform (i.e., the platform selects a driver for an incoming trip request from the rider) or through match offers that participants may accept or reject (i.e., the platform offers the trip to a driver, who can accept, ignore, or decline, and the rider may also cancel the order if the waiting time is too long.) An often-overlooked feature of these markets is the imbalance in patience between the two sides. Riders are typically time-sensitive and may abandon the platform if no match is available or if the expected wait is too long, whereas drivers can remain available while waiting for requests. Using terminology from inventory theory, this market operates under a one-sided backlog regime: supply backlogs and waits to be matched, whereas demand is lost sales that depart if not served upon arrival. Accordingly, we refer to demand-side agents as short-lived and supply-side agents as long-lived. Throughout the paper, we use long-lived and patient interchangeably for agents who can wait but incur a waiting cost per unit time, and likewise use short-lived and impatient interchangeably for agents who will leave if not matched upon arrival. Table \ref{Table.1} summarizes the two-sided markets that exhibit similar imbalances in patience levels between the two sides.
\begin{table}[htpb]
	\centering
	\caption{\label{Table.1} Examples of two-sided platforms with unbalanced patience levels}
	\begin{threeparttable}
		\begin{tabular}{llll}
			\toprule
			Two-sided markets & Supply side & Demand side & Matching mechanism\\
			\midrule
			Ride-hailing platforms (e.g., Uber, Lyft, Didi) & Drivers & Riders\tnote{a} & (De-)centralized \\ 
			Organ transplant exchange (e.g., liver) & Patients & Organs\tnote{b} & Centralized\\
			Freelancer platforms (e.g., Upwork, Fiverr) & Freelancers & Tasks\tnote{c} & Decentralized\\
            Marriage in traditional societies & Men & Women\tnote{d} & Decentralized\\
			\bottomrule
		\end{tabular}
		\begin{tablenotes}
			\item[a] It is estimated that the median waiting time for an Uber car is shorter than 5 minutes.
			\item[b] Harvested organs will perish in a few hours, while patients may remain on the waiting list for one or two years.
			\item[c] The task listing usually comes with a deadline, and the turnover can be fast due to fulfillment from third parties.
            \item[d] Subjected to traditional social norms that prescribe women to marry at younger ages than men, and burdened by a narrower childbearing age span, women exhibit more impatience than men, thereby prompting a stronger inclination toward early marriage.
		\end{tablenotes}
	\end{threeparttable}
\end{table}

\citet{baccara2020optimal} establish the optimal dynamic matching policy for a matching platform with long-lived agents on both sides, where agents can only leave the platform after being successfully matched. Their model is referred to as the full backlog case. However, the optimal dynamic matching policy remains poorly understood in the one-sided backlog case, where demand agents are short-lived. Understanding the structure of this optimal policy would have far-reaching implications for both centralized and decentralized market design. Additionally, it could help bridge the gap between dynamic matching with long-lived agents on both sides and that with short-lived agents on both sides, the latter effectively reduced to a sequence of independent one-shot (static) matching problems.


Motivated by the practical examples listed in Table \ref{Table.1}, we study the centralized and decentralized dynamic matching systems and analyze the impact of unbalanced patience. We focus on an infinite-horizon matching process on a two-sided platform. 
 Agents arrive at the platform sequentially: one supply agent and one demand agent arrive in each period. Supply agents are long-lived, remaining on the platform until matched, whereas demand agents are impatient and leave the platform if unmatched upon arrival. Each side has two types of vertically differentiated agents: H-type and L-type, representing high- and low-quality, respectively. The matching payoffs exhibit homogeneous preferences and satisfy the supermodular property, which intuitively means that the payoff from an (H, H) match (i.e., a match of H-type supply and H-type demand) and an (L, L) match is greater than that from an (H, L) match and an (L, H) match (we consistently use the first argument to denote the supply type and the second argument to denote the demand type). At the end of each period, unmatched supply agents incur per-period waiting costs. We characterize the optimal dynamic matching policy in the centralized setting and the welfare-maximizing equilibrium matching policy in the decentralized setting, i.e., the equilibrium that maximizes long-run social welfare among all possible equilibria. We then compare the induced matching policies and social welfare across the two settings. Finally, to isolate the role of patience imbalance, we benchmark welfare and policy structures across three systems when the arrival probabilities on both sides are identical: (i) a full backlog system (patient agents on both sides), (ii) a one-sided backlog system (impatient demand agents and patient supply agents), and (iii) a no backlog system (impatient agents on both sides). We highlight our main findings and contributions as follows. 

First, we characterize the structure of the optimal dynamic matching policy in the centralized setting. Upon the arrival of an H-type demand agent, the planner matches her with an available H-type supply agent whenever possible; if no such supply agent is present, the planner matches her with an L-type supply agent. Upon the arrival of an L-type demand agent, it is optimal to match her with an available L-type supply agent. If no L-type supply agent is available, the planner matches the L-type demand agent with an H-type supply agent only when the number of H-type supply agents exceeds a threshold. This optimal policy exhibits a threshold structure reminiscent of that in \citet{baccara2020optimal} for the full backlog system, where both sides are long-lived. However, the relevant thresholds and the implied matching priorities differ in our environment due to one-sided impatience on the demand side. In addition, we establish that optimal long-run social welfare is decreasing in the unit waiting cost.

Second, we study the decentralized setting in which matches form only upon mutual acceptance, and the matching payoff is split between the two agents in the match pair. We establish that, under a first-come-first-served (FCFS) discipline, there exists an equilibrium that maximizes long-run social welfare among all equilibria. The induced matching pattern is characterized by a priority-and-threshold structure. Specifically, an (H, H) match is formed whenever H-type agents are simultaneously available on both sides. An (L, H) match is formed only when no H-type supply agent is available. Conversely, an (H, L) match is formed only when the number of H-type supply agents exceeds a threshold. Finally, an (L, L) match is formed only when the number of H-type supply agents falls below this threshold and the number of L-type supply agents exceeds another threshold. This equilibrium contrasts with \citet{baccara2020optimal}, where, in the full backlog system with long-lived agents on both sides, the welfare-maximizing decentralized equilibrium admits (L, L) matches whenever L-type agents are available on both sides. Our one-sided impatience setting, by contrast, restricts (L, L) matches via endogenized thresholds that reflect the trade-off between preserving high-quality supply for future high-quality demand and limiting supply-side waiting costs. Moreover, we show that, unlike in the centralized planner’s problem, equilibrium social welfare in the decentralized setting may increase with the unit waiting cost. The intuition is that, while a higher unit waiting cost raises the per-period waiting burden of each long-lived supply agent, it also makes supply agents less patient. As the unit waiting cost continuously increases, the equilibrium matching thresholds on the queue length decrease discretely. This downward shift reduces the average number of waiting supply agents. The resulting savings in total waiting costs, driven by fewer long-lived supply agents waiting, can more than offset the aggregate increase in waiting costs, thereby raising overall decentralized social welfare.



Third, we demonstrate that when demand agents are more impatient than supply agents (i.e., supply agents are patient and demand agents are impatient), the decentralized setting can be effectively coordinated by appropriately tuning the matching payoff allocation ratio between matched supply and demand agents. By reshaping agents’ acceptance incentives, the payoff allocation influences equilibrium matching behavior and, in turn, the total queue length in the system. This coordination mechanism also differs from the results in \citet{baccara2020optimal}. In the full-backlog system with long-lived demand agents, however, such coordination is not achievable under FCFS. In that environment, changing the payoff split merely reallocates congestion between the two sides by shifting waiting from one queue to the other, without reducing the aggregate queue length. Consequently, the decentralized outcome continues to exhibit an inefficiently larger total backlog than the centralized optimum.

Finally, we investigate the value of patience by comparing long-run social welfare across three systems when the arrival probabilities for H- and L-type arrivals are identical on both sides. In the centralized setting, we find that social welfare is weakly ordered across the three systems: the full backlog system yields weakly higher welfare than the one-sided backlog system, which, in turn, yields weakly higher welfare than the no-backlog system. The intuition is that patience expands the planner’s feasible decision set by allowing unmatched agents to be carried forward, thereby postponing low-value matches and waiting for higher-value pairings. Moreover, we find that patience exhibits diminishing marginal returns: moving from ``no one can wait'' to ``one side can wait'' yields a larger increase in social welfare than moving from ``one side can wait'' to ``both sides can wait.'' This result has a market-design implication for centralized planners: rather than attempting to increase patience on both sides, a planner may want to prioritize ensuring that at least one side of the market is patient, because one-sided patience already delivers most of the social welfare improvement achievable. The patient side acts as a buffer stock that absorbs randomness in arrivals on the other side, enabling the planner to smooth mismatches over time and implement higher value (more assortative) pairings. If patience imbalance does not arise endogenously, a planner can be better off by inducing patience on at least one side through institutional design and incentive schemes, rather than devoting comparable effort to increasing patience on both sides.


In the decentralized setting, however,  social welfare in dynamic matching systems with different patience levels critically depends on the matching payoff allocation rule, which determines the private incentives for agents to wait for a better match. On the one hand, a small payoff share for H-type supply agents makes them unwilling to wait, which also weakens H-type demand agents' incentive to wait if they can. Less waiting by H-type agents favors L-type agents; matches occur more quickly, even though the generated rewards can be low. Thus, the steady state population in the system is low. In such a thin market, additional patience increases opportunities for future assortative matches, while the induced increase in waiting costs is limited; therefore, social welfare rises as the system moves from the no-backlog system to the one-sided backlog system to the full backlog system. On the other hand, a large payoff share for supply agents increases the incentive of H-type agents on both sides to wait for assortative matches, if they can. This, in turn, makes it more difficult for L-type agents to be matched, raising the system's steady-state population and leading to congestion. In such a thick market, aggregate waiting costs increase with the population of the systems, while the marginal improvement in match quality exhibits diminishing returns; hence, social welfare can fall as patience is expanded. Under approximately equal payoff sharing, whether waiting is beneficial depends on the value of the unit waiting cost. As an illustration, in the modern marriage market where women are more patient compared to traditional society, the ``unmarried women'' phenomenon emerges as an equilibrium from the interaction between shifting bargaining power and the rational adoption of more patient (i.e., selective) matching strategies, although social welfare may suffer from a utilitarian perspective.



In this paper, we refer to matches between agents of the same type, i.e., (H, H) and (L, L), as assortative matching. Matches between agents of different types, i.e., (H, L) and (L, H), will be referred to as cross-matching. Throughout the paper, we use $\lfloor x\rfloor$ to denote the greatest integer less than or equal to $x$ and $x^+$ to denote $\max\{x, 0\}$. We also employ the term ``increasing/decreasing" in a weak sense; that is, ``increasing/decreasing" means ``non-decreasing/non-increasing." All proofs are available in the Online Appendix.

\section{Literature Review}\label{sec:literature review}


For classic results on static matching in two-sided markets, we refer readers to \citet{roth1992twosidedmarket}; for a survey of the rapidly growing literature on dynamic matching, see \citet{doval2025dynamic}. Our work is related to the following four streams of literature.

\subsubsection*{Dynamic matching in two-sided markets.}

 Our work is most closely related to \citet{baccara2020optimal}, who study a similar dynamic matching problem but assume that agents on both sides are patient: they incur the same finite waiting costs and remain in the system until matched. In contrast, we introduce unbalanced patience by setting one side's waiting cost to be effectively infinite, resulting in immediate departure if unmatched. \citet{baccara2020optimal} find that the centralized optimal policy and the decentralized equilibrium share a similar threshold-based structure. However, under an FCFS protocol, the decentralized system cannot be fully coordinated to achieve the centralized optimum solely by adjusting the matching payoff allocation proportion. We depart from these findings in two key respects. First, we demonstrate that under unbalanced patience, the centralized and decentralized policies differ fundamentally in both matching priorities and structural thresholds. Second, we show that the decentralized equilibrium in our setting can be perfectly coordinated with the centralized optimum by appropriately tuning the payoff allocation proportion. Thus, the distinction between symmetric and asymmetric patience leads to qualitatively different implications for market design and efficiency. 

\citet{LOERTSCHER2022105383} study an infinite-horizon dynamic trading model (inspired by stock exchanges) with balanced arrivals. Unlike our specific waiting-cost structure, they model patience via a common discount factor, yet they similarly focus on long-lived agents on both sides. In addition, their setting restricts trade feasibility (e.g., low-value buyers cannot trade with high-cost sellers), whereas we assume matching occurs on a complete bipartite graph. \citet{hu2022dynamic} develop a finite-horizon framework for dynamic matching, characterizing centralized optimal policies and heuristics. In contrast, our analysis explicitly characterizes the decentralized equilibrium driven by agents' strategic incentives. \cite{aouad2025centralized} analyze how delegating pricing power to suppliers affects matching efficiency using a fluid approximation. They model agent impatience via exogenous departure rates on both sides. Their matching process follows a uniform random matching rule, whereas we adopt an FCFS protocol. Additionally, while they focus on pricing, we treat the payoff allocation proportion as an exogenous lever to coordinate the decentralized system. 

There is a growing stream of literature that investigates optimal matchmaking strategies and assortment planning to improve equilibrium outcomes in two-sided markets \citep{ashlagi2022assortment, aouad2023online, shi2023optimal, shi2025optimal}.
These studies assume a static pool of agents or finite sequential interactions, focusing on managing agent choices and search costs. In contrast, we study a dynamic matching process with stochastic agent arrivals and patience imbalance on the two sides.

\subsubsection*{Dynamic matching for barter exchange.} A distinct stream of literature studies dynamic matching for barter exchange. These studies can be categorized into models with homogeneous agents \citep{anderson2017efficient, akbarpour2020thickness, ashlagi2023matching} and heterogeneous agents \citep{herbst2016dynamic, ashlagi2019matching, aouad2022dynamic, freund2024two}. This stream of literature typically assumes that matching is probabilistic and that any two agents can trade. In contrast, we study a two-sided matching market in which same-side matches are prohibited. We further depart from this stream of literature by considering an unbalanced patience model: agents on one side are patient and wait in queues, while agents on the other side are impatient and leave immediately if unmatched.

\subsubsection*{Dynamic matching in queueing systems.}

Queueing provides a natural framework for modeling dynamic matching, though computing stationary matching rates across customer and server types is analytically challenging. Consequently, much of the literature relies on algorithmic approximations. \citet{ozkan2020dynamic} model a two-sided dynamic matching process through a queueing model. The planner maximizes the cumulative number of matchings over a finite horizon. They solve this problem via a continuous linear program (CLP) and establish the asymptotic optimality of the CLP-based policy in large markets. \citet{castro2020matching} study a parallel matching queue in which all agents are impatient and will abandon at exogenous rates. Unlike our model, they impose a constrained compatibility structure in which high-type demand can only be served by high-type supply. They provide explicit product forms of the steady-state distributions of this system.

A stream of literature utilizes queueing networks to study dynamic matching. \citet{afeche2022optimal} analyze multi-class, multi-server systems in which customers have preferences over servers, focusing on the design of matching topologies to balance waiting costs and matching rewards. \citet{caldentey2025designing} extend \citet{afeche2022optimal} by allowing consumers to choose the service class they wish to join. Other dynamic matching papers using queuing networks consider either infinitely patient agents \citep{gurvich2014dynamic, buke2015stabilizing, nazari2016optimal} or impatient agents with exogenous abandonment distributions \citep{kohlenberg2024greedy, kohlenberg2024cost, aveklouris2025matching}. These works primarily focus on centralized control and asymptotic optimality in large markets. They do not address the strategic waiting behavior of agents in a decentralized game, which is a central focus of our analysis. 

Some papers utilize queueing models to investigate the strategic behavior of agents in dynamic matching; for example, the cherry-picking behavior in ride-hailing platforms \citep{chu2018harnessing, castro2021randomized}. In their settings, drivers are strategic and prefer to match with high-value riders. However, there is a single driver type; consequently, riders are not strategic and have no driver preferences. \citet{leshno2022dynamic} studies the allocation of stochastically arriving items to a waiting list of strategic agents with heterogeneous preferences. A key distinction is that, in \citet{leshno2022dynamic}, items are passive and are assigned upon arrival, whereas our model features strategic agents on both sides, and agents on the demand side may leave without matching due to impatience. 


\subsubsection*{Resource allocation problems with multiple types.}
Our work also relates to resource allocation with multiple types. \citet{shumsky2009dynamic} and \citet{yu2015dynamic} analyze capacity allocation with upgrading, providing optimal policies and heuristics. Recent work extends this to resource pooling \citep{zhong2018resource} and multi-product inventory systems \citep{tang2025multiproduct, hao2025robust}. For a unifying perspective on a broad class of dynamic resource allocation problems, see \citet{balseiro2024survey}.
Our model differs in two aspects. First, in these studies, capacity is fixed at the beginning of the horizon, whereas our supply and demand evolve dynamically through stochastic arrivals. Second, they typically model one-way substitution (upgrading), which corresponds to an incomplete compatibility graph, whereas we consider a complete bipartite graph in which cross-matching is possible in both directions.


\section{Model}\label{sec:model}

Consider a two-sided matching platform with a supply side and a demand side. An agent arriving on the supply side is called a supply agent (he), while an agent arriving on the demand side is called a demand agent (she). Following \citet{baccara2020optimal}, we assume that one supply agent and one demand agent arrive in each period. On each side, there are two types of agents: H-type and L-type. H-type agents are of higher quality than L-type agents. 
In each period, a newly arrived supply agent is H-type with probability $p$ and L-type with probability $1-p$. Similarly, a newly arrived demand agent is H-type with probability $q$ and L-type with probability $1-q$, where $p, q\in[0,1]$.


Supply agents are long-lived and stay on the platform until they are matched. They incur a waiting cost of $h$ per period if unmatched. These agents are categorized into two distinct queues: H-type and L-type. In contrast, demand agents are impatient and will leave the platform if they are unmatched upon arrival.  When a match is formed, a matching payoff is generated, and both agents exit the platform. 

The matching payoff is determined by the types of agents in the matching pair. Let $r_{ij}$ denote the matching payoff between an $i$-type supply agent and a $j$-type demand agent for $i,j\in\{H, L\}$, which are also referred to as an $(i, j)$ pair, with the first argument referring to a supply type and the second argument referring to a demand type. We use $\mathbf{r}=(r_{HH}, r_{HL}, r_{LH}, r_{LL})$ to represent the payoff vector and assume that these matching payoffs satisfy the following two properties.

\begin{assumption}\textsc{(Homogeneous Preferences)}\label{as:hp}
	$r_{HH}\ge r_{HL}\ge r_{LL}$ and  $r_{HH}\ge r_{LH}\ge r_{LL}$.
\end{assumption}

The assumption of homogeneous preferences states that all agents share the same preferential ranking: matching an (H, H) pair yields the highest payoff, while matching an (L, L) pair yields the lowest. Matches involving  (H, L) or (L, H) result in an intermediate payoff. In other words, for any type of supply or demand agent, matching with an H-type agent generates a higher payoff than matching with an L-type agent. Additionally, we assume that the matching payoffs satisfy the following supermodular property.

\begin{assumption}\textsc{(Supermodularity)}\label{as:sm}
	$r_{HH}+r_{LL}\ge r_{HL}+r_{LH}$.
\end{assumption}

The supermodular payoff property is a common assumption in the matching literature, e.g., \citet{becker1991assortative},  \citet{LOERTSCHER2022105383}, and \citet{baccara2020optimal}. 
This property applies to situations in which the matching payoff grows increasingly rapidly as the quality of the supply agent or the demand agent increases.  
As an example, let us consider an additive payoff structure, with $r_{ij}=r_i^s+r_j^d$ for $i,j\in\{L,H\}$ and  $r_i^s$ and $r_j^d$ being increasing in $i$ and $j$, respectively.
Under the additive payoff structure, the matching payoff clearly grows linearly as the quality of the supply or demand agent increases.
We now modify the payoffs as follows.
Let $r_{LH}\leftarrow r_{LH}+\delta^d_H$ and $r_{HL}\leftarrow r_{HL}+\delta^s_H$, where $\delta_{H}^d$ ($\delta_H^s$) is nonnegative and can be considered as a ``bonus" payoff contributed by the high-quality demand (supply) agent. 
Further, let $r_{HH}\leftarrow r_{HH}+\gamma^s_H+\gamma^d_H$, where  $\gamma_{H}^d\ge \delta_H^d$ ($\gamma_H^s\ge \delta_H^s$) is the higher bonus payoff contributed by the high-quality demand (supply) agent when two type-H agents are matched. It is easy to see that the modified payoff structure (which includes linear payoffs as a special case) satisfies both homogeneous preferences and supermodularity.


For another example, if the matching payoff has a multiplicative form  $r_{ij}= r^s_i r^d_j$, where $r^s_i$ and $r^d_j$ represent the marginal contributions to the matching payoff from the $i$-type supply agent and the $j$-type demand agent, respectively, where $r^s_{H}\ge r^s_{L}$ and $r^d_{H}\ge r^d_{L}$. One can readily verify that those payoffs also satisfy both homogeneous preferences and the supermodular property. 

We assume that the waiting costs for H-type and L-type supply agents are identical and equal to $h>0$ per unit time, since the heterogeneity of supply agents is characterized by type-dependent matching payoffs.
In the following sections, we first examine a centralized platform that makes all matching decisions to maximize long-run average social welfare (defined as the total matching payoffs minus the total waiting costs) over an infinite horizon.  Then, we explore a decentralized platform where supply and demand agents can accept or reject matches to maximize their own surplus. Finally, we will compare the outcomes between the centralized and decentralized platforms.

\section{Centralized Matching}\label{sec:centralized model}

In the centralized setting, matching decisions are made by the platform. The state of the system in each period, after a pair of supply and demand agents arrive in the platform, can be represented by a vector $\mathbf{s}_{t}=(x_{H}, x_{L}, y_{H}, y_{L})$ with $x_{H}, x_{L} \in \mathbb{Z}_{+}$ indicating the number of type-H and type-L supply agents and $y_{H}, y_{L}\in \{0, 1\}$ indicating the number of type-H and type-L demand agents, where each element represents the number of each type available for matching. Specifically, $x_{H}$ and $x_{L}$ represent the number of H-type and L-type agents on the supply side, while $y_{H}$ and $y_{L}$ represent the number of H-type and L-type agents on the demand side, respectively. We use $\mathbf{a}_{t}\in\{u_{\phi}, u_{HH}, u_{HL}, u_{LH}, u_{LL}\}$ to denote the matching decision. 
In particular,  $u_{\phi}$ denotes no matching and $u_{ij}$ denotes matching an $i$-type supply with a $j$-type demand, for $i,j\in\{L,H\}$.

Since one demand agent arrives in each period, at most one match can occur. 
Therefore, the matching payoff generated in each period can be formulated as $\mathbf{r}\cdot \mathbf{a}_{t}:=\sum_{i,j\in\{L,H\}}r_{ij}\mathbf{1}_{\mathbf{a}_t=u_{ij}}$. The total waiting costs incurred by unmatched supply agents are $h(x_{H}+x_{L})$ if no match is formed ($\mathbf{a}_{t}=u_{\phi}$) or $h(x_{H}+x_{L}-1)$ if a match occurs ($\mathbf{a}_{t}=u_{ij}$ for some $i,j\in\{L,H\}$). Then, social welfare generated in each period is given by 
\begin{eqnarray*}
	R(\mathbf{s}_{t}, \mathbf{a}_{t})=\mathbf{r}\cdot \mathbf{a}_{t}-h(x_{H}+x_{L}-\mathbf{1}_{\mathbf{a}_{t}\neq u_{\phi}}).
\end{eqnarray*}

Under a matching policy $\pi$ which maps the state $\mathbf{s}_t$ into a matching decision $\mathbf{a}_t=\mathbf{\pi}(\mathbf{s}_{t})$, the long-run average profit of the central platform, which we define as social welfare, is given by 
\begin{eqnarray*}
	W(\pi)\equiv \lim_{T\rightarrow \infty}\frac{1}{T}\mathbf{E}\left[\sum^{T}_{t=1}R(\mathbf{s}_{t}, \mathbf{\pi}(\mathbf{s}_{t}))\right],
\end{eqnarray*}
where $\pi(\mathbf{s}_{t})$ specifies the matching policy in state $\mathbf{s}_{t}$ under $\pi$. The objective is to characterize the optimal matching policy that maximizes $W(\pi)$.

Without loss of generality, we assume that the system is empty at time 0. Matching decisions are made after the first pair of agents arrive at the platform, starting from period 1. Note that the state $\mathbf{s}_{t+1}$ in period $t+1$ only depends on $\mathbf{s}_{t}$, the matching decision in period $t$, and the types of arrived agents in period $t+1$. Thus, we can formulate the dynamic matching problem as a Markov Decision Process (MDP) and use the framework to characterize the optimal matching policy (see the Online Appendix).

\subsection{Optimal Dynamic Matching Policy}\label{subsec:optimal policy}
In this subsection, we describe the structure of the optimal matching policy. In Lemma \ref{lem: centralized threshold type} below, we demonstrate that it is not optimal to hold a large number of supply agents on the supply side.

\begin{lemma}\label{lem: centralized threshold type} If the total number of supply agents exceeds a certain threshold, it is optimal to match the arrived demand agent with one of the supply agents.
\end{lemma}


Lemma \ref{lem: centralized threshold type} implies that the state space of the dynamic matching problem is finite. Based on this lemma, we characterize the optimal matching policy as follows.

\begin{theorem}[{\sc Optimal Dynamic Matching Policy}] \label{thm:optimal mechanism of centralized system} 
The optimal dynamic matching policy can be characterized by a threshold $k^{ce}\in \mathbb{Z}_+$ such that
	\begin{enumerate}
		\item (Greedy Matching) When an H-type demand agent arrives, it is optimal to match her with an H-type supply agent if one is available. Otherwise, it is optimal to match her with an L-type supply agent.		
		\item (Threshold Matching) When an L-type demand agent arrives, it is optimal to match her with an L-type supply agent if one is available. Otherwise, it is optimal to match her with an H-type supply agent only if the number of H-type supply agents is larger than $k^{ce}$. 
	\end{enumerate}
\end{theorem}


Theorem \ref{thm:optimal mechanism of centralized system} shows that when an H-type demand agent arrives, the platform prefers to match her with an H-type supply rather than an L-type supply because the payoff is higher. Furthermore, the platform has no incentive to withhold any supply agent in the face of an H-type demand, since they would not receive an even higher type than a high type.

When an L-type demand agent arrives, the platform prefers matching her with an L-type supply rather than an H-type supply, if both are available. This is because, due to the supermodular payoff structure (Assumption \ref{as:sm}), it is more profitable to withhold an H-type supply agent (than an L-type supply agent) for a better match in the future. 
In fact, the platform does not have an incentive to withhold any L-type supply agent for a future H-type demand agent, as that demand agent will arrive with a supply agent, which she can potentially match with to receive a matching payoff no less than the payoff from matching with the withheld L-type supply agent. 
However, if at any point only H-type supply agents are available, the platform will (resp., will not) perform a match when the number of supply agents exceeds (resp., falls short of) a threshold.
Intuitively, with many (resp., few) H-type supply agents, it takes a long (resp., short) time for a withheld supply agent to get his turn to match with an H-type demand agent, making the option of withholding a supply agent in the current period less (resp., more) attractive.

We note that the structure of the optimal matching policy under one-sided backlog differs from that under full backlog, as discussed in \citet{baccara2020optimal}. In their study, it is optimal to perform greedy matching of (H, H) and (L, L) pairs as much as possible and to cross-match H-type supply with L-type demand (or L-type supply with H-type demand) until a threshold is reached. However, in their paper, the thresholds for matching the pairs (L, H) and (H, L) are the same (i.e., due to symmetric payoffs and costs as well as symmetric arrival rates), whereas in our model, the threshold for matching a pair (L, H) is zero, while the threshold for matching a pair (H, L) is a non-negative value $k^{ce}$.  More critically,  the matching priorities are different: under our one-sided backlog setting, it is optimal to first greedily match the (H, H) pairs, then greedily cross-match the (L, H) pairs, followed by greedily matching the (L, L) pairs, and finally, by cross-matching the (H, L) pairs down to the threshold on the number of H-type supply, if possible.

Under any policy of the threshold-level structure in Theorem \ref{thm:optimal mechanism of centralized system}, the underlying Markov chain of the centralized matching system is ergodic since the state space is finite. Thus, for a given threshold, we can compute the unique stationary distribution of the Markov Chain. We can then determine the optimal threshold $k^{ce}$ by trading off the cost and benefit of waiting. We provide the closed-form expression for the optimal threshold in the following proposition.

\begin{proposition}[{\sc{Optimal Threshold on H-Type Supply}}]\label{cen: prop optimal threshold}  Under the optimal dynamic matching policy, the optimal threshold $k^{ce}$ is given by	
\begin{align*}
	k^{ce}=\left\{
\begin{array}{ll}
 \left\lfloor\frac{\ln\left(h(1+\delta)+p(1-q)r(1-\delta)^2-\sqrt{\left(h(1+\delta)+p(1-q)r(1-\delta)^2\right)^2-4h^2\delta}\right)-\ln\left(2h\delta\right)}{\ln\delta}\right\rfloor& \hbox{if $p>q$,}\vspace{2mm} \\ 
 \left\lfloor\frac{-h+\sqrt{h^2+4hp(1-q)r}}{2h}\right\rfloor& \hbox{if $p=q$,}\vspace{2mm}\\
 \left\lfloor\frac{\ln\left(h(1+\delta)+p(1-q)r(1-\delta)^2+\sqrt{\left(h(1+\delta)+p(1-q)r(1-\delta)^2\right)^2-4h^2\delta}\right)-\ln\left(2h\delta\right)}{\ln\delta}\right\rfloor& \hbox{if $p<q$,}
\end{array}	
	\right.
\end{align*}
where $\delta:=\frac{q(1-p)}{p(1-q)}$ and $r:=r_{HH}+r_{LL}-r_{LH}-r_{HL}$. Moreover, $k^{ce}$ increases in $r$ and decreases in $h$.
\end{proposition}

Proposition \ref{cen: prop optimal threshold} shows that the optimal threshold increases with the degree of supermodularity $r=r_{HH}+r_{LL}-r_{HL}-r_{LH}$. For example, if the payoffs satisfy the multiplicative form, i.e., $r_{ij}=r_{i}^{s}r_{j}^{d}$,  the degree of supermodularity can be expressed as $r_{HH}+r_{LL}-r_{HL}-r_{LH}=(r^{s}_{H}-r^{s}_{L})(r^{d}_{H}-r^{d}_{L})$. In this case, $r^{s}_{H}-r^{s}_{L}$ and $r^{d}_{H}-r^{d}_{L}$ represent the benefits of waiting for a better matching as an L-type supply agent and an L-type demand agent, respectively. Therefore, a higher degree of supermodularity implies that the platform planner has a stronger incentive to retain more H-type supply agents for potentially more rewarding matches. The optimal threshold $k^{ce}$ also decreases with the unit waiting cost $h$. If the unit waiting cost is extremely high, $k^{ce}$ will be zero, all supply agents will be matched upon arrival, and the platform planner will be reluctant to retain any supply agents under the optimal matching policy.

Note that when the probability of an H-type supply agent arriving is equal to the probability of an H-type demand agent arriving,  i.e., $p=q$,  the optimal threshold can be simplified as 
\begin{equation}\label{eq:centralized threshold}
    k^{ce}=\left\lfloor\frac{-h+\sqrt{h^2+4hp(1-p)(r_{HH}+r_{LL}-r_{LH}-r_{HL})}}{2h}\right\rfloor.
\end{equation}
In this case, the optimal threshold increases in $p(1-p)$, which represents the probability of having an incongruent pair, i.e., (H, L) or (L, H), arriving within a given period. Moreover, the optimal threshold reaches its maximum value when $p=1-p=1/2$. This implies that if the arrival processes of all types are perfectly balanced, the centralized planner exhibits the greatest patience and accommodates the longest waiting queue of supply agents.

With the established optimal threshold expression, we can also derive the optimal social welfare expression.
\begin{proposition}[{\sc{Optimal Social Welfare}}]\label{prop: centralized welfare} In the centralized setting, the optimal social welfare is given by 
\begin{align}\label{eq:optimal sw}
	W^{ce}=\left\{
\begin{array}{ll}
 qr_{HH}+(1-q)r_{LL}+p(1-q)(r_{HH}-r_{LL})(1-\delta)-p(1-q)r\frac{1-\delta}{1-\delta^{k^{ce}+1}}-k^{ce}h& \hbox{if $p\neq q$,}\vspace{2mm} \\ 
 pr_{HH}+(1-p)r_{LL}-\frac{p(1-p)r}{k^{ce}+1}-k^{ce}h& \hbox{if $p=q$.}
 \end{array}	
	\right.
 \end{align}
 Moreover, $W^{ce}$ is decreasing in $h$.
\end{proposition}

Note that the supermodular property suggests that assortative matching is more desirable than cross-matching. We can express the optimal social welfare in four parts. The first part, $q r_{HH}+(1-q) r_{LL}$,  is the revenue generated from assortative matching (i.e., if any demand agent arrives matched with a supply agent of the same type). The second part, $p(1-q)(r_{HH}-r_{LL})(1-\delta)$, represents revenue adjustments due to demand and supply mismatches: if $p>q$, this term is positive, indicating that more revenue can be generated by matching demand with additional H-type supply agents; whereas if $p<q$, this term is negative and indicates a revenue loss since extra H-type demand agents can only be matched with L-type supply agents. We note that this term vanishes when $p=q$. The third part, $-p(1-q)r\frac{1-\delta}{1-\delta^{k^{ce}+1}}$ if $p\neq q$ or $-\frac{p(1-p)r}{k^{ce}+1}$ if $p=q$, accounts for revenue loss due to cross-matching to avoid waiting costs. The final part, $k^{ce} h$, represents the average waiting cost incurred by unmatched supply agents. As the unit waiting cost increases, the optimal threshold decreases. While the average waiting cost may rise or fall, reducing the optimal threshold increases cross-matching and results in greater revenue loss. This additional revenue loss outweighs any potential gains from lower waiting costs. Consequently, overall social welfare decreases. In Figure \ref{figure:centralized sw}, we display an example illustrating how the optimal social welfare changes as the unit waiting cost increases, where $p=0.5$, $q\in\{0.1, 0.3, 0.5, 0.7, 0.9\}$, $\alpha=0.5$, and $\mathbf{r}=(800, 50, 50, 0)$. As plotted, the optimal social welfare decreases continuously with respect to $h$. 
	
 \begin{figure}[pt] 
 \centering
 \includegraphics[width=0.5\textwidth]{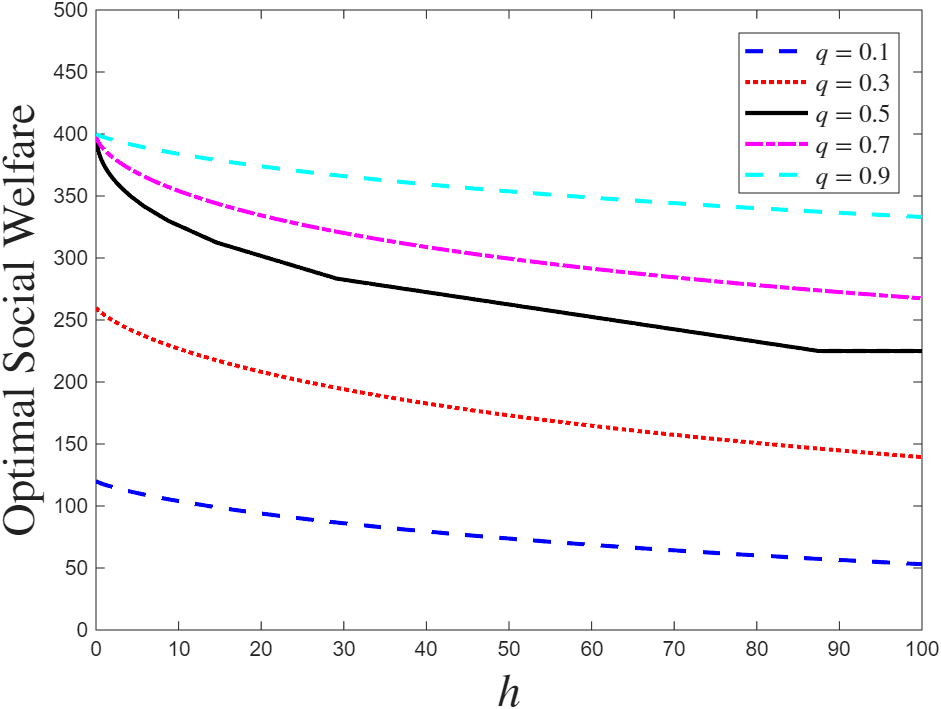}
 \caption{Impact of $h$ on the optimal social welfare in centralized system when $p=0.5$}\label{figure:centralized sw}
 \end{figure}

\section{Decentralized Matching}\label{sec:decentralized model}
In this section, we consider a decentralized setting in which matching decisions are made by self-interested agents. The arrival processes of the supply and demand agents are the same as in the centralized matching setting. We first introduce the decentralized matching setting. Then, we characterize the equilibrium behaviors of the supply and demand agents.

\subsection{Decentralized Matching Model}\label{subsec: decentralized model}

In the decentralized matching setting, agents on the platform can freely declare their matching preferences and willfully accept or reject a matching invitation from the other side. We analyze the matching process over an infinite time horizon and assume that the arrival processes of supply and demand agents under the decentralized model are identical to those under the centralized model.  In each period, a pair of supply and demand agents arrives on the platform, and the type of each agent is revealed upon arrival. We assume that the types of agents on both sides are publicly observable by all participants, including the platform and all agents. 

Consider an online freelancer platform (e.g., Upwork or Fiverr) that matches freelancers (the supply side) with projects posted by clients (the demand side). Agents arrive sequentially and randomly over time: in each period, new projects may be posted, and new freelancers may become available. Quality is publicly observable on both sides. Each project can be labeled as high or low based on a posted attribute, such as task complexity (e.g., a complex software-development project versus a routine data-entry task), and each freelancer can be labeled as high or low based on a publicly visible skill tier (e.g., with or without skill certifications). Demand is relatively impatient: a project is short-lived; if the client does not fill it quickly (or fills it off-platform), the posting is withdrawn. Freelancers are long-lived and can remain available while awaiting opportunities. Matching is decentralized: freelancers choose which projects to apply for, clients select among applicants (and may reject all low-type applications), and a match forms only upon mutual acceptance. In addition, the matching payoffs exhibit supermodularity. Equivalently, task complexity and freelancer skill are complementary; thus, assortative matches (high-skilled freelancers with complex tasks and low-skilled freelancers with routine tasks) are more efficient and generate higher match payoffs than cross-tier matches. 
Furthermore, we assume that agents are matched according to the first-come-first-served (FCFS) discipline within each type. The platform labels each agent within a type based on their entry time into the waiting line. Examples of the FCFS discipline within groups can be found in the matching cases of organ transplants, such as hearts and intestines \citep{bloch2017dynamic, baccara2020optimal}.
  
To be consistent with the notation of the centralized model, we use $r_{ij}$ to represent the total payoff of a match between a type $i$ supply agent and type $j$ demand agent. The payoff is split between the supply agent and the demand agent such that the supply agent receives a proportion $\alpha$ and the demand agent receives $(1-\alpha)$, where $0\le \alpha\le 1$. When $\alpha=1/2$, it means that the total payoff of a matching is equally shared by the two agents. In addition, we assume that $\alpha$ is common knowledge and remains constant over the entire planning horizon. Each unmatched supply agent will wait in the system and incur a waiting cost of $h\ge 0$ per period. In contrast, each unmatched demand agent will leave the system and receive a utility of 0 (without loss of generality). For example, on a freelancer marketplace, a client who does not receive an acceptable proposal within a short time window may withdraw the posting and hire elsewhere (e.g., on another platform or through their own network) without ever forming a match on the platform. Furthermore, we assume that if a supply agent is indifferent between waiting in line and matching immediately, he will choose to match immediately. Similarly, if a demand agent is indifferent between matching and not matching, she will also choose to match.

The sequence of events is as follows. At the beginning of each period, a pair of agents arrives on the platform. After observing the system state, i.e., the number of each type of agent being present, each agent decides whether to match with an H-type or L-type agent on the other side, or not to match with anyone at all. A match will only be formed when both the supply and demand agents agree to match with each other. In this case, the demand agent will be matched with the first willing supply agent of the chosen type in the waiting line. Otherwise, no match will be formed. At the end of each period, unmatched supply agents are carried over to the next period and incur waiting costs, while the unmatched demand agent departs from the market.

We now formally describe the dynamic game between the supply and demand agents in the decentralized matching process.  We use the concept of the Markov perfect equilibrium, which, in our case, is the set of strategies of the supply and demand agents at a Nash equilibrium that describes their matching decisions. For simplicity, we focus on pure strategy equilibrium. In particular, given the system state $\mathbf{s}_{t}=(x_{H}, x_{L}, y_{H}, y_{L})$, let $\mathbf{a}_i(\mathbf{s}_{t})=(a^1_i(\mathbf{s}_{t}),a^2_i(\mathbf{s}_{t}),a^3_i(\mathbf{s}_{t} ), \dots)$, where $a^n_i(\mathbf{s}_{t})\in\{H, L, \phi\}$ denotes the decision made by the $n$th $i$-type supply agent in the waiting line for $i\in\{H, L\}$. With a slight abuse of notation, $H$ and $L$ represent the decisions to request a match with an H-type or L-type demand agent, respectively, while $\phi$ represents the decision not to match. Let $\mathbf{a}^{-n}_i(\mathbf{s}_{t})=(a^1_i(\mathbf{s}_{t}),\dots,a^{n-1}_i(\mathbf{s}_{t}), a^{n+1}_i(\mathbf{s}_{t} ), \dots)$. Similarly, let $\mathbf{b}(\mathbf{s}_{t})=(b_H(\mathbf{s}_{t}),b_L(\mathbf{s}_{t}))$, where $b_j(\mathbf{s}_{t})\in\{H, L, \phi\}$ denotes the decision made by the the $j$-type demand agent for $j\in\{H,L\}$. A match between an $i$-type supply agent and a $j$-type demand agent will be formed only if there exists a $k\leq x_i$ such that $a^k_i(\mathbf{s}_{t})=j$ and $b_j(\mathbf{s}_{t})=i$. Moreover, the match will be formed between the $k^{*}$th $i$-type supply agent in the waiting line and the only $j$-type demand agent, where $k^*=\min\{k:a^k_i(\mathbf{s}_{t})=j\}$.

Let $v^n_{i}$ be the best expected total payoff of the $n$th $i$-type supply agent in the waiting line, which depends on the system state $\mathbf{s}_{t}$ and the matching decisions of all other agents.  This payoff consists of three parts: the waiting cost during the periods before he gets matched, a portion of the matching payoff in the period when he gets matched, and zero afterward. 
Similarly, let $w_{j}$ be the best payoff of the $j$-type demand agent. Note that she will leave the system and get the payoff of 0 if unmatched. Therefore,  $w_{j}=(1-\alpha) r_{Hj}$ as long as there exists an H-type supply agent willing to be matched with her. If no H-type supply agents are willing to match with her but an L-type supply agent is willing to do so, we have $w_{j}=(1-\alpha) r_{Lj}$; otherwise, $w_{j}=0$. Let $\sigma_{ij}$ be the indicator denoting whether a match is formed between an $i$-type supply agent and a $j$-type demand agent.
Note that $\sigma_{ij}=1$ if and only if the demand agent's type is $j$ and this demand agent and an $i$-type supply agent both choose to match with each other. Thus, $\sigma_{ij}$ is a function of the decision rules of the supply and demand agents $a^n_i(\mathbf{s}_{t})$ and $b_j(\mathbf{s}_{t})$.  

We now define the pure strategy Markov perfect equilibrium.

\begin{definition}[Markov Perfect Equilibrium]
    We say that the decision rules of the supply and demand agents $a^n_i(\mathbf{s}_{t})$ and $b_j(\mathbf{s}_{t})$ with $i,j=H, L$ and $n=1,2,\dots$ form a Markov perfect equilibrium if the following conditions are satisfied:\\
    \textit{(i)} there exists payoff functions $v^n_{i}$ and $w_{j}$ such that   $a^n_i(\mathbf{s}_{t})$ and $b_j(\mathbf{s}_{t})$ are the maximizers of the following two optimization problems, respectively:
    \begin{align}\label{eq:bellman1}
    v^n_{i}(\mathbf{s}_{t},&\mathbf{a}^{-n}_i(\mathbf{s}_{t}),\mathbf{a}_{-i}(\mathbf{s}_{t}),\mathbf{b}(\mathbf{s}_{t}))=\max\limits_{a^n_i\in\{H,L,\phi\}}\bigg\{\bold{1}_{\{\sigma_{i,a^{n}_i}=1\}} \bold{1}_{\{n=\min\{k:a^k_i(\mathbf{s}_{t})=a^n_i\}\}} \alpha r_{i,a^n_i} \nonumber\\ 
    &+(1-\bold{1}_{\{\sigma_{i,a^{n}_i}=1\}}\bold{1}_{\{n=\min\{k:a^k_i(\mathbf{s}_{t})=a^n_i\}\}})[-h+Ev^{n_{t+1}}_{i}(\mathbf{s}_{t+1},\mathbf{a}^{-n_{t+1}}_i(\mathbf{s}_{t+1}),\mathbf{a}_{-i}(\mathbf{s}_{t+1}),\mathbf{b}(\mathbf{s}_{t+1}))]\bigg\}, \
\end{align}
and \begin{equation}\label{eq:bellman2}
    w_{j}(\mathbf{s}_{t},\mathbf{a}_{H}(\mathbf{s}_{t}),\mathbf{a}_{L}(\mathbf{s}_{t}), b_{-j}(\mathbf{s}_{t}))=\max\limits_{b_j\in\{H,L,\phi\}}\bigg\{(1-\alpha) r_{b_j,j}\bold{1}_{\{\sigma_{b_j,j}=1\}}, 0\bigg\};
\end{equation}
    \textit{(ii)} $\mathbf{s}_{t+1}$ and $n_{t+1}$ in (\ref{eq:bellman1}) are updated from $\mathbf{s}_{t}$ and $n$ according to $a^n_i(\mathbf{s}_{t})$ and $b_j(\mathbf{s}_{t})$.
\end{definition}

Equations (\ref{eq:bellman1}) and (\ref{eq:bellman2}) describe the decision rules for the supply and demand agents, respectively. Specifically, for an $i$-type supply agent, if he is matched with a $j$-type demand agent, he receives a payoff of $\alpha r_{ij}$ and leaves the system. If he is not matched during this period, he incurs a holding cost of $h$ and transits to the next period with a new state $\mathbf{s}_{t+1}$ and a new position $n_{t+1}$ in the waiting line. Given the decisions of all other agents, he will choose the option that maximizes his expected total payoff. For a $j$-type demand agent, since she will not wait, she will select the option that maximizes her current payoff. Moreover, the updates to the state and position must be consistent with the decision rules for both supply and demand agents. 

\subsection{Equilibrium in Decentralized Model}
H-type supply is the most desirable partner for any demand agent, and H-type demand is the most desirable partner for any supply agent; thus,  (H, H) matches are given priority. Accordingly, when an H-type demand agent arrives and at least one H-type supply agent is waiting, the arriving demand agent is matched with the first H-type supply agent in the FCFS queue. By contrast, when an L-type demand agent arrives, supply agents may be reluctant to match immediately: they trade off the immediate payoff from matching the L-type demand agent against the option value of waiting for a future H-type demand arrival, net of the holding costs incurred while remaining unmatched.



First, consider an H-type supply agent. The incremental payoff from matching with an H-type demand agent rather than an L-type demand agent is the constant $r_{HH}-r_{HL}$. In contrast, the expected waiting cost is position-dependent and increases with the number of H-type supply agents ahead of the individual in the queue. Under the FCFS discipline, any arriving H-type demand agent matches with the first waiting H-type supply agent, if one exists. Consequently, when the H-type supply queue becomes sufficiently long, a later-arriving H-type supply agent may prefer to match immediately with an L-type demand agent rather than incur the waiting cost of holding out for an H-type demand agent. This implies that, in equilibrium, the H-type supply queue length must be bounded, as formalized in the following lemma.

\begin{lemma}\label{lem: dec maxium queue length}
	In the decentralized setting, the queue length of H-type supply agents is bounded by $k^{de}$, given by 
 \begin{equation}\label{eq:threshold in decentralized setting}
     {k}^{de}=\left\lfloor \frac{q \alpha (r_{HH}-r_{HL})}{h}\right\rfloor.
 \end{equation}
\end{lemma}

\noindent Note that the time until the next $H$-type demand arrival is geometrically distributed with parameter $q$, and hence has a mean $1/q$. Let $x_H$ denote the current queue length of $H$-type supply agents, and consider an $H$-type supply agent in position $k \in \{1,\dots,x_H\}$ in the $H$-type supply queue. If this agent waits to match with an $H$-type demand agent, he must wait for $k$ $H$-type demand arrivals to be served; thus, his expected waiting time is $k/q$, implying an expected waiting cost of $hk/q$. The expected incremental benefit from waiting—relative to matching immediately with an $L$-type demand agent—is $\alpha\,(r_{HH}-r_{HL})$. Therefore, the agent prefers to continue waiting if and only if the expected benefit exceeds the expected cost, i.e.,
\begin{equation*}
\frac{hk}{q}<\alpha\,(r_{HH}-r_{HL}).
\end{equation*}
Moreover, under the FCFS discipline, the agent's position $k$ weakly decreases over time as earlier agents are matched, which reduces their expected remaining waiting time in subsequent periods, while the incremental benefit $\alpha(r_{HH}-r_{HL})$ remains unchanged. Hence, once an $H$-type supply agent chooses to wait, they will not later deviate to match with an $L$-type demand agent. It follows that any $H$-type supply agent in positions $k=1,\ldots,k^{de}$ is willing to match only with an $H$-type demand agent. Conversely, for an $H$-type supply agent at position $k^{de}+1$, the expected waiting cost exceeds the expected incremental benefit; thus,  we must have
\begin{equation*}
\frac{h(k^{de}+1)}{q}>\alpha\,(r_{HH}-r_{HL}).
\end{equation*}

Second, the behavior of an L-type supply agent is more nuanced. When an H-type demand agent arrives, an L-type supply agent would request a match with her, as it would be less profitable for the supply agent to wait for a future demand agent of either type. 
When an L-type demand agent arrives, an L-type supply agent must decide whether to match immediately or continue waiting in hopes of matching with a future H-type demand agent. The incremental payoff from waiting for such an H-type demand match, relative to matching now with the L-type demand agent, is $r_{LH}-r_{LL}$. However, the associated expected waiting cost is not straightforward. To be matched with an H-type demand agent, an L-type supply agent must wait not only for the L-type supply agents ahead of him, but also for the clearance of the H-type supply queue, since any arriving H-type demand agent is prioritized to match with the first available H-type supply agent. Moreover, the H-type supply queue may expand over time due to continued H-type supply arrivals, which complicates the calculation of the L-type supply agent’s expected waiting time.

Nevertheless, a similar logic applies as before. When the number of L-type supply agents ahead of a given L-type supply agent exceeds a threshold $k_L$, the agent is more likely to stop holding out for a future H-type demand arrival and instead match with the current L-type demand agent. This threshold $k_L$ depends on the state of the H-type supply queue: as the number of H-type supply agents in the system increases, the expected delay until an L-type supply agent can access an H-type demand match rises, which lowers $k_L$ and induces L-type supply agents to be more likely to match with L-type demand.

\begin{lemma}\label{lemma:threshold property}
     The threshold $ {k}_{L}(x_{H})$ has the following properties. 
     \begin{enumerate}
         \item[(i)] $ {k}_{L}(x_{H})$ decreases in $x_{H}$; 
         \item[(ii)] $ {k}_L(x_{H})\leq  {k}^{de}-x_{H}$; 
         \item[(iii)] When $p\geq q$, ${k}_L(x_{H})=0$.
     \end{enumerate}
\end{lemma}

On the demand side, consider an arriving demand agent. Because demand agents are short-lived, she accepts a match whenever at least one supply agent in the system is willing to match with her. Moreover, whenever both types are available, she strictly prefers to match with an H-type supply agent rather than an L-type supply agent. 

We summarize the decision rules of the agents in equilibrium in the following proposition. 
\begin{proposition}\label{prop:matching action in equilibrium} There exists a Markov perfect equilibrium in which \\
(i) \quad  $a^n_i(x_{H}, x_{L}, 1, 0)=H$, for $i=H,L$ and $n=1,2,\dots$; \\
(ii) \begin{align*}
    a^n_H(x_{H}, x_{L}, 0, 1)=\begin{cases}
        H  &\mbox{ if $n\leq k^{de}$,}\\
        L  &\mbox{ otherwise,}
    \end{cases} \quad \mbox{and} \quad a^n_L(x_{H}, x_{L}, 0, 1)=\begin{cases}
        H  &\mbox{ if $n\leq  {k}_L(x_{H})$,}\\
        L  &\mbox{ otherwise;}
    \end{cases}
\end{align*}
(iii) \begin{align*}
    b_H(x_{H}, x_{L}, 1, 0)=\begin{cases}
        H  &\mbox{ if $x_H>0$,}\\
        L  &\mbox{ otherwise, }
    \end{cases} \quad \mbox{and} \quad b_L(x_{H}, x_{L}, 0, 1)=\begin{cases}
        H  &\mbox{ if $x_H> k^{de}$,}\\
        L  &\mbox{ otherwise. }
    \end{cases}
\end{align*}
Additionally, this equilibrium is welfare-maximizing, achieving the highest long-run social welfare among all possible Markov perfect equilibria.
\end{proposition}

Proposition \ref{prop:matching action in equilibrium} establishes the existence of a welfare-maximizing Markov perfect equilibrium with the following matching behavior. When an H-type demand agent arrives, all supply agents request to match with her; the demand agent matches with an H-type supply agent if one is available, and otherwise matches with an L-type supply agent. When an L-type demand agent arrives, the first $k^{de}$ H-type supply agents in the waiting line choose to wait for a future H-type demand agent, while any remaining H-type supply agents (if any) offer to match with the arriving L-type demand agent. Likewise, among L-type supply agents, the first $k_L(x_H)$ agents in the waiting line wait for a future H-type demand agent, while all others offer to match with the arriving L-type demand agent. On the demand side, an L-type demand agent requests to match with an H-type supply agent if the number of H-type supply agents exceeds $k^{de}$; otherwise, she offers to match with an L-type supply agent.
Given the above agents’ action strategies, the equilibrium matching outcomes can be summarized as follows.
\begin{theorem} \label{thm:stable matching of decentralized system} 
In the welfare-maximizing equilibrium, matches are formed as follows. 
	\begin{enumerate}
		\item (Greedy Matching) When an H-type demand agent arrives, a match is formed between her and the first H-type supply agent in the waiting queue, if one is available. If there are no H-type supply agents available, a match is formed with the first L-type supply agent in the queue.
		\item (Threshold Matching) When an L-type demand agent arrives, one of two possible cases occurs:
		\begin{description}
			\item[(a)] When the number of H-type supply agents exceeds the threshold, i.e., $x_{H}> k^{de}$, a match is formed with the $(k^{de}+1)$st H-type supply agent in the queue.
			\item[(b)] When the number of H-type supply agents is equal to or less than the threshold, i.e., $x_{H}\le k^{de}$, a match is formed with the $ ({k}_{L}(x_{H})+1)$st L-type supply agent in the queue, if such an agent exists. 
		\end{description}
	\end{enumerate}
\end{theorem}

Under the optimal centralized matching policy, an L-type demand agent is matched with an L-type supply agent whenever one is available (see Theorem \ref{thm:optimal mechanism of centralized system}). In contrast, in the decentralized equilibrium, an L-type demand agent may be matched with an H-type supply agent when the H-type supply queue is sufficiently large, even if an L-type supply agent is available.


The matching process characterized in Theorem \ref{thm:stable matching of decentralized system} differs fundamentally from decentralized matching under full backlog, as studied in \citet{baccara2020optimal}. In a full backlog environment, the welfare-maximizing equilibrium forms an (L, L) match whenever both L-type agents are present on the two sides. As a result, the system cannot simultaneously have an L-type supply agent and an L-type demand agent waiting.

Without loss of generality, we assume that the system is empty at time 0, i.e., no agents are waiting on the platform at the beginning of the planning horizon. Following the matching process in the welfare-maximizing equilibrium, the number of supply agents of both types in the waiting lines at the beginning of each period, before new agents arrive, forms a Markov chain. The transition matrix of this Markov chain can be derived based on the arrival and matching processes, which can be quite complex. However, we can show that the steady states of the Markov chain consist of only $k^{de}+1$ states, and the corresponding steady state distribution can be computed by solving the balance equations. 

\begin{proposition}[{\sc{Decentralized Steady States}}] \label{prop: steady states in equilibrium} In the welfare-maximizing equilibrium, the number of supply agents present at the beginning of each period before new agents arrive follows a Markov chain with the steady states $\{(0, k^{de}), (1, k^{de}-1),\dots, (k^{de}, 0)\}$ and a steady state distribution
\begin{equation*}
    \phi^{de}(i, {k}^{de}-i)=\begin{cases}
        \frac{(1-\delta)\delta^{{k}^{de}-i}}{1-\delta^{k^{de}+1}}, \quad \quad 0\le i\le  k^{de}, & \mbox{ if $p\neq q$,}\\
        \frac{1}{k^{de}+1},  & \mbox{ if $p= q$,}
    \end{cases}
\end{equation*}
for $i=0,\dots, k^{de}$, where state $(n,k^{de}-n)$ represents that $n$ H-type supply agents and $k^{de}-n$ L-type supply agents are waiting before new agents arrive.
\end{proposition}

Proposition \ref{prop: steady states in equilibrium} shows that in a steady state, there will be a total of $k^{de}$ supply agents. Therefore, as long as the initial number of supply agents is below $k^{de}$, it will stay bounded by $k^{de}$. 


We can derive the steady-state social welfare in the welfare-maximizing equilibrium.

\begin{proposition}[{\sc{Social Welfare in Equilibrium}}] \label{prop:decentralized sw}
In the decentralized setting, social welfare in the welfare-maximizing equilibrium is 
\begin{align}\label{eq:decentralized sw}
	W^{de}=\left\{
\begin{array}{ll}
 qr_{HH}+(1-q)r_{LL}+p(1-q)(r_{HH}-r_{LL})(1-\delta)-p(1-q)r\frac{1-\delta}{1-\delta^{k^{de}+1}}-k^{de}h,& \hbox{if $p\neq q$,}\vspace{2mm} \\ 
 pr_{HH}+(1-p)r_{LL}-\frac{p(1-p)r}{k^{de}+1}-k^{de}h,& \hbox{if $p=q$,}
 \end{array}	
	\right.
 \end{align}
 where $k^{de}$ is given by (\ref{eq:threshold in decentralized setting}).  Moreover, $W^{de}$ is non-monotonic in $h$.
\end{proposition}

Proposition \ref{prop:decentralized sw} shows that social welfare in the decentralized setting takes the same functional form as in the centralized setting, with the centralized threshold $k^{ce}$ replaced by the equilibrium threshold $k^{de}$.

%
However, unlike in the centralized setting, social welfare in the decentralized setting does not always decrease with respect to the unit waiting cost $h$. Figure \ref{figure:decentralized sw} illustrates how social welfare in the decentralized system varies as the unit waiting cost increases, with $(p,q)\in\{(0.4, 0.6), (0.5, 0.5), (0.6, 0.4)\}$, $\mathbf{r}=(800, 50, 50, 0)$, and $\alpha=0.2$. More numerical examples with different parameter combinations can be found in Appendix \ref{appendix: numberical results}. As shown in the figure, the welfare curve is piecewise decreasing, with upward jumps at discrete cutoff points. This implies that decentralized social welfare declines in $h$ as long as the equilibrium threshold remains fixed, but changes discontinuously when the threshold adjusts. Within each segment, a higher $h$ increases average waiting costs and therefore reduces welfare. At some cutoff point, the threshold decreases by one, sharply lowering average waiting costs and producing the upward jump in welfare. When $h$ is sufficiently large, no agent is willing to wait, and social welfare becomes constant.

\begin{figure}[htbp]
    \centering
\subfigure[$p=0.4$ and $q=0.6$\label{subfig:main_a}]{
\includegraphics[width=0.3\textwidth]      {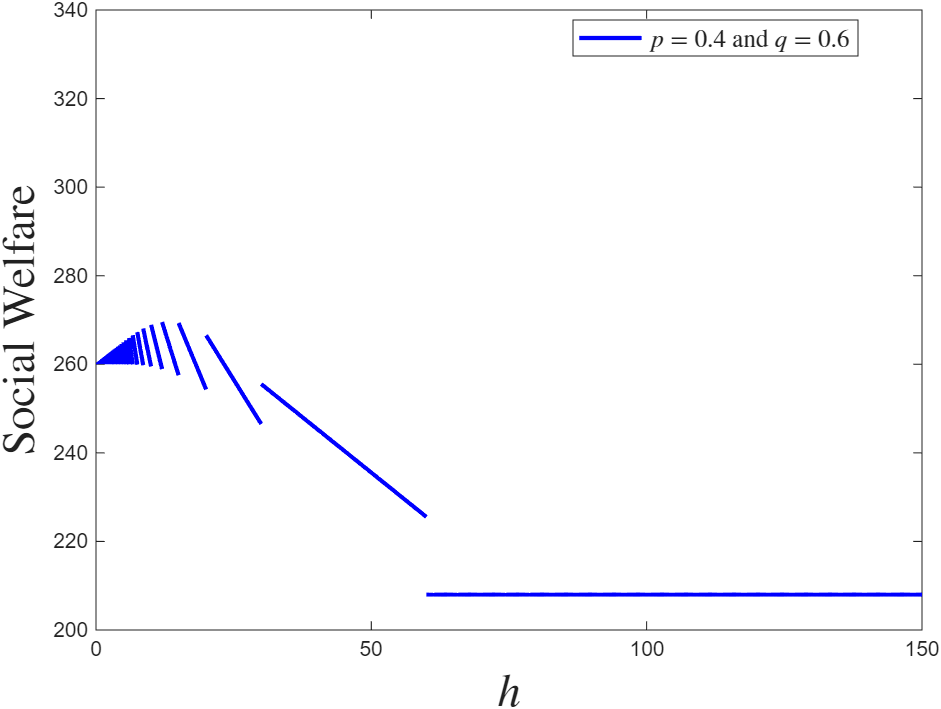}}
\hfill
\subfigure[$p=0.5$ and $q=0.5$\label{subfig:main_b}]{    \includegraphics[width=0.3\textwidth]{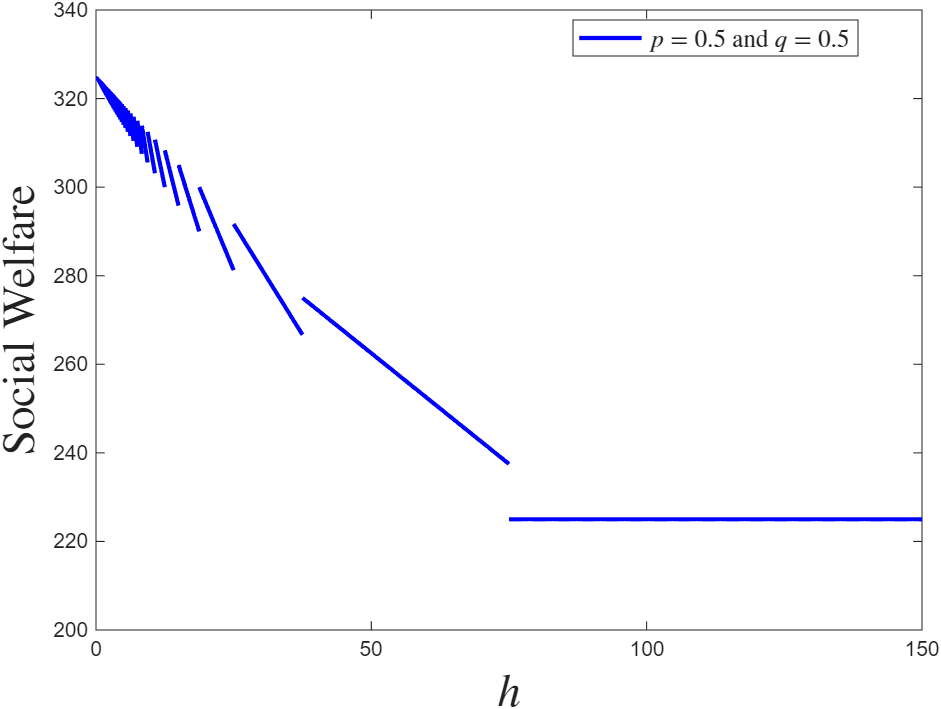}}
\hfill
\subfigure[$p=0.6$ and $q=0.4$\label{subfig:main_c}]{
\includegraphics[width=0.3\textwidth]{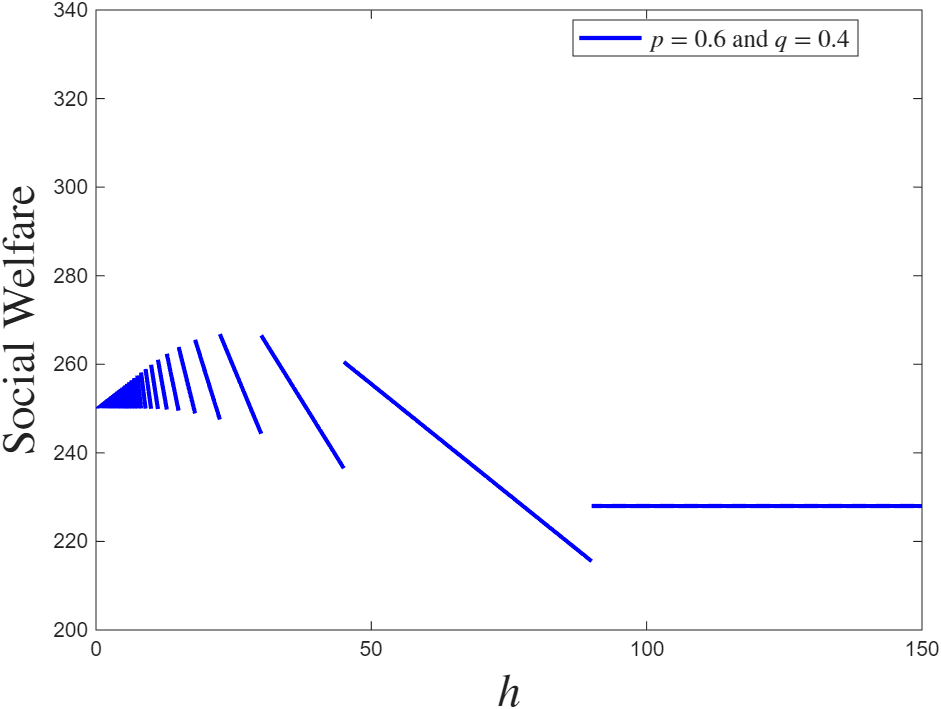}}
\caption{Impact of $h$ on social welfare in the decentralized system}\label{figure:decentralized sw}
 \end{figure}



\subsection{Comparison Between Centralized and Decentralized Settings}

In the centralized and decentralized settings, the matching processes are governed by the thresholds $k^{ce}$ and $k^{de}$, respectively. We will begin by comparing these two thresholds.

\begin{proposition}\label{prop:threshold comparison}
When $\alpha<\frac{hk^{ce}}{q(r_{HH}-r_{HL})}$, $k^{de}<k^{ce}$; when $\alpha\geq \frac{h(k^{ce}+1)}{q(r_{HH}-r_{HL})}$, $k^{de}>k^{ce}$; and when $\alpha\in\left[\frac{hk^{ce}}{q(r_{HH}-r_{HL})}, \frac{h(k^{ce}+1)}{q(r_{HH}-r_{HL})}\right)$, $k^{de}=k^{ce}$.
\end{proposition}

Note that the thresholds $k^{ce}$ and $k^{de}$ correspond to the maximum total number of supply agents waiting at the beginning of each period in the centralized and decentralized settings, respectively. Proposition \ref{prop:threshold comparison} implies that the maximum total number of supply agents waiting in the decentralized process might be either inefficiently longer or shorter than that under the optimal central policy. More specifically, in the decentralized setting, the threshold $k^{ce}$ depends on the payoff allocation proportion $\alpha$ assigned to supply agents. If $\alpha=0$, supply agents would avoid waiting since they would incur costs without receiving any matching payoff. As a result, the queue lengths of supply agents would be zero. Conversely, if $\alpha=1$, supply agents receive all matching payoffs, which provides a strong incentive for them to wait for better matches and may result in inefficiently longer waiting lines. For example,  when $\alpha=1$, $p=q=1/2$, and $h=1$, we have 
\begin{eqnarray*}
	k^{de}&=&\left\lfloor\frac{\alpha q(r_{HH}-r_{HL})}{h}\right\rfloor=\left\lfloor\frac{1}{2}(r_{HH}-r_{HL})\right\rfloor\\
	k^{ce}&=&\left\lfloor\frac{-h+\sqrt{h^2+4hp(1-p)r}}{2h}\right\rfloor=\left\lfloor\frac{-1+\sqrt{1+(r_{HH}-r_{HL})-(r_{LH}-r_{LL})}}{2}\right\rfloor.
\end{eqnarray*}
It is easy to verify that $\frac{1}{2}(r_{HH}-r_{HL})>\frac{-1+\sqrt{1+(r_{HH}-r_{HL})-(r_{LH}-r_{LL})}}{2}$, and thus $k^{de}\ge  k^{ce}$.

We also examine the impact of the unit waiting cost on the two thresholds in Figure \ref{fig:threshold}, where $p=q=0.5$ and $\mathbf{r}=(800, 50, 50, 0)$. We can see that when $\alpha$ is large, $k^{de}> k^{ce}$ for all values of $h$. When $\alpha$ is small, $k^{de}> k^{ce}$ when $h$ is low, and $k^{de}\leq  k^{ce}$ when $h$ is high. Moreover, we note that the threshold of the decentralized system decreases more rapidly than that of the centralized system as $h$ increases. This is because, as shown in equations (\ref{eq:centralized threshold}) and (\ref{eq:threshold in decentralized setting}), $k^{ce}$ decreases at the rate of $1/\sqrt{h}$, while $k^{de}$ decreases at the rate of $1/h$.

\begin{figure}[htpb]
	\centering
	\subfigure[$\alpha=0.5$]{
		\label{fig:threshold:a} 
		\includegraphics[width = 0.45\linewidth]{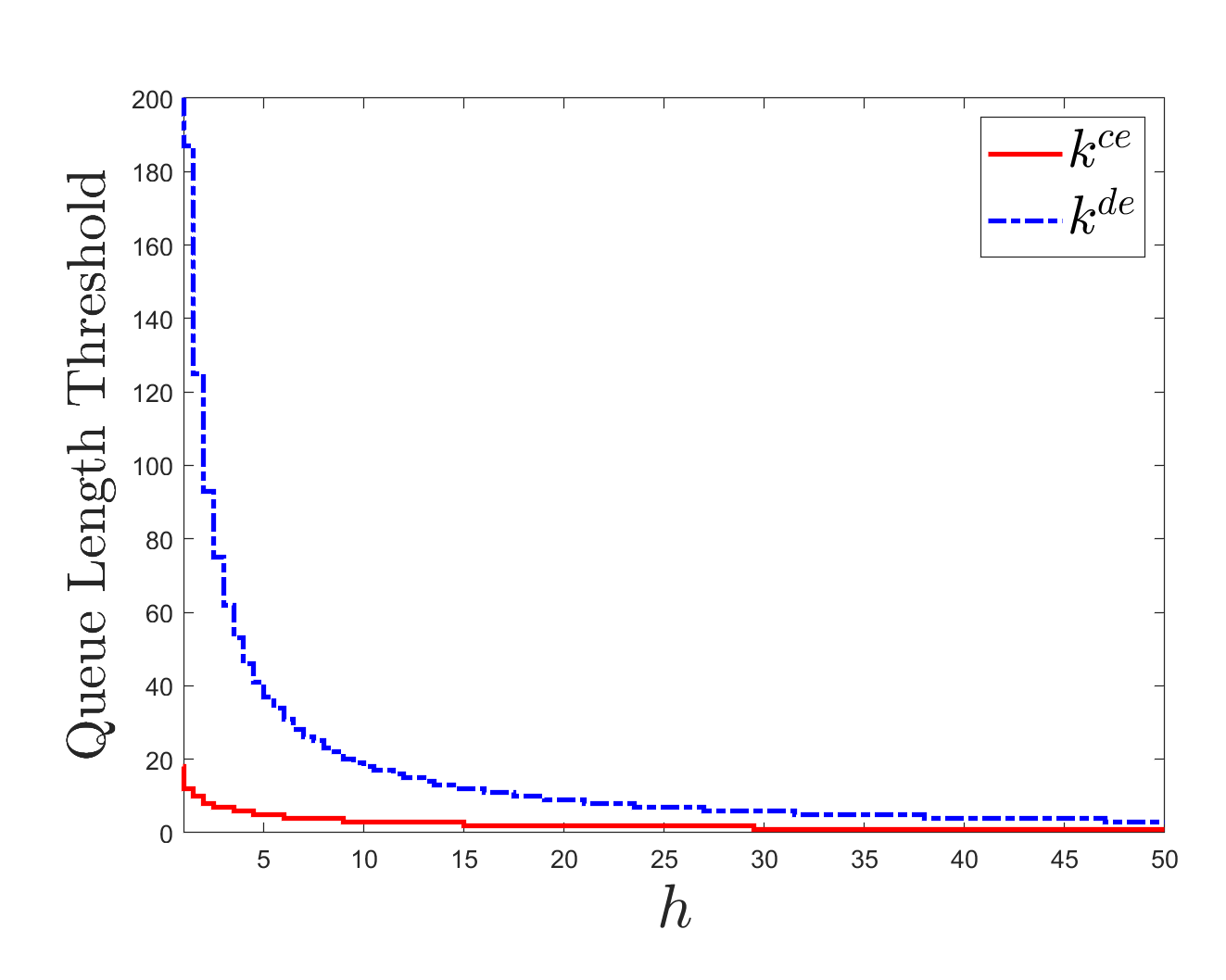}}
	\subfigure[$\alpha=0.1$]{
		\label{fig:threshold:b} 
		\includegraphics[width = 0.45\linewidth]{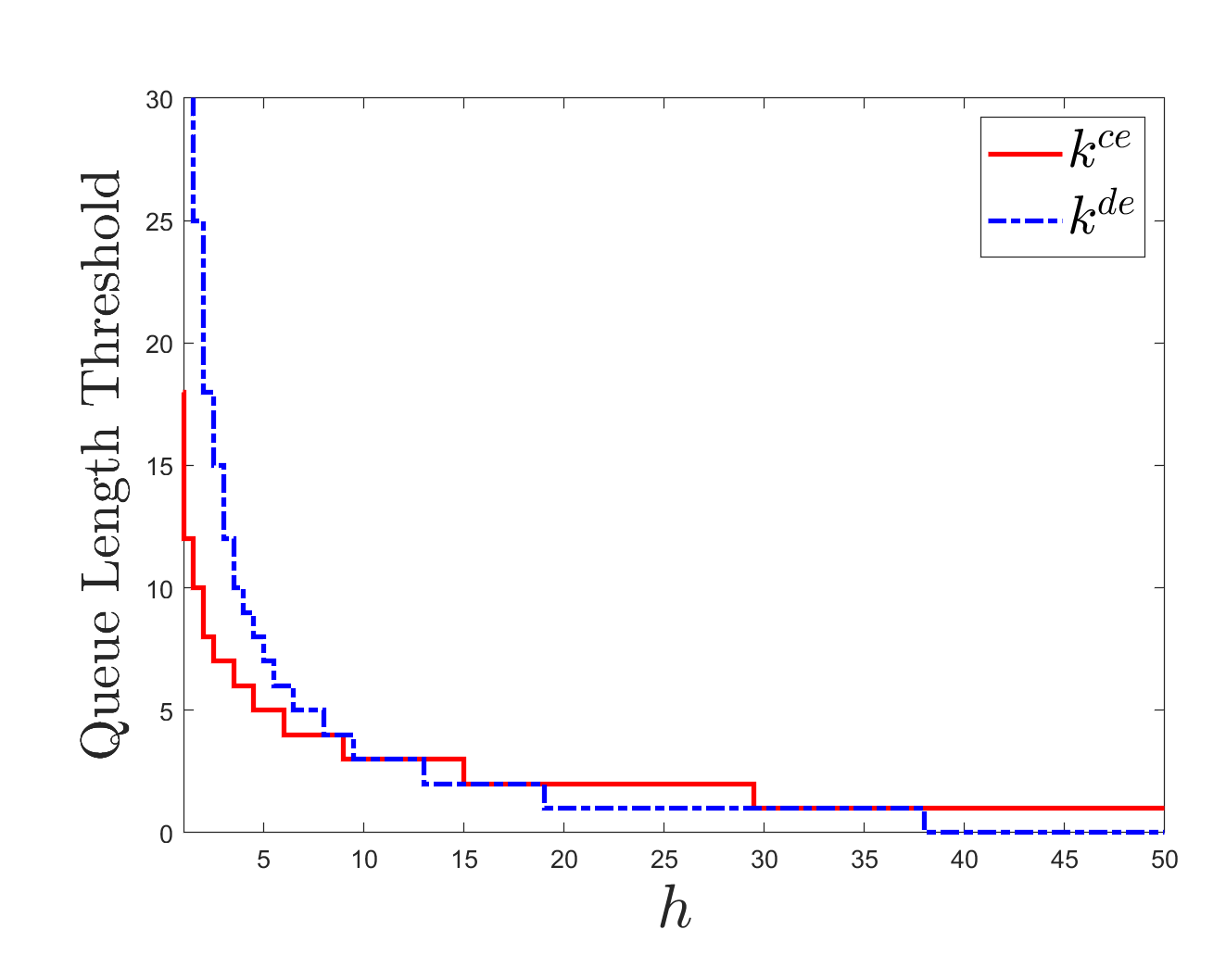}}
	\caption{Thresholds of centralized and decentralized systems}\label{fig:threshold}
\end{figure}

When $\alpha$ is set to be the value that makes $k^{de}=k^{ce}$, the two settings become identical, and the decentralized matching process is coordinated.

\begin{proposition}[{\sc{Social Welfare Coordination}}] \label{prop: coordination} When $\alpha\in\left[\frac{hk^{ce}}{q(r_{HH}-r_{HL})}, \frac{h(k^{ce}+1)}{q(r_{HH}-r_{HL})}\right)$, the decentralized matching process is coordinated and behaves similarly to the centralized system.
\end{proposition}

The above proposition shows that the decentralized matching process can be coordinated by adjusting the payoff allocation proportion $\alpha$. This coordination result contrasts sharply with the full backlog setting. \citet{baccara2020optimal} show that in systems with a full backlog,  the decentralized matching process {\it cannot} be coordinated under the FCFS discipline. When both sides can wait, changing $\alpha$ primarily reallocates waiting from one side to the other: increasing the supply share induces supply agents to wait longer and lengthens the supply queue, whereas increasing the demand share shifts waiting to the demand side and lengthens the demand queue. Consequently, regardless of how payoffs are split, the decentralized system exhibits a total queue that is inefficiently large relative to the centralized benchmark. In our model, by contrast, demand is short-lived and exits if unmatched, so only the supply side accumulates a queue. The queue length is therefore directly governed by $\alpha$: a higher $\alpha$ increases the supply’s incentive to wait and lengthens the queue, while a lower $\alpha$ reduces waiting incentives and shortens the queue. Because there is no offsetting demand-side backlog, the platform can tune the decentralized equilibrium queue length to coincide with the centralized optimum. In particular, by selecting $\alpha$ within the appropriate range, perfect coordination is achieved.

Next, we compare social welfare in the two settings. Clearly, social welfare in the decentralized system is no greater than that in the centralized system. However, as discussed above, the social welfare gap can be zero when the decentralized matching process is coordinated, even with a small unit waiting cost. In such cases, $W^{de}=W^{ce}$ and $k^{de}=k^{ce}>0$, i.e., the maximum queue length is positive. 

In Figure \ref{figure:sw gap}, we plot the social welfare gap as a function of the unit waiting cost $h$. In this example, we set $p=q=0.5$, $\mathbf{r}=(800, 50, 50, 0)$, and $\alpha=0.2$. As we can see, when 
$h$ is between 40 and 60, the social welfare gap is zero. This differs from the results obtained in \citet{baccara2020optimal}: in systems with full backlog, before $h$ becomes extremely large, at which point no agents choose to wait in either the centralized or the decentralized system, social welfare in the centralized system is always {\it strictly} greater than that in the decentralized system.

Moreover, we observe that the social welfare gap consists of increasing, decreasing, and constant segments in $h$. When the threshold in the decentralized system is larger than that in the centralized system, the social welfare gap has increasing segments in $h$. When thresholds are equal, the social welfare gap remains constant within a segment. When the threshold in the decentralized system is smaller than that in the centralized system, we observe decreasing segments. This observation again contrasts with the results in \citet{baccara2020optimal}. For systems with a full backlog, the social welfare gap consists solely of increasing segments, as the threshold in the decentralized system is always greater than that in the centralized system.

\begin{figure}[pt] 
 \centering
 \includegraphics[width=0.5\textwidth]{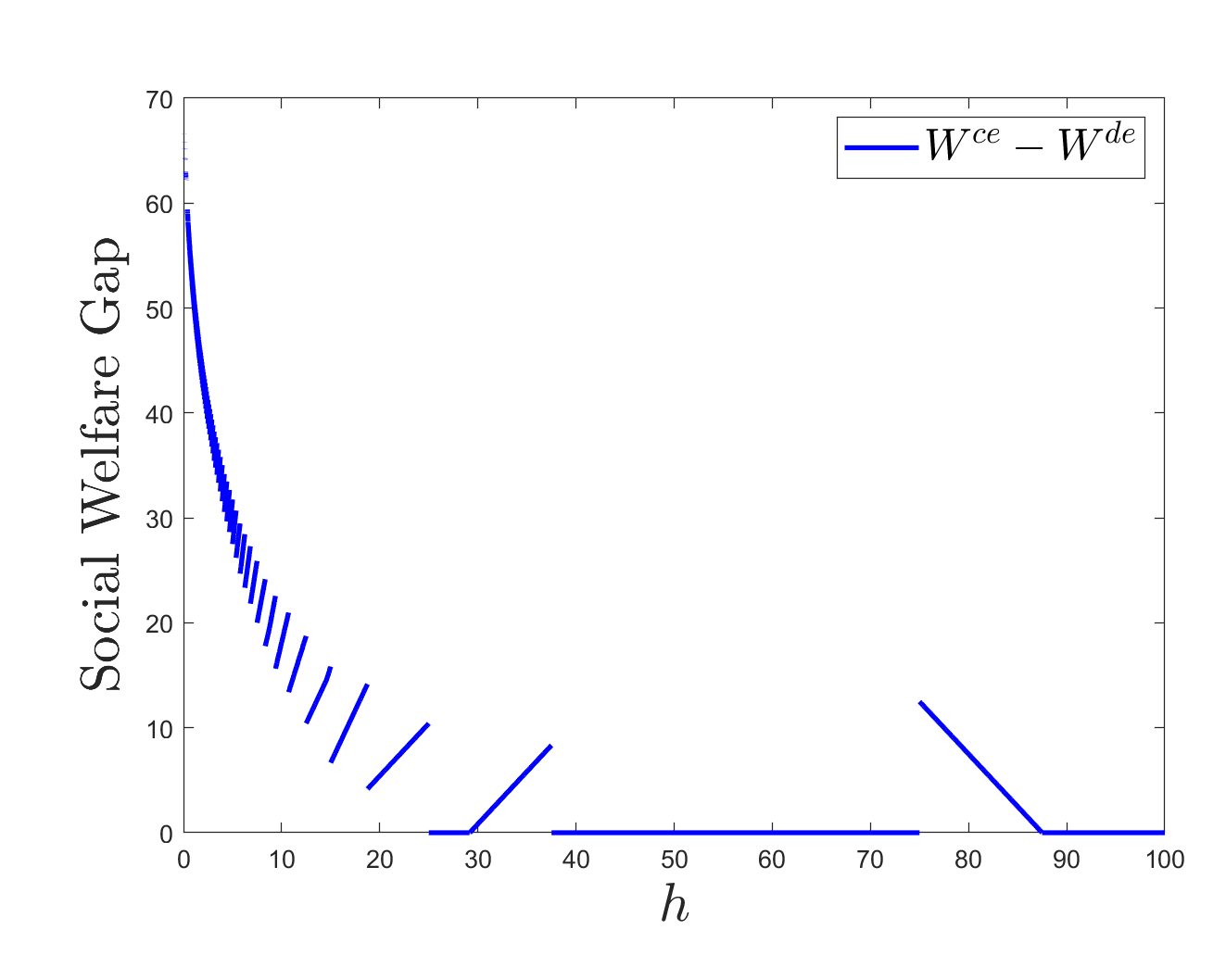}
 \caption{Impact of $h$ on the social welfare gap}\label{figure:sw gap}
 \end{figure}

\section{Impacts of Patience}\label{sec:impact}

In this section, we study the impacts of patience by comparing social welfare across three systems: the full-backlog system, the one-sided backlog system, and the no-backlog system. We refer to the full-backlog system as the FB system, the one-sided backlog system as the OB system, and the no-backlog system as the NB system for short. For tractability, we restrict attention to the symmetric-arrival case $p=q$. 

The full backlog system is analyzed in \cite{baccara2020optimal}, while the one-sided backlog system is discussed in Sections \ref{sec:centralized model} and \ref{sec:decentralized model}. In the no-backlog system, agents are impatient on both sides and cannot be carried over to the next period. Thus, the dynamic matching problem degenerates to a single-period problem. In the following sections, we compare under centralized and decentralized settings.

\subsection{Comparison in Centralized Setting}

In the centralized setting, social welfare in the FB and OB systems is expressed as follows: 
\begin{align*}
    W^{ce}_{FB}&=pr_{HH}+(1-p)r_{LL}-\frac{p(1-p)r}{2k^{ce}_{FB}+1}-\frac{2k^{ce}_{FB}(k^{ce}_{FB}+1)}{2k^{ce}_{FB}+1}h,\\
   W^{ce}_{OB}&=pr_{HH}+(1-p)r_{LL}-\frac{p(1-p)r}{k^{ce}+1}-k^{ce}h,
\end{align*}
respectively, where $k^{ce}_{FB}=\left\lfloor\sqrt{\frac{p(1-p)r}{2h}}\right\rfloor$ and $k^{ce}=\left\lfloor\frac{-h+\sqrt{h^2+4hp(1-p)r}}{2h}\right\rfloor$.

In the NB system, it is always optimal to match the arriving supply agent with the arriving demand agent. Otherwise, both will leave the system without generating any revenue. Thus, the optimal social welfare can be expressed as 
\begin{equation*}
    W^{ce}_{NB}=pr_{HH}+(1-p)r_{LL}-p(1-p)r.
\end{equation*}

We now compare social welfare across the three systems within the centralized setting. 
\begin{proposition}\label{prop:comparison in centralized setting}
    $W^{ce}_{FB}\geq W^{ce}_{OB}\geq W^{ce}_{NB}$ and $W^{ce}_{FB}-W^{ce}_{OB}\leq W^{ce}_{OB}-W^{ce}_{NB}$.
\end{proposition}

 \begin{figure}[pt]
 \centering
 \includegraphics[width=0.5\textwidth]{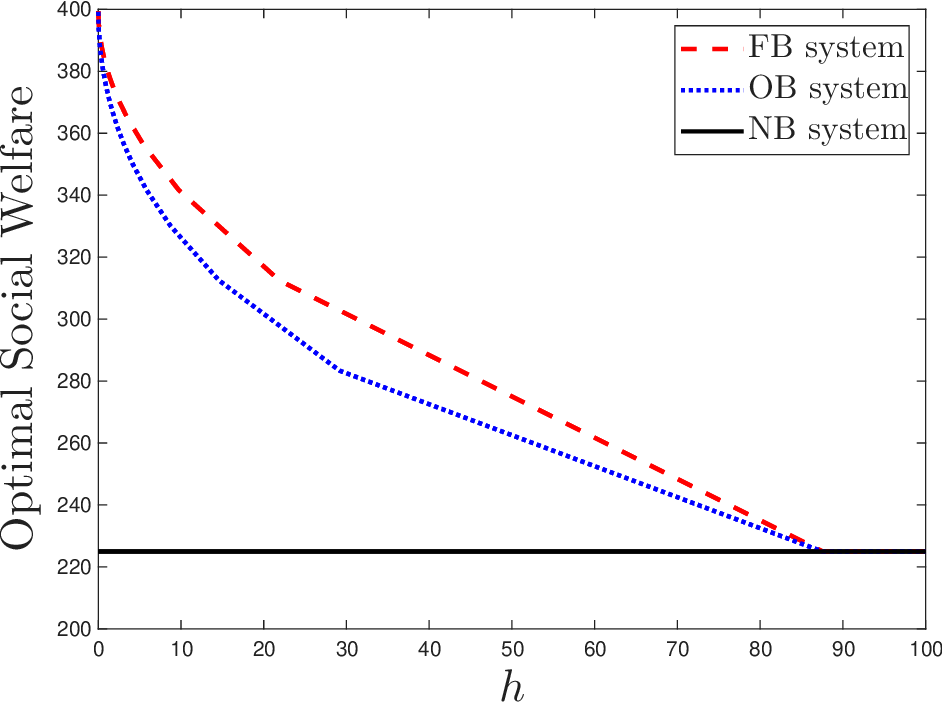}
 \caption{Comparison of the optimal social welfare ($p=0.5$ and $\mathbf{r}=(800, 50, 50, 0)$)} \label{figure:sw comparison}
 \end{figure}

As illustrated in Figure \ref{figure:sw comparison}, Proposition \ref{prop:comparison in centralized setting} shows that, under centralized control, social welfare is highest when both sides are patient and lowest when neither side can wait, which is consistent with intuition. Allowing agents to remain in the system expands the planner’s feasible set and increases the ability to defer matches in order to realize higher-value pairings, thereby improving welfare under the optimal policy. More importantly, Proposition \ref{prop:comparison in centralized setting} shows that the social welfare gain from introducing patience on only one side is at least as large as the incremental gain from making the second side patient as well, indicating diminishing marginal returns to patience. For a platform planner, this implies that most of the welfare benefits of dynamic matching can be achieved by securing patience on a single side, a feature that is already present in many platforms (see, e.g., Table \ref{Table.1}). Accordingly, extending patience from one side to both sides does not double welfare; the marginal contribution of making the second side patient is strictly smaller (at most equal) than that of the first.

This result points to an asymmetric market-design implication: rather than attempting to raise the patience on both sides symmetrically, a planner can leverage inherent asymmetries by ensuring that at least one side of the market is deep and patient. This implication helps explain the high efficiency of platforms and systems that naturally exhibit such a structure. For example, in many local ride-hailing markets, a centralized platform such as Uber relies on a relatively patient pool of drivers to absorb the demand of time-sensitive riders, thereby capturing most of the social welfare gains highlighted by our model. Similarly, in organ transplantation, much of the system’s value is generated by a long-lived recipient list that can be matched promptly to an extremely ``impatient'' perishable organ. The implication for central planners is clear: if the patience imbalance of both sides does not arise endogenously, the primary objective of the planner should be to cultivate patience on at least one side of the market. Doing so is the key to unlocking the vast majority of the system's potential social welfare. In essence, our model shows that, for centralized systems, a little patience on one side can yield substantial improvements in social welfare.


\subsection{Comparison in Decentralized Setting}

In the decentralized setting, social welfare in both the FB (full backlog) system and the OB (one-sided backlog) system is given by  
\begin{align*}
    W^{de}_{FB}&=pr_{HH}+(1-p)r_{LL}-\frac{p(1-p)r}{2k^{de}+1}-\frac{2k^{de}(k^{de}+1)}{2k^{de}+1}h,\\
    W^{de}_{OB}&=pr_{HH}+(1-p)r_{LL}-\frac{p(1-p)r}{k^{de}+1}-k^{de}h,
\end{align*}
where $k^{de}=\left\lfloor \frac{p \alpha (r_{HH}-r_{HL})}{h}\right\rfloor$. 

In the NB system, agents on both sides are impatient. Hence, in the decentralized setting, each agent accepts a match whenever a counterpart on the other side is willing to match. It follows that a welfare-maximizing equilibrium exists in which all feasible matches are formed, such that the decentralized social welfare coincides with that in the centralized benchmark. In particular, social welfare is given by:
\begin{equation*}
    W_{NB}^{de}=W^{ce}_{NB}=pr_{HH}+(1-p)r_{LL}-p(1-p)r.
\end{equation*}

Unlike in the centralized setting, the value of patience in the decentralized setting is shaped by agents’ self-interested incentives. As a result, greater patience can either increase social welfare, i.e., by enabling more efficient matches, or reduce it by distorting agents’ matching decisions.
\begin{proposition}\label{prop:comparison in decentralized setting}
There exist two thresholds $\alpha_1$ and $\alpha_2$ with $\alpha_1<\alpha_2$, such that \\
(1) if $\alpha\leq \alpha_1$, $W^{de}_{FB}\geq W^{de}_{OB}\geq W^{de}_{NB}$;\\
(2) if $\alpha\geq \alpha_2$, $W^{de}_{FB}\leq W^{de}_{OB}\leq W^{de}_{NB}$;\\
(3) if $\alpha\in(\alpha_1, \alpha_2)$, either $W^{de}_{FB}\geq W^{de}_{OB}\geq W^{de}_{NB}$ or $W^{de}_{FB}\leq W^{de}_{OB}\leq W^{de}_{NB}$, depending on the waiting cost $h$. 
\end{proposition}

The above results are illustrated in Figure \ref{fig:dsw comparison} with $p=0.5$ and $\mathbf{r}=(800, 50, 50, 0)$. When $\alpha$ is low ($\alpha=0.2$), social welfare in the full backlog system is greater than that in the one-sided backlog system, followed by social welfare in the no backlog system. When $\alpha$ is high ($\alpha=0.8$), the reverse is true. When $\alpha$ is in the middle range ($\alpha=0.4$), social welfare in the full backlog system is largest for certain values of $h$, while for other values of $h$, social welfare in the no backlog system is the largest.

\begin{figure}[htpb]
	\centering
	\subfigure[$\alpha=0.2$]{
		\includegraphics[width = 0.3\linewidth]{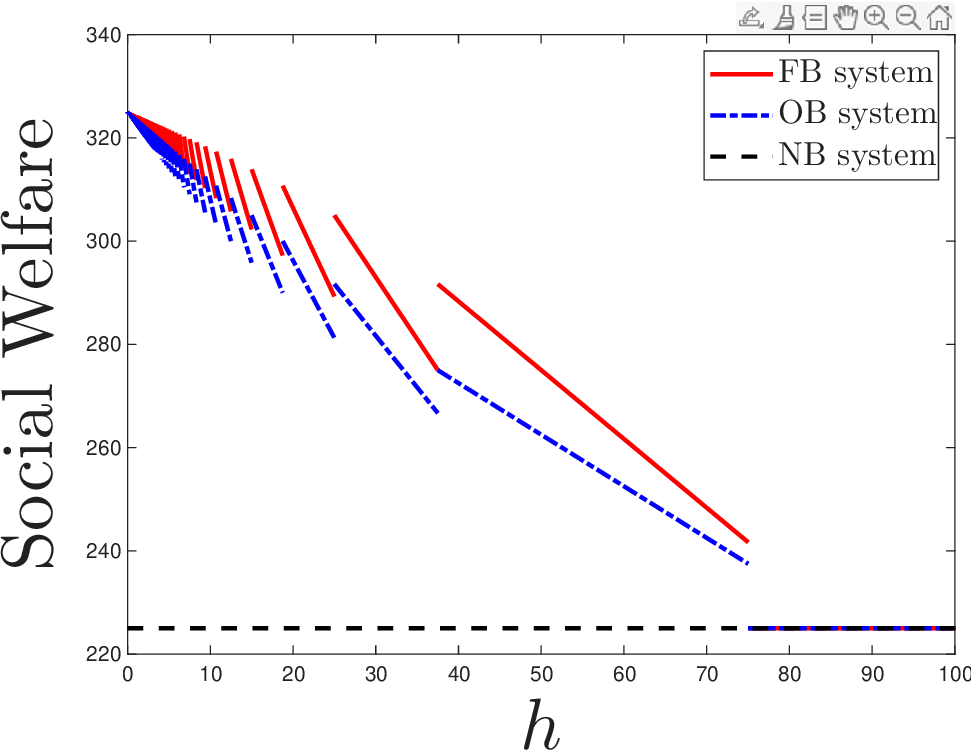}}
	\subfigure[$\alpha=0.4$]{
		\includegraphics[width = 0.3\linewidth]{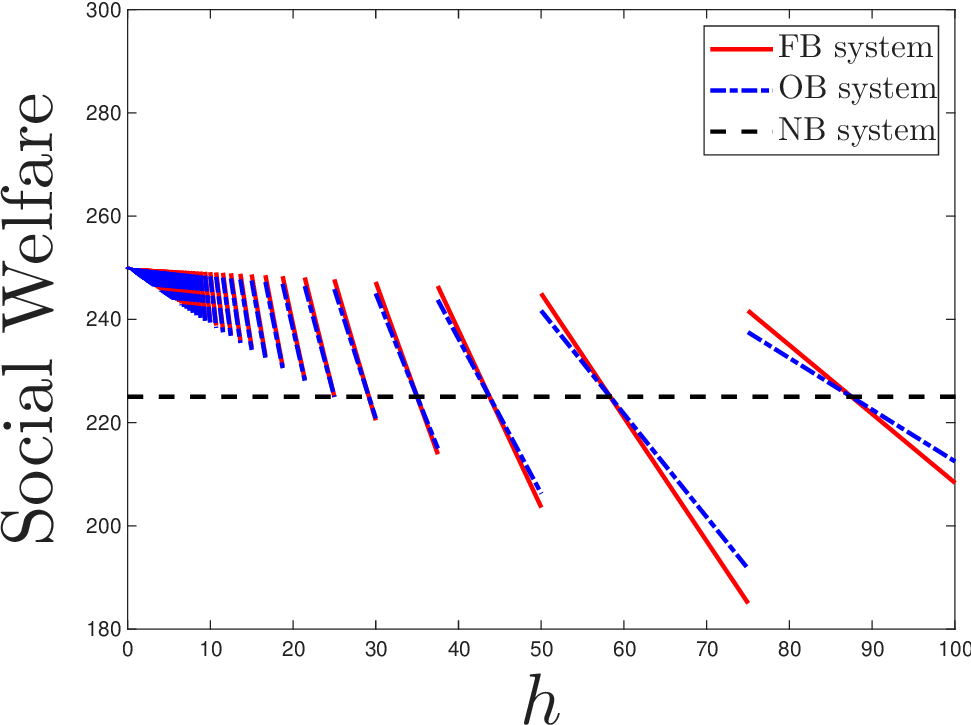}}
        \subfigure[$\alpha=0.8$]{
		\includegraphics[width = 0.3\linewidth]{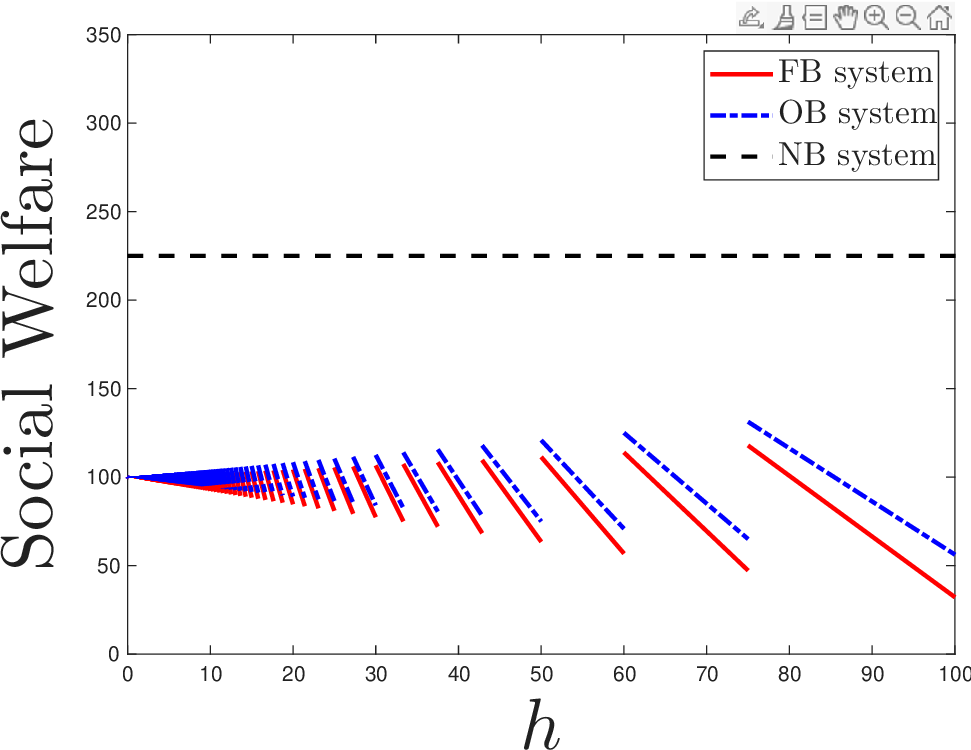}}
	\caption{Comparison of social welfare in the decentralized setting}\label{fig:dsw comparison}
\end{figure}

Proposition \ref{prop:comparison in decentralized setting} implies that social welfare of dynamic matching systems with different patience levels in a decentralized market depends critically on the payoff allocation rule. 
Recall that the equilibrium queue length of the H-type supply queue is $k^{de}=\left\lfloor \frac{q \alpha (r_{HH}-r_{HL})}{h}\right\rfloor$, which increases linearly in 
$\alpha$. Here, $r_{HH}-r_{HL}$ captures the incremental benefit from waiting for a future high-quality assortative match rather than accepting a lower value cross-match in the current period. Thus, $\alpha$ directly scales the private incentive of an H-type supply agent to remain in the system and wait for a better match. 

When $\alpha$ is small enough, H-type supply agents capture only a small share of the match payoff, making them less willing to wait. This leads to a short H-type supply queue in equilibrium. Anticipating a thin future pool of H-type supply, H-type demand agents are also less willing to wait for high-quality matches, if they can. Consequently, both sides behave relatively impatiently, matches form more quickly, and the steady-state population in the system is small. In such a thin market, the matching-quality benefit dominates: additional patience enables valuable opportunities for future assortative matches while the induced increase in waiting costs remains limited. Therefore, social welfare increases as the system moves from a no-backlog system to a one-sided backlog system to a full backlog system. This dynamic characterizes Case (1) of Proposition \ref{prop:comparison in decentralized setting}.

When $\alpha$ is large enough, H-type supply agents internalize a larger share of the surplus and become more willing to wait, which increases the population of H-type supply agents in the system. This thicker pool of patient H-type supply further raises the option value of waiting for H-type demand agents if they can, reinforcing assortative matching behavior. The resulting accumulation of patient H-type agents makes it harder for L-type agents to be matched, leading to substantial congestion and a large steady-state population in the system. In such a thick market, the congestion-cost effect dominates: aggregate waiting costs increase with the population, while the marginal improvement in match quality exhibits diminishing returns. Hence, social welfare declines as patience is expanded. This dynamic characterizes Case (2) of Proposition \ref{prop:comparison in decentralized setting}.

Finally, when $\alpha$ is in a middle range, the system balances between the ``thin" and ``thick" market dynamics described above, and the net welfare impact of patience depends on the unit waiting cost $h$. Within a range of $h$ where the equilibrium matching threshold remains constant, the steady-state population is fixed. As $h$ increases within this range, the total waiting cost increases while match quality remains unchanged; consequently, the net welfare impact of patience deteriorates, potentially shifting from positive to negative. However, as $h$ continues to increase, it triggers a discrete decrease in the matching thresholds. This reduces the steady-state population and makes the system transit to a ``thinner" market state. In such a thinner market, with additional patience, the marginal benefit of assortative matching typically outweighs the congestion cost, making additional patience beneficial again. This dynamic characterizes Case (3) of Proposition \ref{prop:comparison in decentralized setting}.

Proposition  \ref{prop:comparison in decentralized setting} offers a theoretical framework for analyzing the evolution of modern marriage markets, particularly in explaining complex social phenomena such as the rise of ``unmarried women" in China (i.e., ``the proportion of unmarried women between the ages of 25 to 29 in China soared from 9\% in 2000 to 33\% in 2020 and 43\% in 2023, and this trend continues to accelerate.''\footnote{``Why China’s marriage crisis matters''. (2025, Apr 2). The Korea Times. \url{https://www.project-syndicate.org/commentary/china-marriage-crisis-undermines-its-fertility-rate-by-yi-fuxian-2025-03}}). In this context, we can view men as the traditionally more patient ``supply side" and women as the ``demand side." Agent quality (H/L) represents socio-economic status, while the payoff allocation parameter, $\alpha$, can be interpreted as societal bargaining power. A high $\alpha$ reflects a traditional, patriarchal structure in which the match's value is disproportionately credited to the man, while a low $\alpha$ signifies a more progressive society in which female economic and social capital confer greater bargaining power.

 In a traditional societal structure, the marriage market resembled a system with unbalanced patience. Women, facing greater societal and biological pressures, constituted the impatient side, compelled to match quickly. Furthermore, the patriarchal context meant that $\alpha$ was high, favoring the patient (male) side. According to Case (2) of Proposition \ref{prop:comparison in decentralized setting}, in such a high-$\alpha$ regime, waiting would lead to lower social welfare. Therefore, men's strong incentive to wait would have created excessive ``queues" and market gridlock, resulting in detrimental outcomes. However, the market's stability was enforced by women's impatience. By matching quickly, they forced the market to clear, thus preventing the welfare loss that would have occurred had they been patient.

The dynamics shift dramatically in the contemporary era, especially among the highly educated. Increased female economic independence transforms the market into a full backlog system, where both sides are patient and can afford to wait for a desirable partner. This empowerment concurrently shifts bargaining power toward women, reducing $\alpha$. This creates a new market structure. While Case (1) Proposition \ref{prop:comparison in decentralized setting} suggests that a low-$\alpha$ environment benefits from patience, the ``unmarried women" phenomenon persists as a rational equilibrium in the market. This occurs because the low $\alpha$ reduces the net payoff for H-type men to wait for future H-type partners compared to matching immediately with L-type partners. This diminished incentive leads many H-type men to accept immediate cross-matches. Meanwhile, H-type women, now patient enough to select for quality, remain in the queue. 


Consequently, despite the shift toward equality, the decentralized equilibrium results in an ``inefficient" outcome. It fails to produce the high-value assortative matches celebrated as optimal in the marriage market literature (e.g., \citealt{becker1973theory}), and thus generates lower social welfare than a centralized planner could achieve. This observation extends beyond the marriage market: a similar ``efficiency gap" between the decentralized equilibrium and the centralized optimum exists in many modern platforms, from freelance marketplaces like Upwork to ride-hailing services like Uber in California, where drivers act as independent contractors.

This gap also suggests a powerful role for the platform as a central coordinator. If the platform can design mechanisms to properly allocate the surplus welfare gained from achieving the centralized optimum, for example, through targeted incentives, subsidies, or dynamic payoff structures, it can realign agent incentives with the system-wide optimum. Such an intervention could create a Pareto improvement in which the system as a whole is more efficient, and all agents are better off. Returning to the example of the marriage market, a government agency designed around these principles could respond to the ``unmarried women" phenomenon by creating specific incentives for H-type men to match with their H-type women counterparts, thereby correcting market failure and enabling the high-value matches that are currently being missed.



\section{Conclusion} \label{sec:conclusion}

In this paper, we study the dynamics of two-sided matching markets under the ubiquitous yet underexplored condition of unbalanced patience. By modeling a patient-supply side and an impatient-demand side, we derive strategic implications for market design in a two-sided dynamic matching platform and offer a new lens for interpreting complex market outcomes.

Our analysis of the system's structure highlights a fundamental distinction between centralized and decentralized markets. In the centralized benchmark, the optimal policy utilizes a threshold-based rule to ration high-quality supply, preserving it for future high-quality demand. In the decentralized setting, the welfare-maximizing Markov perfect equilibrium is characterized by thresholds that emerge from agents' forward-looking, self-interested incentives. More importantly, we demonstrate that, unlike in the full-backlog system, where inefficiency is inevitable, the decentralized one-sided backlog system can be perfectly coordinated with the centralized optimum. By appropriately tuning the matching payoff allocation, a platform can regulate the supply-side incentive to wait, thereby aligning private equilibrium queues with the social optimum.

We also investigate the value of patience across different system designs. In centralized settings, patience always improves welfare, though with diminishing marginal returns, suggesting that inducing patience on just one side captures most of the potential efficiency gains. In decentralized settings, however, the value of patience is contingent on market conditions determined by the matching payoff allocation proportion $\alpha$. A low $\alpha$ creates a ``thin" market where patience is beneficial, enabling high-quality matches with minimal congestion. Conversely, a high $\alpha$ induces a ``thick" market where the congestion costs of excessive queuing outweigh the benefits of assortative matching.

We note that these implications are robust and continue to hold when the matching payoff follows a horizontal structure: $r_{HH}\geq r_{LL}\geq \max\{r_{HL}, r_{LH}\}$. In this case, both H-type and L-type agents on one side would prefer to match with agents similar to themselves on the other side.

\bibliographystyle{ormsv080}
\bibliography{dm}
\newpage

\setcounter{definition}{0}
\setcounter{lemma}{0}
\setcounter{corollary}{0}
\setcounter{equation}{0}
\setcounter{theorem}{0}
\renewcommand\thedefinition{A.\arabic{definition}}
\renewcommand\thelemma{A.\arabic{lemma}}
\renewcommand\thecorollary{A.\arabic{corollary}}
\renewcommand\theequation{A.\arabic{equation}}
\renewcommand\thetheorem{B.\arabic{theorem}}

\pagenumbering{arabic}
\setcounter{page}{1}

\begin{center}
{\Large Online Appendix to \\
``Dynamic Matching Under Unbalanced Impatience''}
\end{center}
\numberwithin{lemma}{section}
\numberwithin{definition}{section}
\numberwithin{equation}{section}



\subsection*{A. Proof of Lemma \ref{lem: centralized threshold type}}

{
Suppose that there are $N$ supply agents in a period $t$.
We will show that when $N$ is sufficiently large, any policy $\pi$ that does not match any of the supply agents to the arrived demand agent in period $t$ is suboptimal.

Since there are $N$ supply agents in total, there exists a type $i\in\{H,L\}$ such that at least $N/2$ of the supply agents are of type $i$. 
Under policy $\pi$, let us consider a supply agent such that it is either never matched or matched after all other supply agents who are available in period $t$. We will refer to this supply agent as $\omega$.

To show that $\pi$ is suboptimal, we will modify it to obtain a policy with better performance.
Let the matching policy $\pi'$ be identical to $\pi$ except that $\pi'$ matches supply agent $\omega$ in period $t$.
Since there is exactly one demand arrival in each period and $\omega$ is matched only after all other supply agents under policy $\pi$, this will save a total waiting cost of at least $h \cdot N/2$ for holding $\omega$. 
On the other hand, compared with $\pi'$, policy $\pi$ may lead to an increase in the matching payoff by at most $r_{ij}-r_{iH}\le r_{HL}-r_{HH}$, where $j$ is the type of the demand agent who arrives in period $t$ and the inequality is due to Assumptions \ref{as:hp} and \ref{as:sm}.
It is easy to see that for sufficiently large $N$, the total waiting cost saved by policy $\pi'$, which is at least $h \cdot N/2$, outweighs the potential increase in the matching payoff.
Thus, $\pi$ is subpotimal.
%
}
$\Box$

\subsection*{B. Proof of Theorem \ref{thm:optimal mechanism of centralized system}}
In Lemma \ref{lem: centralized threshold type}, we demonstrated that it is always optimal to match the arriving demand agent with a supply agent when the total number of supply agents is sufficiently large. Therefore, we consider a bounded state space where the number of supply agents does not exceed some $M>\frac{2(r_{HH}-r_{HL})}{h}-1$. We will characterize the optimal policies for any finite $M$ and show that they have the same structure. 

For any finite $M$, we formulate this problem as the following Markov Decision Process (MDP): 
\begin{eqnarray*}
	{\rm MDP}=\{T, \mathbf{S}, \mathbf{s}^{0}, (A_\mathbf{s})_{\mathbf{s}\in \mathbf{S}}, R(\mathbf{s}, a), P(\mathbf{s}'|\mathbf{s}, a)\}
\end{eqnarray*}
where
\begin{enumerate}
	\item $T=\{1, 2, ...\}$ is the set of decision periods. 
	\item $\mathbf{S}=\mathbf{S}_{H}\cup\mathbf{S}_{L}$ is the set of possible states: $$\mathbf{S}_{H}=\{(x_{H}, x_{L}, y_{H}, y_{L})| 0\le x_{H}+x_{L} \le M, x_{H}, x_{L} \in \mathbb{Z}\cup \{0\}, y_{H}=1, y_{L}=0\},$$ $$\mathbf{S}_{L}=\{(x_{H}, x_{L}, y_{H}, y_{L})| 0\le x_{H}+x_{L} \le M, x_{H}, x_{L} \in \mathbb{Z}\cup \{0\}, y_{H}=0, y_{L}=1\},$$
	where $\mathbf{S}_{H}$ includes all possible states when an H-type demand agent arrives, and $\mathbf{S}_{L}$ includes all possible states when an L-type demand agent arrives. Following the discussion above, we assume that $M$ is sufficiently large (i.e., $M>\frac{2(r_{HH}-r_{HL})}{h}-1$) such that the possible states are indeed within in $\mathbf{S}$. When the total amount of supply agents reaches $M$, it is always optimal to match every arriving demand agent with some type of supply agent. As a consequence, the number of supply agents will not increase further. We use $\mathbf{s}$ to represent the state vector of a specific period. 
	\item $\mathbf{s}^{0}=(0, 0, 0, 0)$, i.e., we assume that there is no agent before the period $1$. Agents arrive in the system, and matching decisions are made starting in period $1$.  
	\item $A_{\mathbf{S}}=A_{\mathbf{S}_{H}}\cup A_{\mathbf{S}_{L}}$ is the set of possible actions: 
	$$A_{\mathbf{S}_{H}}=\{u_{\phi}, u_{HH}, u_{LH}\},~~A_{\mathbf{S}_{L}}=\{u_{\phi}, u_{HL}, u_{LL}\},$$ 
	where $u_{\phi}$ denotes the action of no matching, and $u_{ij}$ denotes the action of using a $i$-type supply agent to satisfy a $j$-type demand agent. In the state of $\mathbf{S}_{H}$, when an H-type demand agent arrives, the action set, $A_{\mathbf{S}_{H}}$, includes no matching, $u_{\phi}$, to use an H-type supply agent to match with the H-type demand agent, $u_{HH}$, and to use L-type supply agent to match with the H-type demand agent, $u_{LH}$. Similarly, when an L-type demand agent arrives, the action set, $A_{\mathbf{S}_{L}}$, includes no matching, $u_{\phi}$, to use an H-type supply agent to match with the L-type demand agent, $u_{HL}$, and to use a L-type supply agent to match with the L-type demand agent, $u_{LL}$. Specifically, we let $u_{\phi}=(0,0,0,0)$, $u_{HH}=(1, 0, 0, 0)$, $	u_{HL}=(0, 1, 0, 0)$, $u_{LH}=(0, 0, 1, 0)$, and $u_{LL}=(0, 0, 0, 1)$.
	\item $R(\mathbf{s}, a)$ is the immediate payoff  received in state $\mathbf{s}$ with the action $a\in \mathbf{A}_{\mathbf{S}}$. Specifically,
	\begin{eqnarray*}
		R(\mathbf{s}, a)=\left\{
		\begin{array}{ll}
			-h(x_{H}+x_{L}), & \mbox{ if $\mathbf{s}=(x_{H}, x_{L}, 1, 0)$ and $a=u_{\phi}$,} \\
			r_{HH}-h(x_{H}-1+x_{L}), & \mbox{ if $\mathbf{s}=(x_{H}, x_{L}, 1, 0)$ and $a=u_{HH}$,} \\  
			r_{LH}-h(x_{H}+x_{L}-1), &  \mbox{ if $\mathbf{s}=(x_{H}, x_{L}, 1, 0)$ and $a=u_{LH}$,} \\
			-h(x_{H}+x_{L}), & \mbox{ if $\mathbf{s}=(x_{H}, x_{L}, 0, 1)$ and $a=u_{\phi}$,} \\  
			r_{HL}-h(x_{H}-1+x_{L}), &  \mbox{ if $\mathbf{s}=(x_{H}, x_{L}, 0, 1)$ and $a=u_{HL}$,}   \\
			r_{LL}-h(x_{H}+x_{L}-1), & \mbox{ if $\mathbf{s}=(x_{H}, x_{L}, 0, 1)$ and $a=u_{LL}$.}  
		\end{array}
		\right.
	\end{eqnarray*} 
	\item $P(\mathbf{s}'|\mathbf{s}, a)$ is the transition probability function  that action $a\in A_{\mathbf{s}}$ in state $\mathbf{s}\in \mathbf{S}$ at period $t$ will lead to state $\mathbf{s}'\in \mathbf{S}$ at period $t+1$. Specifically,
	$$P(\mathbf{s}'\in \mathbf{S}_{H}|\mathbf{s}\in \mathbf{S}_{H}, a\in A_{\mathbf{S}_{H}})=\left\{
	\begin{array}{ll}
	pq,  & \mbox{ if $\mathbf{s}'=(x_{H}+1, x_{L}, 1, 0)$, $\mathbf{s}=(x_{H}, x_{L}, 1, 0)$, and $a=u_{\phi}$,}  \\
	(1-p)q, & \mbox{ if $\mathbf{s}'=(x_{H}, x_{L}+1, 1, 0)$, $\mathbf{s}=(x_{H}, x_{L}, 1, 0)$, and $a=u_{\phi}$,}  \\
	pq, & \mbox{ if $\mathbf{s}'=(x_{H}, x_{L}, 1, 0)$, $\mathbf{s}=(x_{H}, x_{L}, 1, 0)$, and $a=u_{HH}$,}\\
	(1-p)q, & \mbox{ if $\mathbf{s}'=(x_{H}-1, x_{L}+1, 1, 0)$, $\mathbf{s}=(x_{H}, x_{L}, 1, 0)$, and $a=u_{HH}$,}  \\
	pq, & \mbox{ if $\mathbf{s}'=(x_{H}+1, x_{L}-1, 1, 0)$, $\mathbf{s}=(x_{H}, x_{L}, 1, 0)$, and $a=u_{LH}$,}   \\
	(1-p)q, &   \mbox{ if $\mathbf{s}'=(x_{H}, x_{L}, 1, 0)$, $\mathbf{s}=(x_{H}, x_{L}, 1, 0)$, and $a=u_{LH}$.}
	\end{array}
	\right.$$
	The formulas of
	$P(\mathbf{s}'\in \mathbf{S}_{L}|\mathbf{s}\in \mathbf{S}_{H}, a\in A_{\mathbf{S}_{H}})$,  $P(\mathbf{s}'\in \mathbf{S}_{H}|\mathbf{s}\in \mathbf{S}_{L}, a\in A_{\mathbf{S}_{L}})$, and $P(\mathbf{s}'\in \mathbf{S}_{L}|\mathbf{s}\in \mathbf{S}_{L}, a\in A_{\mathbf{S}_{L}})$ can be written similarly. 
\end{enumerate}

Under a given policy $\pi$, the long-run average profit is given by:
\begin{eqnarray}\label{policy limit}
J(\pi)=\lim_{T\rightarrow \infty}\frac{1}{T}E\left[\sum^{T}_{t=1}R(\mathbf{s}^t, \pi(\mathbf{s}^t))\right].
\end{eqnarray}
Our objective is to find a policy $\pi^*$ such that the long-run average profit is maximized.

It is easy to verify that the MDP described above is a multichain model. Since the state space $\mathbf{S}$ and the action space $(A_\mathbf{s})_{\mathbf{s}\in \mathbf{S}}$ are both finite, by Theorem 9.1.4 of \citet{puterman2014markov}, the optimal policy $\pi^*$ exists. We characterize the structure of the optimal policy below. 

Consider the following recursive equations 
\begin{eqnarray}\label{appendix:dp.eq0}
H_{t+1}(\mathbf{s})=\max_{a\in A_{\mathbf{S}}}\left\{R(\mathbf{s}, a)+\sum_{\mathbf{s}'\in \mathbf{S}}P(\mathbf{s}'|\mathbf{s}, a)H_t(\mathbf{s}')\right\}, \forall \mathbf{s}\in\mathbf{S}, t=0,1,2,\dots,
\end{eqnarray}
with $H_0(\mathbf{s})=0$ for all $s\in\mathbf{S}$. Let 
$$d_t(\mathbf{s})\in \arg\max_{a\in A_{\mathbf{S}}}\left\{R(\mathbf{s}, a)+\sum_{\mathbf{s}'\in \mathbf{S}}P(\mathbf{s}'|\mathbf{s}, a)H_t(\mathbf{s}')\right\}, \forall \mathbf{s}\in\mathbf{S}, t=0,1,2,\dots.$$
We can verify that for this MDP, every stationary deterministic policy has an aperiodic transition matrix. Therefore, by Theorem 9.4.5 of  \citet{puterman2014markov}, $\{d_t(\mathbf{s})\}_{\mathbf{s}\in \mathbf{S}}$ converges to the optimal policy $\pi^*$ as $t$ goes to infinity. In what follows, we first show that $H_t$ satisfies some functional properties. Then, we use those properties to characterize the structure of $d_t$ and show that $d_t$ has the same structure for all $t=0,1,2,\dots$. Therefore, the optimal policy $\pi^*$ has the same structure.

Note that the recursive equations \eqref{appendix:dp.eq0} can be rewritten as follows: for $t=1,2,\dots,$

\noindent $\bullet$ Case 1: If $x_{H}=0$,
\begin{equation}\label{appendix:dp.eq2}
H_t(0, x_{L}, 1, 0)
= -h x_{L}+\max \{r_{LH}+h+V_t (0, x_{L}-1), V_t (0, x_{L})\}.
\end{equation}
\begin{equation}\label{appendix:dp.eq3}
H_t(0, x_{L}, 0, 1)
=-h x_{L}+\max \{r_{LL}+h+V_t (0, x_{L}-1), V_t (0, x_{L})\}.
\end{equation}

\noindent $\bullet$ Case 2: If $x_{L}=0$,
\begin{equation}\label{appendix:dp.eq4}
H_t(x_{H}, 0, 1, 0)
= -h x_{H}+\max \{ r_{HH}+h+V_t (x_{H}-1, 0),V_t (x_{H}, 0)\}.
\end{equation}
\begin{equation}\label{appendix:dp.eq5}
H_t(x_{H}, 0, 0, 1)
=-h x_{H}+\max \{ r_{HL}+h+V_t (x_{H}-1, 0), V_t (x_{H}, 0)\}.
\end{equation}

\noindent $\bullet$ Case 3: If $x_{H}>0$ and $x_{L}>0$,
\begin{eqnarray}\label{appendix:dp.eq6}
H_t(x_{H}, x_{L}, 1, 0)
&=&-h(x_{H}+x_{L})+\max \{ r_{HH}+h+V_t (x_{H}-1, x_{L}),\nonumber \\
&&r_{LH}+h+V_t (x_{H}, x_{L}-1), V_t (x_{H}, x_{L})\}.
\end{eqnarray}
\begin{eqnarray}\label{appendix:dp.eq7}
H_t(x_{H}, x_{L}, 0, 1)
&=&-h(x_{H}+x_{L})+\max \{ r_{HL}+h+V_t (x_{H}-1, x_{L}),\nonumber \\
&&r_{LL}+h+V_t (x_{H}, x_{L}-1), V_t(x_{H}, x_{L})\}.
\end{eqnarray}
where 
\begin{eqnarray}\label{appendix:dp.eq1}
V_t(x_{H}, x_{L})  &=& pqH_{t-1}(x_{H}+1, x_{L}, 1, 0)+(1-p)q H_{t-1}(x_{H}, x_{L}+1, 1, 0) \nonumber \\
&&+p(1-q)H_{t-1}(x_{H}+1, x_{L}, 0, 1)+(1-p)(1-q)H_{t-1}(x_{H}, x_{L}+1, 0, 1).
\end{eqnarray}

For any functions $H$ defined on $\mathbf{S}$ and $V$ defined on $\{(x_{H}, x_{L})|0\leq x_{H}+x_{L}\leq M\}$, let
\begin{eqnarray*}
	\Delta_{x_{H}} H(x_{H}, x_{L}, y_{H}, y_{L})&=& H(x_{H}+1, x_{L},y_{H}, y_{L})-H(x_{H}, x_{L},y_{H}, y_{L})\\
	\Delta_{x_{L}} H(x_{H}, x_{L}, y_{H}, y_{L})&=& H(x_{H}, x_{L}+1, y_{H}, y_{L})-H(x_{H}, x_{L}, y_{H}, y_{L})\\
	\Delta_{x_{H}} V(x_{H}, x_{L})&=& V(x_{H}+1, x_{L})-V(x_{H}, x_{L})\\
	\Delta_{x_{L}} V(x_{H}, x_{L})&=& V(x_{H}, x_{L}+1)-V(x_{H}, x_{L}).
\end{eqnarray*}

Next, we use induction to show that for $t=1,2,\dots$, $V_t(x_{H}, x_{L})$ satisfies 
\begin{description}
	\item[(i)] $\Delta_{x_{H}}V_t(x_{H}, x_{L})\le r_{HH}+h$;
	\item[(ii)] $\Delta_{x_{L}}V_t(x_{H}, x_{L})\le r_{LL}+h$;
	\item[(iii)] $r_{HL}-r_{LL}\leq  \Delta_{x_{H}}V_t(x_{H}, x_{L})-\Delta_{x_{L}}V_t(x_{H}, x_{L})\le r_{HH}-r_{LH}$;
	\item[(iv)] $V_t(x_{H}+2, x_{L})-2V_t(x_{H}+1, x_{L})+V_t(x_{H}, x_{L})\leq 0$. 
\end{description}

Clearly, $V_t(x_{H}, x_{L})=0$ for all $x_{H}, x_{L}$ and thus satisfies properties (\textit{i})-(\textit{iv}). Now suppose  $V_t(x_{H}, x_{L})$ satisfies properties (\textit{i})-(\textit{iv}). We need to show that $V_{t+1}(x_{H}, x_{L})$ also satisfies properties (\textit{i})-(\textit{iv}).  To show the preservation, we first prove in {\bf Step 1} that $H_t(x_{H},x_{L}, y_{H},y_{L})$ satisfies 
\begin{description}
	\item[(Hi)] $\Delta_{x_{H}}H_t(x_{H},x_{L}, y_{H},y_{L})\le r_{HH}+h$;
	\item[(Hii)] $\Delta_{x_{L}}H_t(x_{H},x_{L}, y_{H},y_{L})\le r_{LL}+h$;
	\item[(Hiii)] $ r_{HL}-r_{LL}\leq  \Delta_{x_{H}}H_t(x_{H},x_{L}, y_{H},y_{L})-\Delta_{x_{L}}H_t(x_{H},x_{L}, y_{H},y_{L})\le r_{HH}-r_{LH}$;
	\item[(Hiv)] $H_t(x_{H}+2, x_{L}, y_{H},y_{L})-2H_t(x_{H}+1, x_{L}, y_{H},y_{L})+H_t(x_{H}, x_{L}, y_{H},y_{L})\leq 0$. 
\end{description}
Then, we show in {\bf Step 2} that $V_{t+1}(x_{H}, x_{L})$ also satisfies properties (\textit{i})-(\textit{iv}).

{\bf Step 1:} Since  $V_t(x_{H}, x_{L})$ satisfies properties (\textit{i})-(\textit{iv}), we can rewrite equations (\ref{appendix:dp.eq2})-(\ref{appendix:dp.eq7}) as follows:
\noindent $\bullet$ Case 1: If $x_{H}=0$,
\begin{align*}
H_t(0, x_{L}, 1, 0)
&= -h x_{L}+r_{LH}+h+V_t (0, x_{L}-1),\\
H_t(0, x_{L}, 0, 1)
&=-h x_{L}+r_{LL}+h+V_t (0, x_{L}-1).
\end{align*}

\noindent $\bullet$ Case 2: If $x_{L}=0$,
\begin{align*}
H_t(x_{H}, 0, 1, 0)
&= -h x_{H}+r_{HH}+h+V_t (x_{H}-1, 0),\\
H_t(x_{H}, 0, 0, 1)
&=-h x_{H}+\max \{ r_{HL}+h+V_t (x_{H}-1, 0), V_t (x_{H}, 0)\}.
\end{align*}

\noindent $\bullet$ Case 3: If $x_{H}>0$ and $x_{L}>0$,
\begin{align*}
H_t(x_{H}, x_{L}, 1, 0)
&=-h(x_{H}+x_{L})+ r_{HH}+h+V_t (x_{H}-1, x_{L}),\\
H_t(x_{H}, x_{L}, 0, 1)
&=-h(x_{H}+x_{L})+r_{LL}+h+V_t (x_{H}, x_{L}-1).
\end{align*}

\noindent\textbf{Proof of (\textit{Hi}):} If $x_{H}>0$, we have 
\begin{align*}
\Delta_{x_{H}}H_t(x_{H},x_{L}, 1,0)&=H_t(x_{H}+1, x_{L}, 1, 0)-H_t(x_{H}, x_{L}, 1, 0)\\
&=[-h(x_{H}+1+x_{L})+ r_{HH}+h+V_t (x_{H}, x_{L})]-[-h(x_{H}+x_{L})+ r_{HH}+h+V_t (x_{H}-1, x_{L})]\\
&=-h+\Delta_{x_{H}}V_t(x_{H}-1,x_{L})\\
&\leq -h+(r_{HH}+h)\\
&\leq r_{HH}+h.
\end{align*}
If $x_{H}=0$, we have 
\begin{align*}
\Delta_{x_{H}}H_t(0,x_{L}, 1,0)&=H_t(1, x_{L}, 1, 0)-H_t(0, x_{L}, 1, 0)\\
&=[-h(1+x_{L})+ r_{HH}+h+V_t (0, x_{L})]-[-h x_{L}+r_{LH}+h+V_t (0, x_{L}-1)]\\
&=-h+r_{HH}-r_{LH}+\Delta_{x_{L}}V_t(0,x_{L}-1)\\
&\leq -h+r_{HH}-r_{LH}+r_{LL}+h\\
&\leq r_{HH}+h.
\end{align*}
If $x_{L}>0$, we have 
\begin{align*}
\Delta_{x_{H}}H_t(x_{H},x_{L}, 0,1)&=H_t(x_{H}+1, x_{L}, 0, 1)-H_t(x_{H}, x_{L}, 0, 1)\\
&=[-h(x_{H}+1+x_{L})+ r_{LL}+h+V_t (x_{H}+1, x_{L}-1)]\\
&-[-h(x_{H}+x_{L})+ r_{LL}+h+V_t (x_{H}, x_{L}-1)]\\
&=-h+\Delta_{x_{H}}V_t(x_{H},x_{L}-1)\\
&\leq -h+(r_{HH}+h)\\
&\leq r_{HH}+h.
\end{align*}
If $x_{L}=0$, we have 
\begin{align*}
\Delta_{x_{H}}H_t(x_{H},0, 0,1)&=H_t(x_{H}+1, 0, 0, 1)-H_t(x_{H}, 0, 0, 1)\\
&=[-h(x_{H}+1)+\max \{ r_{HL}+h+V_t (x_{H}, 0), V_t (x_{H}+1, 0)\}]\\
&-[-h x_{H}+\max \{ r_{HL}+h+V_t (x_{H}-1, 0), V_t (x_{H}, 0)\}].
\end{align*}
When $\max \{ r_{HL}+h+V_t (x_{H}, 0), V_t (x_{H}+1, 0)\}=r_{HL}+h+V_t (x_{H}, 0)$,
\begin{align*}
\Delta_{x_{H}}H_t(x_{H},0, 0,1)&\leq [-h(x_{H}+1)+r_{HL}+h+V_t (x_{H}, 0)]-[-h x_{H}+ r_{HL}+h+V_t (x_{H}-1, 0)]\\
&=-h+\Delta_{x_{H}}V_t(x_{H}-1,0)\\
&\leq -h+(r_{HH}+h)\\
&\leq r_{HH}+h.
\end{align*}
When $\max \{ r_{HL}+h+V_t (x_{H}, 0), V_t (x_{H}+1, 0)\}=V_t (x_{H}+1, 0)$,
\begin{align*}
\Delta_{x_{H}}H_t(x_{H},0, 0,1)&\leq [-h(x_{H}+1)+V_t (x_{H}+1, 0)]-[-h x_{H}+ V_t (x_{H}, 0)]\\
&=-h+\Delta_{x_{H}}V_t(x_{H},0)\\
&\leq -h+(r_{HH}+h)\\
&\leq r_{HH}+h.
\end{align*}

\noindent\textbf{Proof of (\textit{Hii}):} If $x_{H}>0$, we have 
\begin{align*}
\Delta_{x_{L}}H_t(x_{H},x_{L}, 1,0)&=H_t(x_{H}, x_{L}+1, 1, 0)-H_t(x_{H}, x_{L}, 1, 0)\\
&=[-h(x_{H}+x_{L}+1)+ r_{HH}+h+V_t (x_{H}-1, x_{L}+1)]\\
&-[-h(x_{H}+x_{L})+ r_{HH}+h+V_t (x_{H}-1, x_{L})]\\
&=-h+\Delta_{x_{L}}V_t(x_{H}-1,x_{L})\\
&\leq -h+(r_{LL}+h)\\
&\leq r_{LL}+h.
\end{align*}
If $x_{H}=0$, we have 
\begin{align*}
\Delta_{x_{L}}H_t(0,x_{L}, 1,0)&=H_t(0, x_{L}+1, 1, 0)-H_t(0, x_{L}, 1, 0)\\
&=[-h(x_{L}+1)+ r_{LH}+h+V_t (0, x_{L})]-[-h x_{L}+r_{LH}+h+V_t (0, x_{L}-1)]\\
&=-h+\Delta_{x_{L}}V_t(0,x_{L}-1)\\
&\leq -h+r_{LL}+h\\
&\leq r_{LL}+h.
\end{align*}
If $x_{L}>0$, we have 
\begin{align*}
\Delta_{x_{L}}H_t(x_{H},x_{L}, 0,1)&=H_t(x_{H}, x_{L}+1, 0, 1)-H_t(x_{H}, x_{L}, 0, 1)\\
&=[-h(x_{H}+x_{L}+1)+ r_{LL}+h+V_t (x_{H}, x_{L})]-[-h(x_{H}+x_{L})+ r_{LL}+h+V_t (x_{H}, x_{L}-1)]\\
&=-h+\Delta_{x_{L}}V_t(x_{H},x_{L}-1)\\
&\leq -h+(r_{LL}+h)\\
&\leq r_{LL}+h.
\end{align*}
If $x_{L}=0$, we have 
\begin{align*}
\Delta_{x_{H}}H_t(x_{H},0, 0,1)&=H_t(x_{H}, 1, 0, 1)-H_t(x_{H}, 0, 0, 1)\\
&=[-h(x_{H}+1)+r_{LL}+h+V_t (x_{H}, 0)]-[-h x_{H}+\max \{ r_{HL}+h+V_t (x_{H}-1, 0), V_t (x_{H}, 0)\}]\\
&\leq [-h(x_{H}+1)+r_{LL}+h+V_t (x_{H}, 0)]-[-h x_{H}+ V_t (x_{H}, 0)]\\
&\leq r_{LL}+h.
\end{align*}

\noindent\textbf{Proof of (\textit{Hiii}):} If $x_{H}>0$, we have 
\begin{align*}
\Delta_{x_{H}}H_t(x_{H},x_{L}, &1, 0)-\Delta_{x_{L}}H_t(x_{H},x_{L},1, 0) \\
&=H_t(x_{H}+1,x_{L}, 1,0)-H_t(x_{H},x_{L}+1,1,0)\\
&=[-h(x_{H}+1+x_{L})+ r_{HH}+h+V_t (x_{H}, x_{L})]-[-h(x_{H}+x_{L}+1)+ r_{HH}+h+V_t (x_{H}-1, x_{L}+1)]\\
&=V_t(x_{H},x_{L})-V_t(x_{H}-1,x_{L}+1)\\
&=\Delta_{x_{H}}V_t(x_{H}-1,x_{L})-\Delta_{x_{L}}V_t(x_{H}-1,x_{L})\in [r_{HL}-r_{LL}, r_{HH}-r_{LH}].
\end{align*}
If $x_{H}=0$, we have 
\begin{align*}
\Delta_{x_{H}}H_t(0,x_{L}, &1, 0)-\Delta_{x_{L}}H_t(0,x_{L},1, 0) \\
&=H_t(1,x_{L}, 1,0)-H_t(0,x_{L}+1,1,0)\\
&=[-h(1+x_{L})+ r_{HH}+h+V_t (0, x_{L})]-[-h(x_{L}+1)+ r_{LH}+h+V_t (0, x_{L})]\\
&=r_{HH}-r_{LH}\in [r_{HL}-r_{LL}, r_{HH}-r_{LH}].
\end{align*}
If $x_{L}>0$, we have
\begin{align*}
\Delta_{x_{H}}H_t(x_{H},x_{L}, &0,1)-\Delta_{x_{L}}H_t(x_{H},x_{L}, 0,1) \\
&=H_t(x_{H}+1,x_{L}, 0,1)-H_t(x_{H},x_{L}+1, 0,1)\\
&=[-h(x_{H}+1+x_{L})+ r_{LL}+h+V_t (x_{H}+1, x_{L}-1)]-[-h(x_{H}+x_{L}+1)+ r_{LL}+h+V_t (x_{H}, x_{L})]\\
&=V_t(x_{H}+1,x_{L}-1)-V_t(x_{H},x_{L})\\
&=\Delta_{x_{H}}V_t(x_{H},x_{L}-1)-\Delta_{x_{L}}V_t(x_{H},x_{L}-1)\in [r_{HL}-r_{LL}, r_{HH}-r_{LH}].
\end{align*}
If $x_{L}=0$, we have
\begin{align*}
\Delta_{x_{H}}H_t(x_{H},0, &0,1)-\Delta_{x_{L}}H_t(x_{H},0, 0,1) \\
&=H_t(x_{H}+1,0, 0,1)-H_t(x_{H},1, 0,1)\\
&=[-h(x_{H}+1)+\max \{ r_{HL}+h+V_t (x_{H}, 0), V_t (x_{H}+1, 0)\}]-[-h(x_{H}+1)+ r_{LL}+h+V_t (x_{H}, 0)].
\end{align*}
When $\max \{ r_{HL}+h+V_t (x_{H}, 0), V_t (x_{H}+1, 0)\}=r_{HL}+h+V_t (x_{H}, 0)$,
\begin{align*}
\Delta_{x_{H}}H_t(x_{H},0, &0,1)-\Delta_{x_{L}}H_t(x_{H},0, 0,1) \\
&= [-h(x_{H}+1)+r_{HL}+h+V_t (x_{H}, 0)]-[-h(x_{H}+1)+ r_{LL}+h+V_t (x_{H}, 0)]\\
&=r_{HL}-r_{LL}\in [r_{HL}-r_{LL}, r_{HH}-r_{LH}].
\end{align*}
When $\max \{ r_{HL}+h+V_t (x_{H}, 0), V_t (x_{H}+1, 0)\}=V_t (x_{H}+1, 0)$,
\begin{align*}
\Delta_{x_{H}}H_t(x_{H},0, &0,1)-\Delta_{x_{L}}H_t(x_{H},0, 0,1) \\
&= [-h(x_{H}+1)+V_t (x_{H}+1, 0)]-[-h(x_{H}+1)+ r_{LL}+h+V_t (x_{H}, 0)]\\
&\geq[-h(x_{H}+1)+r_{HL}+h+V_t (x_{H}, 0)]-[-h(x_{H}+1)+ r_{LL}+h+V_t (x_{H}, 0)]\\
&=r_{HL}-r_{LL},
\end{align*}
and 
\begin{align*}
\Delta_{x_{H}}H_t(x_{H},0, &0,1)-\Delta_{x_{L}}H_t(x_{H},0, 0,1) \\
&= [-h(x_{H}+1)+V_t (x_{H}+1, 0)]-[-h(x_{H}+1)+ r_{LL}+h+V_t (x_{H}, 0)]\\
&\leq [-h(x_{H}+1)+V_t (x_{H}+1, 0)]-[-h(x_{H}+1)+V_t (x_{H},1)]\\
&= V_t (x_{H}+1, 0)-+V_t (x_{H},1)\\
&=\Delta_{x_{H}}V_t(x_{H},0)-\Delta_{x_{L}}V_t(x_{H},0)\leq r_{HH}-r_{LH}.
\end{align*}

\noindent\textbf{Proof of (\textit{Hiv}):}
If $x_{H}>0$, we have  
\begin{align*}
H_t(&x_{H}+2,x_{L}, 1, 0)-2H_t(x_{H}+1,x_{L},1, 0)+H_t(x_{H},x_{L},1, 0)\\
&=[-h(x_{H}+2+x_{L})+ r_{HH}+h+V_t (x_{H}+1, x_{L})]-2[-h(x_{H}+1+x_{L})+ r_{HH}+h+V_t (x_{H}, x_{L})]\\
&\;\;\;\;\;\;\;\;\;\;\;\;\;\;\;\;\;\;\;\;\;\;\;\;\;\;\;\;\;\;\;\;\;\;\;\;\;\;\;\;\;\;\;\;\;\;\;\;\;\;\;\;\;\;\;\;\;\;\;\;\;\;\;\;\;\;\;\;\;\;\;\;\;\;\;\;\;\;\;\;\;\;\;\;\;\;\;\;\;\;\;\;\;\;\;\;\;\;\;\;\;\;+[-h(x_{H}+x_{L})+ r_{HH}+h+V_t (x_{H}-1, x_{L})]\\
&=V_t(x_{H}+1,x_{L})-2V_t(x_{H},x_{L})+V_t(x_{H}-1,x_{L})\leq 0.
\end{align*}
If $x_{H}=0$, we have  
\begin{align*}
H_t(&2,x_{L}, 1, 0)-2H_t(1,x_{L},1, 0)+H_t(0,x_{L},1, 0)\\
&=[-h(2+x_{L})+ r_{HH}+h+V_t (1, x_{L})]-2[-h(1+x_{L})+ r_{HH}+h+V_t (0, x_{L})]\\
&+[-hx_{L}+ r_{LH}+h+V_t (0, x_{L}-1)]\\
&=r_{LH}-r_{HH}+\Delta_{x_{H}}V_t(0,x_{L})-\Delta_{x_{L}}V_t(0,x_{L}-1)\\
&\leq \Delta_{x_{L}}V_t(0,x_{L})-\Delta_{x_{L}}V_t(0,x_{L}-1)\leq 0,
\end{align*}
where the first inequality is because $V_t$ satisfies property (\textit{iii}) and the second  inequality is because $V_t$ satisfies property (\textit{iv}).

If $x_{L}>0$, we have  
\begin{align*}
H_t(&x_{H}+2,x_{L}, 0, 1)-2H_t(x_{H}+1,x_{L},0, 1)+H_t(x_{H},x_{L},0, 1)\\
&=[-h(x_{H}+2+x_{L})+ r_{LL}+h+V_t (x_{H}+2, x_{L}-1)]-2[-h(x_{H}+1+x_{L})+ r_{LL}+h+V_t (x_{H}+1, x_{L}-1)]\\
&+[-h(x_{H}+x_{L})+ r_{LL}+h+V_t (x_{H}, x_{L}-1)]\\
&=V_t(x_{H}+2,x_{L}-1)-2V_t(x_{H}+1,x_{L}-1)+V_t(x_{H},x_{L}-1)\leq 0.
\end{align*}If $x_{L}=0$, we have  
\begin{align*}
H_t(&x_{H}+2,0, 0, 1)-2H_t(x_{H}+1,0,0, 1)+H_t(x_{H},0,0, 1)\\
&=[-h(x_{H}+2)+\max\{ r_{HL}+h+V_t (x_{H}+1, 0), V_t (x_{H}+2, 0)\}]\\
&-2[-h(x_{H}+1)+ \max\{ r_{HL}+h+V_t (x_{H}, 0), V_t (x_{H}+1, 0)\}]\\
&+[-h x_{H}+\max\{ r_{HL}+h+V_t (x_{H}-1, 0), V_t (x_{H}, 0)\}].\\
\end{align*}
When $\max\{ r_{HL}+h+V_t (x_{H}-1, 0), V_t (x_{H}, 0)\}=r_{HL}+h+V_t (x_{H}-1, 0)$, we have $$r_{HL}+h+V_t (x_{H}-1, 0)\geq V_t (x_{H}, 0), $$
i.e., $\Delta_{x_{H}}V_t(x_{H}-1,0)\leq r_{HL}+h$. Since $V_t$ satisfies property (\textit{iv}), $\Delta_{x_{H}}V_t(x_{H}+1,0)\leq \Delta_{x_{H}}V_t(x_{H}-1,0)\leq r_{HL}+h$. This implies that 
$$\max\{ r_{HL}+h+V_t (x_{H}+1, 0), V_t (x_{H}+2, 0)\}=r_{HL}+h+V_t (x_{H}+1, 0).$$ 
In this case, 
\begin{align*}
H_t(&x_{H}+2,0, 0, 1)-2H_t(x_{H}+1,0,0, 1)+H_t(x_{H},0,0, 1)\\
&\leq [-h(x_{H}+2)+ r_{HL}+h+V_t (x_{H}+1, 0)]-2[-h(x_{H}+1)+ r_{HL}+h+V_t (x_{H}, 0)]\\
&\;\;\;\;\;\;\;\;\;\;\;\;\;\;\;\;\;\;\;\;\;\;\;\;\;\;\;\;\;\;\;\;\;\;\;\;\;\;\;\;\;\;\;\;\;\;\;\;\;\;\;\;\;\;\;\;\;\;\;\;\;\;\;\;\;\;\;\;\;\;\;\;\;\;\;\;\;\;\;\;\;\;\;\;\;\;\;\;\;+[-hx_{H}+ r_{HL}+h+V_t (x_{H}-1,0)]\\
&=V_t(x_{H}+1,0)-2V_t(x_{H},0)+V_t(x_{H}-1,0)\leq 0.
\end{align*}
When $\max\{ r_{HL}+h+V_t (x_{H}-1, 0), V_t (x_{H}, 0)\}=V_t (x_{H}, 0)$, we consider two cases: (1) $\max\{ r_{HL}+h+V_t (x_{H}+1, 0), V_t (x_{H}+2, 0)\}=r_{HL}+h+V_t (x_{H}+1, 0)$ and (2) $\max\{ r_{HL}+h+V_t (x_{H}+1, 0), V_t (x_{H}+2, 0)\}=V_t (x_{H}+2, 0)$.
For the first case, we have
\begin{align*}
H_t(&x_{H}+2,0, 0, 1)-2H_t(x_{H}+1,0,0, 1)+H_t(x_{H},0,0, 1)\\
&\leq [-h(x_{H}+2)+ r_{HL}+h+V_t (x_{H}+1, 0)]-[-h(x_{H}+1)+ r_{HL}+h+V_t (x_{H}, 0)]\\
&\;\;\;\;-[-h(x_{H}+1)+ V_t (x_{H}+1, 0)]+[-hx_{H}+V_t (x_{H},0)]=0.
\end{align*}
For the second case, we have 
\begin{align*}
H_t(&x_{H}+2,0, 0, 1)-2H_t(x_{H}+1,0,0, 1)+H_t(x_{H},0,0, 1)\\
&\leq [-h(x_{H}+2)+ V_t (x_{H}+2, 0)]-2[-h(x_{H}+1)+ V_t (x_{H}+1, 0)]+[-hx_{H}+ V_t (x_{H},0)]\\
&=V_t(x_{H}+2,0)-2V_t(x_{H}+1,0)+V_t(x_{H},0)\leq 0.
\end{align*}

Therefore,  $H_t(x_{H},x_{L},y_{H}, y_{L})$ satisfies properties (\textit{Hi})-(\textit{Hiv}).

{\bf Step 2:} Next, we show  that  $V_{t+1}(x_{H}, x_{L})$ also satisfies properties (\textit{i})-(\textit{iv}). Recall that equation (\ref{appendix:dp.eq1}) can be written as $V_{t+1}(\mathbf{x})=\mathbb{E}H_t(\mathbf{x}+\mathbf{S}, \mathbf{D})$. We have shown that $H_t$ satisfies properties (\textit{Hi})-(\textit{Hiv}). Since the expectation operator preserves inequalities, it is straightforward that $V_{t+1}(x_{H}, x_{L})$ also satisfies properties (\textit{i})-(\textit{iv}).

Finally, we show that $d_{t}$ has the same structure as described in Proposition 1. 

Note that $d_{t}$ can be interpreted as the optimal decision in period $t$. In period $t$, there are only two possible types of demand agents arriving in each period, either an H-type demand agent or an L-type demand agent. We analyze these two cases separately. 

\noindent $\bullet$ Case 1: If the newly arrived demand agent is an H-type one, the centralized planner decides to match her with an H-type supply agent or an L-type supply agent if available, or not to match her with any supply agent.

Since $V_t(x_{H}, x_{L})$ satisfies properties (i)-(iii), we have 
\begin{align*}
r_{HH}+h+V_t(x_{H}-1, x_{L})&\ge V_t(x_{H}, x_{L}) &  \mbox{(by property (\textit{i}))}\\
r_{LH}+h+V_t(x_{H}, x_{L}-1)&\ge V_t(x_{H}, x_{L}) & \mbox{(by property (\textit{ii})  and $r_{LH}\geq r_{LL}$)}\\
r_{HH}+h+V_t(x_{H}-1, x_{L})&\ge r_{LH}+h+V_t(x_{H}, x_{L}-1) & \mbox{(by property (\textit{iii}))}
\end{align*}
This implies that 
\begin{align*}
&\max\{r_{HH}+h+V_t(x_{H}-1, x_{L}), r_{LH}+h+V_t(x_{H}, x_{L}-1), V_t(x_{H}, x_{L}) \}=r_{HH}+h+V_t(x_{H}-1, x_{L}), \\
&\max\{r_{HH}+h+V_t(x_{H}-1, 0), V_t(x_{H}, 0) \}=r_{HH}+h+V_t(x_{H}-1, 0), \\
&\max\{ r_{LH}+h+V_t(0, x_{L}-1), V_t(0, x_{L}) \}=r_{LH}+h+V_t(0, x_{L}-1). 
\end{align*}
Therefore, it is optimal to use an H-type supply agent to satisfy the H-type demand agent when both H-type and L-type supply agents are available, to use an H-type supply agent to satisfy the H-type demand agent when only H-type supply agents are available, and to use an L-type supply agent to satisfy the H-type demand agent when only L-type supply agents are available. 

In other words, for a newly arrived H-type demand agent, the optimal matching decision is to use an H-type supply agent to satisfy the H-type demand agent immediately if H-type supply agents are available. If H-type supply agents are not available, use an L-type supply agent to satisfy the H-type demand agent immediately. Since at least one supply agent will arrive on the platform during this period, the no-matching option is not optimal. We call this case {\it Greedy Matching}.

\noindent $\bullet$ Case 2: If the newly arrived demand agent is an L-type one, the centralized planner decides to match her with an H-type supply agent or an L-type supply agent if available, or not to match her with any supply agent.

Since $V_t(x_{H}, x_{L})$ satisfies properties (i)-(iii), we have 
\begin{align*}
r_{LL}+h+V_t(x_{H}, x_{L}-1)&\ge V_t(x_{H}, x_{L}) & \mbox{(by property (\textit{ii}))}\\
r_{LL}+h+V_t(x_{H}, x_{L}-1)&\ge r_{HL}+h+V_t(x_{H}-1, x_{L}) &  \mbox{(by property (\textit{iii}))}
\end{align*}
This implies that 
\begin{align*}
&\max\{r_{HL}+h+V_t(x_{H}-1, x_{L}), r_{LL}+h+V_t(x_{H}, x_{L}-1), V_t(x_{H}, x_{L}) \}=r_{LL}+h+V_t(x_{H}, x_{L}-1),\\ 
&\max\{r_{HL}+h+V_t(x_{H}-1, x_{L}), V_t(x_{H}, x_{L}) \}=r_{LL}+h+V_t(x_{H}, x_{L}-1),
\end{align*}
Therefore, it is optimal to use an L-type supply agent to satisfy the L-type demand agent whenever  L-type supply agents are available. When only H-type supply agents are available, the optimal decision is determined by the larger one between
$$\{r_{HL}+h+V_t(x_{H}-1, 0), V_t(x_{H}, 0)\}.$$
Since $V_t(x_{H}, x_{L})$ satisfies property (v), we have $V_t(x_{H}, 0)-[r_{HL}+h+V_t(x_{H}-1, 0)]$ is decreasing in $x_{H}$. Let 
$$k^{ce}_{t}=\min\{x_{H}>0: V_t(x_{H}, 0)-[r_{HL}+h+V_t(x_{H}-1, 0)]<0\}-1.$$
Then, when $x_{H}>k^{ce}_n$, $V_t(x_{H}, 0)-[r_{HL}+h+V_t(x_{H}-1, 0)]<0$, and when $x_{H}\leq k^{ce}_{t}$, $V_t(x_{H}, 0)-[r_{HL}+h+V_t(x_{H}-1, 0)]\geq 0$. That is, if the number of H-type supply agents is larger than $k^{ce}_{t}$, it is optimal to use an H-type supply agent to satisfy the L-type demand agent, and otherwise, not to match the L-type demand agent. We call this case as {\it Threshold Matching}.

Since  $\{d_{t}(\mathbf{s})\}_{\mathbf{s}\in \mathbf{S}}$ converges to the optimal policy $\pi^*$ as $t$ goes to infinity, the optimal policy $\pi^*$ has the same structure as $\{d_{t}(\mathbf{s})\}_{\mathbf{s}\in \mathbf{S}}$, which is described in Proposition 1, where the threshold $k^{ce}=\lim_{n\rightarrow \infty}k^{ce}_{t}$. 

$\Box$

\subsection*{C. Proof of Proposition \ref{cen: prop optimal threshold}}

Based on Theorem 1, the optimal matching policy $\pi^*$ has the following structure:
\begin{itemize}
	\item (Greedy Matching) When an H-type demand agent arrives, it is optimal to match her with an H-type supply agent if one is available. Otherwise, it is optimal to match her with an L-type supply agent.
	\item (Threshold Matching) When an L-type demand agent arrives, it is optimal to match her with an L-type supply agent if one is available. Otherwise, it is optimal to match her with an H-type supply agent only if the number of H-type supply agents is larger than $k^{ce}$.
\end{itemize}
In the following, we first solve the stationary distribution of steady states of the MDP and obtain the long-run average profit for a given threshold $k^{\pi}$. Then we optimize the long-run average profit over the threshold $k^{\pi}$ to obtain the optimal threshold $k^{ce}$.

For a given threshold $k^{\pi}$, let $(x_{H},x_{L})$ denote the number of H-type and L-type supplies at the beginning of each period before the arrival of the agents. Then, this problem can be formulated as a discrete-time Markov chain (DTMC) with state space $\{(x_{H}, x_{L})|x_{H}+x_{L}\leq k^{\pi}\}$. We can also plot the state transition diagram of this DTMC; see Figure \ref{fig.ce.MDP state trasition main}. The gray circles indicate the recurrent state set $\{(x_{H}, x_{L})|x_{H}+x_{L}=k^{\pi}\}$ and white circles indicate the transient state set $\{(x_{H}, x_{L})|x_{H}+x_{L}<k^{\pi}\}$.
\begin{figure}[htpb]
	\centering
	\includegraphics[width=0.5\textwidth]{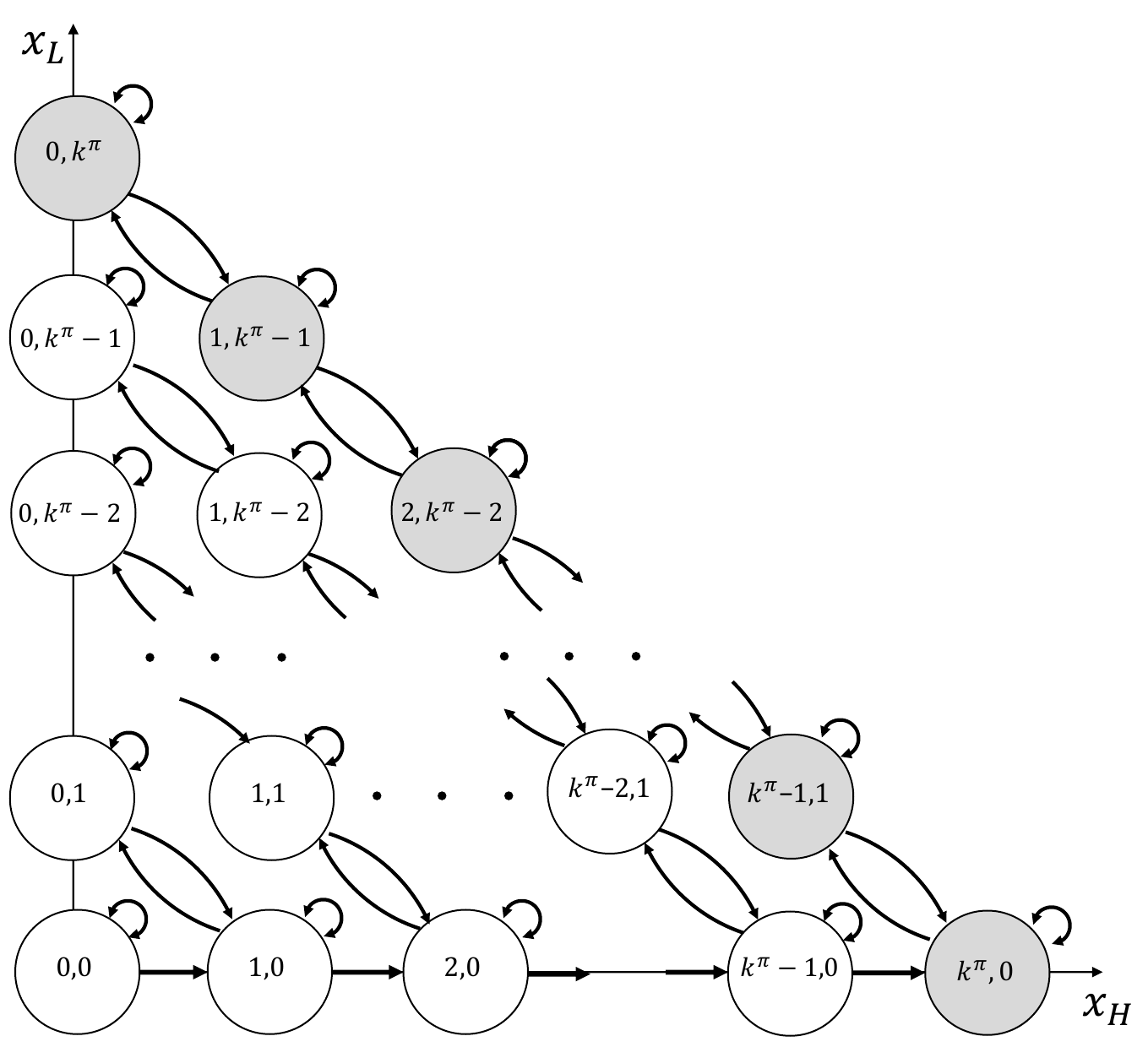}\\
	\caption{States transition diagram of centralized system}\label{fig.ce.MDP state trasition main}
\end{figure}

As shown in Figure \ref{fig.ce.MDP state trasition main}, there is only one communicating class in this DTMC. This communicating class is irreducible. Moreover, in this communicating class, each state is aperiodic.  Additionally, since all states in this communicating class are finite and positive recurrent, this Markov chain is stable. For this irreducible, aperiodic, and stable Markov chain, the steady state distribution exists.


Let $\phi(x_{H}, x_{L})$ represent the steady state probability for state $(x_{H}, x_{L})$. Clearly, $\phi(x_{H}, x_{L})=0$ for $x_{H}+x_{L}< k^{\pi}$. Moreover, $[\phi(k^{\pi}, 0),\phi(k^{\pi}-1, 1),\dots, \phi(0,k^{\pi})]$
is the unique solution to 
\begin{align*}
[\phi(k^{\pi}, 0),\phi(k^{\pi}-1, 1),\dots,\phi(0,k^{\pi})]\cdot P&=[\phi(k^{\pi}, 0),\phi(k^{\pi}-1, 1),\dots, \phi(0,k^{\pi})], \\
\phi(k^{\pi}, 0)+\phi(k^{\pi}-1, 1)+\dots+ \phi(0,k^{\pi})&=1,
\end{align*}
where 
\begin{eqnarray*}
	P=\left(
	\begin{array}{ccccccc}
		1-q(1-p) & q(1-p) & 0 & 0 & ...& 0& 0\\
		p(1-q) & pq+(1-p)(1-q) & q(1-p) & 0 & ...& 0& 0 \\
		0 &  	p(1-q) & pq+(1-p)(1-q) &q(1-p)  &  ...& 0& 0\\
		... & ... & ... & ... & ...&...& ...\\
		0 &  	0 & 0 &0  &  ...& pq+(1-p)(1-q)& q(1-p)\\
		0 &  	0 & 0 &0  &  ...& p(1-q)& 1-p(1-q)\\
	\end{array}
	\right)
\end{eqnarray*}
is the transition matrix. That is, 
\begin{align}
[1-q(1-p)]\phi(k^{\pi}, 0)+p(1-q)\phi(k^{\pi}-1, 1)&=\phi(k^{\pi}, 0), \label{lineareq1} \\
q(1-p)\phi(k^{\pi}, 0)+[pq+(1-p)(1-q)]\phi(k^{\pi}-1, 1)+p(1-q)\phi(k^{\pi}-2, 2)&=\phi(k^{\pi}-1, 1)\label{lineareq2}\\
\dots \\
q(1-p)\phi(1,k^{\pi}-1)+[1-p(1-q)]\phi(0,k^{\pi})&=\phi(0,k^{\pi}) \label{lineareqk}\\
\phi(k^{\pi}, 0)+\phi(k^{\pi}-1, 1)+\dots+ \phi(0,k^{\pi})&=1, \label{lineareqsum}
\end{align}

Let  $\delta=\frac{q(1-p)}{p(1-q)}$. By (\ref{lineareq1}), we have $\phi(k^{\pi}-1, 1)=\delta\phi(k^{\pi},0)$. Plugging it to (\ref{lineareq2}), we have $\phi(k^{\pi}-2, 2)=\delta^2\phi(k^{\pi},0)$. Similarly, $\phi(k^{\pi}-i, i)=\delta^i\phi(k^{\pi},0)$, for $i=3,\dots, k^{\pi}$. Then, by (\ref{lineareqsum}), 
$$\phi(k^{\pi}, 0)+\delta\phi(k^{\pi}, 0)+\delta^2\phi(k^{\pi}, 0)+\dots+ \delta^{k^{\pi}}\phi(k^{\pi}, 0)=1.$$

Next, we consider two cases: (1) $\delta\neq 1$ and (2) $\delta= 1$. 

(1) When $\delta\neq 1$, we have $$\phi(k^{\pi}, 0)=\frac{1}{1+\delta+\dots+\delta^{k^{\pi}}}=\frac{1-\delta}{1-\delta^{k^{\pi}+1}},$$ and 
$$\phi(k^{\pi}-i, i)=\delta^i\phi(k^{\pi}, 0)=\frac{(1-\delta)\delta^i}{1-\delta^{k^{\pi}+1}}, \quad \quad 0\le l\le k^{\pi}.$$

We can then compute the long-run average revenue $R(k^{\pi})$ and the long-run average waiting cost $C(k^{\pi})$ in the following: 
\begin{eqnarray*}
	R(k^{\pi})&=&\phi(k^{\pi}, 0)(pq r_{HH}+p(1-q)r_{HL}+(1-p)q r_{HH}+(1-p)(1-q)r_{LL})\\
	&&+\sum^{k^{\pi}-1}_{l=1}\phi(k^{\pi}-l, l)(pq r_{HH}+p(1-q)r_{LL}+(1-p)q r_{HH}+(1-p)(1-q)r_{LL})\\
	&&+\phi(0, k^{\pi})(pq r_{HH}+p(1-q)r_{LL}+(1-p)q r_{LH}+(1-p)(1-q)r_{LL})\\
	&=&\frac{1}{k^{\pi}+1}(pq r_{HH}+p(1-q)r_{HL}+(1-p)q r_{HH}+(1-p)(1-q)r_{LL})\\
	&&+\frac{k^{\pi}-1}{k^{\pi}+1}(pq r_{HH}+p(1-q)r_{LL}+(1-p)q r_{HH}+(1-p)(1-q)r_{LL})\\
	&&+\frac{1}{k^{\pi}+1}(pq r_{HH}+p(1-q)r_{LL}+(1-p)q r_{LH}+(1-p)(1-q)r_{LL})\\
	&=&pqr_{HH}+(1-p)(1-q)r_{LL}+p(1-q)\frac{(1-\delta)r_{HL}+(\delta-\delta^{k^{\pi}+1})r_{LL}}{1-\delta^{k^{\pi}+1}}\\
	&& +(1-p)q\frac{(1-\delta^{k^{\pi}})r_{HH}+(\delta^{k^{\pi}}-\delta^{k^{\pi}+1})r_{HL}}{1-\delta^{k^{\pi}+1}}\\
	&=& q r_{HH}+(1-q) r_{LL}-(1-p)q(r_{HH}-r_{LH}) \frac{\delta^{k^{\pi}}-\delta^{k^{\pi}+1}}{1-\delta^{k^{\pi}+1}}+p(1-q)(r_{HL}-r_{LL})\frac{1-\delta}{1-\delta^{k^{\pi}+1}}\\
	&=& q r_{HH}+(1-q) r_{LL}-p(1-q)(r_{HH}-r_{LH})(1-\delta) \frac{\delta^{k^{\pi}+1}}{1-\delta^{k^{\pi}+1}}+p(1-q)(r_{HL}-r_{LL})(1-\delta)\frac{1}{1-\delta^{k^{\pi}+1}}\\
	&=&q r_{HH}+(1-q) r_{LL}+p(1-q)(r_{HH}-r_{LH})(1-\delta)-p(1-q)(r_{HH}-r_{LH}-r_{HL}+r_{LL})\frac{1-\delta}{1-\delta^{k^{\pi}+1}},\\
	&=&q r_{HH}+(1-q) r_{LL}+p(1-q)(r_{HH}-r_{LH})(1-\delta)-p(1-q)r\frac{1-\delta}{1-\delta^{k^{\pi}+1}},\\
	C(k^{\pi})&=&\sum^{k^{\pi}}_{l=0}\phi(k^{\pi}-l, l) k^{\pi} h=k^{\pi} h,
\end{eqnarray*}
where $r=r_{HH}-r_{LH}-r_{HL}+r_{LL}$. Therefore, the long-run average profit is 
\begin{equation}\label{eq:W formula 1}
	W(k^{\pi})=q r_{HH}+(1-q) r_{LL}+p(1-q)(r_{HH}-r_{LH})(1-\delta)-p(1-q)r\frac{1-\delta}{1-\delta^{k^{\pi}+1}}-k^{\pi} h.
\end{equation}
It is easy to verify that $W(k^{\pi})$ is a concave function of $k^{\pi}$. Therefore, if $W(0)\geq W(1)$, i.e., 
$$h\geq p(1-q)r\frac{\delta}{1+\delta},$$
the optimal threshold is $k^{ce}=0$. Otherwise, $k^{ce}=\max\{k\geq 0:W(k)-W(k-1)\geq 0\}$.
Note that 
$$W(k)-W(k-1)=p(1-q)r(1-\delta)\frac{\delta^{k}-\delta^{k+1}}{(1-\delta^{k+1})(1-\delta^{k})}-h.$$
If $\delta<1$, i.e.,  $p>q$, $W(k)-W(k-1)\geq 0$ if and only if 
$$k\leq \frac{\ln[h(1+\delta)+p(1-q)r(1-\delta)^2-\sqrt{[h(1+\delta)+p(1-q)r(1-\delta)^2]^2-4h^2\delta}]-\ln(2h\delta)}{\ln\delta}.$$
In this case, $$k^{ce}=\left\lfloor\frac{\ln[h(1+\delta)+p(1-q)r(1-\delta)^2-\sqrt{[h(1+\delta)+p(1-q)r(1-\delta)^2]^2-4h^2\delta}]-\ln(2h\delta)}{\ln\delta}\right\rfloor.$$
If $\delta>1$, i.e.,  $p<q$, $W(k)-W(k-1)\geq 0$ if and only if 
$$k\leq \frac{\ln[h(1+\delta)+p(1-q)r(1-\delta)^2+\sqrt{[h(1+\delta)+p(1-q)r(1-\delta)^2]^2-4h^2\delta}]-\ln(2h\delta)}{\ln\delta}.$$
In this case, $$k^{ce}=\left\lfloor\frac{\ln[h(1+\delta)+p(1-q)r(1-\delta)^2+\sqrt{[h(1+\delta)+p(1-q)r(1-\delta)^2]^2-4h^2\delta}]-\ln(2h\delta)}{\ln\delta}\right\rfloor.$$

(2) When $\delta=1$, using a similar approach, we have 
\begin{equation}\label{eq:W formula 2}
	W(k^{\pi})=pr_{HH}+(1-p)r_{LL}-\frac{p(1-p)r}{k^{\pi}+1}-k^{\pi}h.
\end{equation} 
and  $$k^{ce}=\left\lfloor\frac{-h+\sqrt{h^2+4hp(1-q)r}}{2h}\right\rfloor.$$

$\Box$

\subsection*{D. Proof of Proposition \ref{prop: centralized welfare}}

The optimal social welfare is $W^{ce}=W(k^{ce})=\max\limits_{k^{\pi}= 0,1,\dots} W(k^{\pi})$, where the formula of $W$  is described by equations (\ref{eq:W formula 1}) and (\ref{eq:W formula 2}) for the cases of $\delta\neq 1$ and $\delta=1$, respectively. Since $W(k^{\pi})$ is decreasing in $h$, by the envelope theorem, $W^{ce}$ is also decreasing in $h$.  

$\Box$

\subsection*{E. Proofs of Lemma \ref{lemma:threshold property}, Proposition \ref{prop:matching action in equilibrium}, and Theorem \ref{thm:stable matching of decentralized system}}

In the following, we prove Lemma \ref{lemma:threshold property}, Proposition \ref{prop:matching action in equilibrium}, and Theorem \ref{thm:stable matching of decentralized system} together. 

We consider two scenarios: the arrival of an H-type demand agent and the arrival of an L-type demand agent. In the first scenario, when an H-type demand agent arrives, a match is naturally formed with the first H-type supply agent in the queue, as both parties are willing to engage in the match based on a first-come, first-served rule. If no H-type supply agents are available in the queue, the H-type demand agent will opt to match with the first L-type supply agent waiting in the L-type supply queue, as leaving the system unmatched would result in lower utility. Consequently, the first L-type supply agent is also willing to match with the H-type demand agent, since this match helps reduce holding costs and maximizes their payoff.

Next, we consider the scenario in which an L-type demand agent arrives. According to Lemma 2, the queue length of H-type supply agents is bounded by $k^{de}$.  This means that any $k$-th H-type supply agent (where $1\le k\le k^{de}$) is only willing to match with H-type demand agents and will reject matches with L-type demand agents. However, H-type supply agents ranked $k\ge k^{de}+1$ will be willing to match with the L-type demand agent.

Now, we will analyze the behavior of L-type supply agents waiting in the queue, starting with the first L-type supply agent at the front of the queue. As illustrated in Figure \ref{fig.dec.lost sales 1}, we denote the queue length of H-type supply agents as $z$, where $0\le z\le {k}^{de}$.

\begin{figure}[htpb]
	\centering
	\includegraphics[width=0.45\textwidth]{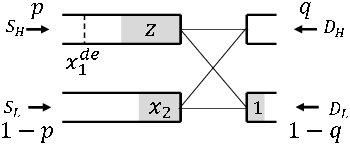}\\
	\caption{Example of a matching system with lost sales when an L-type demand agent arrives}\label{fig.dec.lost sales 1}
\end{figure}

Based on the first-come, first-served rule, L-type supply agents will be matched sequentially. We begin our analysis with the first L-type supply agent in the queue. Suppose an L-type demand agent is willing to match with this L-type supply agent. The first L-type supply agent must decide whether to match with her or continue waiting. If he chooses to match, the immediate payoff is $\alpha r_{LL}$. Conversely, if he decides to wait, his surplus will be the potential payoff $\alpha r_{LH}$ minus the expected waiting cost. Thus, the trade-off for the first L-type supply agent lies between the expected waiting benefit for a potentially superior match and the expected waiting cost associated with that decision. 

The expected waiting benefit is represented by the payoff improvement $\alpha(r_{LH}-r_{LL})$. The expected waiting cost is analyzed as follows.

{\bf Waiting cost of the first L-type supply agent}\\
Denote $EW_{L}^{1}(z)$ as the waiting cost of the first L-type supply agent when there are $z$ H-type supply agents waiting in their queue for $0\le z\le {k}^{de}$. Then, $EW_{L}^{1}(0)$ represents the queue of H-type supply agents is empty and $EW_{L}^{1}({k}^{de})$ represents the queue length of H-type supply agents is ${k}^{de}$. If the first L-type agent chooses to wait, there are three possible cases as follows.

\noindent $\bullet$ Case 1 with $z=0$: we have
\begin{eqnarray*}
	EW_{L}^{1}(0)=[p q+(1-p)(1-q)](EW_{L}^{1}(0)+h)+p(1-q)(EW_{L}^{1}(1)+h)+(1-p)q h
\end{eqnarray*}
It implies that
\begin{equation}\label{proof.eq1}
p(1-q)[EW_{L}^{1}(1)-EW_{L}^{1}(0)]-q(1-p)EW_{L}^{1}(0)+h=0
\end{equation}

\noindent $\bullet$ Case 2, with $1\le z< {k}^{de}$: we have
\begin{eqnarray*}
	EW_{L}^{1}(z)=[pq +(1-p)(1-q)](EW_{L}^{1}(z)+h)+p(1-p)(EW_{L}^{1}(z+1)+h)+q(1-p)(EW_{L}^{1}(z-1)+h)
\end{eqnarray*}
It implies that
\begin{equation}\label{proof.eq2}
p(1-q)[EW_{L}^{1}(z+1)-EW_{L}^{1}(z)]-q(1-p)[EW_{L}^{1}(z)-EW_{L}^{1}(z-1)]+h=0
\end{equation}

\noindent $\bullet$ Case 3 with $z={k}^{de}$: we have
\begin{eqnarray*}
	EW_{L}^{1}({k}^{de})=[pq+(1-p)(1-q)+p(1-p)](EW_{L}^{1}({k}^{de})+h)+q(1-p)(EW_{L}^{1}({k}^{de}-1)+h)
\end{eqnarray*}
It implies that
\begin{equation}\label{proof.eq3}
-q(1-p)[EW_{L}^{1}({k}^{de})-EW_{L}^{1}({k}^{de}-1)]+h=0
\end{equation}

We define $a(z)=EW_{L}^{1}(z)-EW_{L}^{1}(z-1)$ and $a(0)=EW_{L}^{1}(0)$ for $0\le z\le {k}^{de}$. Then, equations \eqref{proof.eq1}, \eqref{proof.eq2}, and \eqref{proof.eq3} can be rewritten as
\begin{eqnarray}\label{proof.eq4}
\left\{
\begin{array}{ll}
p(1-q)a(1)-q(1-p)a(0)+h=0, & \hbox{if $z=0$;} \\
p(1-q)a(z+1)-q(1-p)a(z)+h=0, & \hbox{if $0<z<{k}^{de}$;} \\
-q(1-p)a({k}^{de})+h=0, & \hbox{if $z={k}^{de}$.}
\end{array}
\right.
\end{eqnarray}
Solving\eqref{proof.eq4}, we have
\begin{eqnarray*}
	\left\{
	\begin{array}{ll}
		a(z)=\left(\frac{q(1-p)}{p(1-q)}\right)^z a(0)-\frac{h}{p(1-q)}\left(1+\frac{q(1-p)}{p(1-q)}+\left(\frac{q(1-p)}{p(1-q)}\right)^2+...+\left(\frac{q(1-p)}{p(1-q)}\right)^{z-1}\right),\\
		... \\
		a(3)=\frac{q(1-p)}{p(1-q)} a(2)-\frac{h}{p(1-q)}=\left(\frac{q(1-p)}{p(1-q)}\right)^3 a(0)-\frac{h}{p(1-q)}\left(1+\frac{q(1-p)}{p(1-q)}+\left(\frac{q(1-p)}{p(1-q)}\right)^2\right),  \\
		a(2)=\frac{q(1-p)}{p(1-q)} a(1)-\frac{h}{p(1-q)}=\left(\frac{q(1-p)}{p(1-q)}\right)^2 a(0)-\frac{h}{p(1-q)}\left(1+\frac{q(1-p)}{p(1-q)}\right),  \\
		a(1)=\frac{q(1-p)}{p(1-q)} a(0)-\frac{h}{p(1-q)},  \\
		a(0)=EW_{L}^{1}(0).
	\end{array}
	\right.
\end{eqnarray*}
To simplify the notation, we let $\delta=\frac{q(1-p)}{p(1-q)}$ where $0<\delta<+\infty$.

 If $\delta= 1$, we have $q(1-p)= p(1-q)$ and $p=q$. By the definition of $a(z)$ and the solution to \eqref{proof.eq4}, we have 
\begin{eqnarray*}
	\left\{
	\begin{array}{ll}
		a(z)=EW_{L}^{1}(z)-EW_{L}^{1}(z-1), \\
		a(z-1)=EW_{L}^{1}(z-1)-EW_{L}^{1}(z-2), \\
		..., & \\
		a(1)=EW_{L}^{1}(1)-EW_{L}(0),  \\
		a(0)=EW_{L}^{1}(0),
	\end{array}
	\right.\mbox{ and }
	\left\{
	\begin{array}{ll}
		a(z)= a(0)-\frac{zh}{p(1-q)},\\
		a(z-1)= a(0)-\frac{(z-1)h}{p(1-q)},  \\
		... \\
		a(1)= a(0)-\frac{h}{p(1-q)},  \\
		a(0).
	\end{array}
	\right.
\end{eqnarray*}
Hence, we can represent $EW_{L}^{1}(z)$ by $a(0)$ as follows
\begin{eqnarray}\label{proof.eq5}
EW_{L}^{1}(z)=\sum^{z}_{j=0}a(j)=(z+1)a(0)-\frac{z(z+1)h}{2p(1-q)}.
\end{eqnarray}
Plugging $EW_{L}^{1}(z)$ into equation \eqref{proof.eq3}, we have
\begin{eqnarray*}
	&&EW_{L}^{1}({k}^{de})-EW_{L}^{1}({k}^{de}-1)-\frac{h}{q(1-p)}\\
	&=&\left[({k}^{de}+1)a(0)-\frac{{k}^{de}({k}^{de}+1)h}{2p(1-q)}\right]
	-\left[{k}^{de}a(0)-\frac{({k}^{de}-1){k}^{de}h}{2p(1-q)}\right]
	-\frac{h}{q(1-p)}\\
	&=&a(0)-\frac{{k}^{de}h}{p(1-q)}-\frac{h}{q(1-p)}\\
	&=&EW_{L}^{1}(0)-\frac{{k}^{de}h}{p(1-q)}-\frac{h}{q(1-p)}=0.
\end{eqnarray*}
Thus, $EW_{L}^{1}(0)=\frac{{k}^{de}h}{p(1-q)}+\frac{h}{q(1-p)}$.

Recall that $${k}^{de}=\lfloor\frac{q\alpha(r_{HH}-r_{HL})}{h}\rfloor\geq \frac{q\alpha(r_{HH}-r_{HL})}{h}.$$ Thus,
\begin{eqnarray*}
	EW_{L}^{1}(0)\geq \frac{q\alpha(r_{HH}-r_{HL})}{p(1-q)}+\frac{h}{q(1-p)}>\frac{\alpha(r_{HH}-r_{HL})}{1-p}>\alpha(r_{HH}-r_{HL})\geq \alpha(r_{LH}-r_{LL}),
\end{eqnarray*} and 
where the second equality is because $\frac{q}{p(1-q)}=\frac{1}{1-p}$ when $\delta=1$ and $\frac{h}{q(1-p)}>0$ and the last equality is due to Assumption 2 that $r_{HH}-r_{HL}\geq r_{LH}-r_{LL}$. This implies that when $z=0$,  the first L-type supply agent prefers to match with the L-type demand agent immediately rather than waiting.

Next, we show that $EW_{L}^{1}(z)$ is a strictly increasing function of $z$ when $0\le z\le {k}^{de}$. When $z={k}^{de}$, by equation \eqref{proof.eq3}, we have $EW_{L}^{1}({k}^{de})-EW_{L}^{1}({k}^{de}-1)=\frac{h}{q(1-p)}>0$. It implies $EW_{L}^{1}({k}^{de})>EW_{L}^{1}({k}^{de}-1)$. When $0\le z<{k}^{de}$, 
\begin{eqnarray*}
	EW_{L}^{1}(z)-EW_{L}^{1}(z-1)&=&\left((z+1)EW_{L}^{1}(0)-\frac{z(z+1)h}{2 p(1-p)}\right)-\left(zEW_{L}^{1}(0)-\frac{(z-1)zh}{2 p(1-p)}\right)\\
	&=&EW_{L}^{1}(0)-\frac{hz}{p(1-p)}\\
	&>&EW_{L}^{1}(0)-\frac{h{k}^{de}}{p(1-p)}=\frac{h}{p(1-p)}>0.
\end{eqnarray*}
 Therefore, $EW_{L}^{1}(z)>EW_{L}^{1}(0)>\alpha (r_{LH}-r_{LL})$. This implies that the first L-type supply agent prefers to match with the L-type demand agent immediately, regardless of the queue length of H-type supply agents. The L-type supply agents who follow will wait longer than the first agent. Consequently, none of these agents will choose to wait for H-type supply agents either.

If $\delta\neq 1$, by the solution to \eqref{proof.eq4}, we have 
\begin{eqnarray*}
	\left\{
	\begin{array}{ll}
		a(z)=\delta^z a(0)-\frac{h}{p(1-q)}\left(1+\delta+\delta^2+...+\delta^{z-1}\right),\\
		... \\
		a(3)=\delta^3 a(0)-\frac{h}{p(1-q)}\left(1+\delta+\delta^2\right),  \\
		a(2)=\delta^2 a(0)-\frac{h}{p(1-q)}\left(1+\delta\right),  \\
		a(1)=\delta a(0)-\frac{h}{p(1-q)},  \\
		a(0)=EW_{L}^{1}(0),
	\end{array}
	\right.
	\Longrightarrow
	\left\{
	\begin{array}{ll}
		a(z)=\delta^z a(0)-\frac{h}{p(1-q)}\left(\frac{1-\delta^z}{1-\delta}\right),\\
		..., \\
		a(3)=\delta^3 a(0)-\frac{h}{p(1-q)}\left(\frac{1-\delta^3}{1-\delta}\right),  \\
		a(2)=\delta^2 a(0)-\frac{h}{p(1-q)}\left(\frac{1-\delta^2}{1-\delta}\right),  \\
		a(1)=\delta a(0)-\frac{h}{p(1-q)}\left(\frac{1-\delta}{1-\delta}\right),  \\
		a(0)=a(0)-\frac{h}{p(1-q)}\left(\frac{1-1}{1-\delta}\right).
	\end{array}
	\right.
\end{eqnarray*}
Summing over all the equations, we have 
$$\sum_{j=0}^{z}a(j)=\frac{1-\delta^{z+1}}{1-\delta}a(0)-\frac{h}{p(1-q)(1-\delta)}\left(z+1-\frac{1-\delta^{z+1}}{1-\delta}\right).$$
Recall the definition of $a(z)$:
\begin{eqnarray*}
	\left\{
	\begin{array}{ll}
		a(z)=EW_{L}^{1}(z)-EW_{L}^{1}(z-1), \\
		a(z-1)=EW_{L}^{1}(z-1)-EW_{L}^{1}(z-2), \\
		..., & \hbox{;} \\
		a(1)=EW_{L}^{1}(1)-EW_{L}^{1}(0),  \\
		a(0)=EW_{L}^{1}(0).
	\end{array}
	\right.
\end{eqnarray*}
Again summing over all above equations, we can represent $EW_{L}^{1}(z)$ by $a(0)$ as follows
\begin{eqnarray}\label{proof.eq.z}
EW_{L}^{1}(z)=\sum^{z}_{j=0}a(j)=\frac{1-\delta^{z+1}}{1-\delta}a(0)-\frac{h}{p(1-q)(1-\delta)}\left(z+1-\frac{1-\delta^{z+1}}{1-\delta}\right)
\end{eqnarray}
Plugging $EW_{L}^{1}(z)$ into equation \eqref{proof.eq3}, we have 
\begin{eqnarray*}
	&&EW_{L}^{1}({k}^{de})-EW_{L}^{1}({k}^{de}-1)-\frac{h}{q(1-p)}\\
	&=&\left[\frac{1-\delta^{{k}^{de}+1}}{1-\delta}\left(a(0)+\frac{h}{p(1-q)(1-\delta)}\right)-\frac{h({k}^{de}+1)}{p(1-q)(1-\delta)}\right]\\
	&&-\left[\frac{1-\delta^{{k}^{de}}}{1-\delta}\left(a(0)+\frac{h}{p(1-q)(1-\delta)}\right)-\frac{h{k}^{de}}{p(1-q)(1-\delta)}\right]-\frac{h}{q(1-p)}\\
	&=&\left[(1+\delta+...+\delta^{{k}^{de}-1}+\delta^{{k}^{de}})\left(a(0)+\frac{h}{p(1-q)(1-\delta)}\right)-\frac{h({k}^{de}+1)}{p(1-q)(1-\delta)}\right]\\
	&&-\left[(1+\delta+...+\delta^{{k}^{de}-1})\left(a(0)+\frac{h}{p(1-q)(1-\delta)}\right)-\frac{h{k}^{de}}{p(1-q)(1-\delta)}\right]-\frac{h}{q(1-p)}\\
	&=&\delta^{{k}^{de}}\left(a(0)+\frac{h}{p(1-q)(1-\delta)}\right)-\frac{h}{p(1-q)(1-\delta)}-\frac{h}{q(1-p)}\\
	&=&\delta^{{k}^{de}}a(0)-\frac{h(1-\delta^{{k}^{de}})}{p(1-q)(1-\delta)}-\frac{h}{q(1-p)}=0
\end{eqnarray*}
This allows us to solve $a(0)$ as
\begin{eqnarray}\label{proof.eq.a0}
EW_{L}^{1}(0)=a(0)&=&\frac{h}{\delta^{{k}^{de}}}\left(\frac{1}{p(1-q)}\left(\frac{1-\delta^{{k}^{de}}}{1-\delta}\right)+\frac{1}{q(1-p)}\right)\nonumber\\
&=&\frac{h}{q(1-p)\delta^{{k}^{de}}}\left(1+\frac{\delta(1-\delta^{{k}^{de}})}{1-\delta}\right)\nonumber\\
&=&\frac{h}{q(1-p)}\left(\frac{1-\delta^{{k}^{de}+1}}{\delta^{{k}^{de}}(1-\delta)}\right),
\end{eqnarray}
where the third equality is because $\frac{1}{p(1-q)}=\frac{\delta}{q(1-p)}$ since $\delta=\frac{q(1-p)}{p(1-q)}$. 

In the following, we consider two subcases: $0<\delta<1$ and $\delta>1$.

\noindent $\bullet$ Subcase 1 with $0<\delta<1$: in this case, we have $p>q$ and
\begin{eqnarray*}
	\frac{1-\delta^{{k}^{de}+1}}{\delta^{{k}^{de}}(1-\delta)}
	=\frac{1+\delta+\delta^2+...+\delta^{{k}^{de}}}{\delta^{{k}^{de}}}\ge {k}^{de},
\end{eqnarray*}
where the inequality is due to the fact that $\delta<1$.
Note that $$ {k}^{de}=\lfloor\frac{q\alpha(r_{HH}-r_{HL})}{h}\rfloor\geq \frac{q\alpha(r_{HH}-r_{HL})}{h}.$$ Thus, we have
\begin{eqnarray*}
	EW_{L}^{1}(0)=\frac{h}{q(1-p)}\left(\frac{1-\delta^{{k}^{de}+1}}{\delta^{{k}^{de}}(1-\delta)}\right)\ge \frac{{k}^{de}h}{q(1-p)}\geq \frac{\alpha(r_{HH}-r_{HL})}{1-p}>\alpha(r_{HH}-r_{HL})\geq \alpha(r_{LH}-r_{LL}).
\end{eqnarray*}
This indicates that when $z=0$, the first L-type supply agent prefers to match with the L-type demand agent immediately rather than wait.

Next, we show that $EW_{L}^{1}(z)$ is a strictly increasing function of $z$ when $0\le z\le {k}^{de}$. When $z={k}^{de}$, by equation \eqref{proof.eq3}, we have $EW_{L}^{1}({k}^{de})-EW_{L}^{1}({k}^{de}-1)=\frac{h}{q(1-p)}>0$. It implies $EW_{L}^{1}({k}^{de})>EW_{L}^{1}({k}^{de}-1)$. When $0\le z<{k}^{de}$, 
\begin{eqnarray}
&& EW_{L}^{1}(z)-EW_{L}^{1}(z-1)\nonumber\\
&=&\left(\frac{1-\delta^{z+1}}{1-\delta}a(0)-\frac{h}{p(1-q)(1-\delta)}\left(z+1-\frac{1-\delta^{z+1}}{1-\delta}\right)\right)\nonumber\\
&&-\left(\frac{1-\delta^{z}}{1-\delta}a(0)-\frac{h}{p(1-q)(1-\delta)}\left(z-\frac{1-\delta^{z}}{1-\delta}\right)\right)\nonumber\\
&=&\left(\frac{1-\delta^{z+1}}{1-\delta}-\frac{1-\delta^{z}}{1-\delta}\right)\left(a(0)+\frac{h}{p(1-q)(1-\delta)}\right)-\frac{h}{p(1-q)(1-\delta)}\nonumber\\
&=&\delta^{z}\left(a(0)+\frac{h}{p(1-q)(1-\delta)}\right)-\frac{h}{p(1-q)(1-\delta)}\nonumber\\
&=&\delta^{z}a(0)-\frac{(1-\delta^{z})h}{p(1-q)(1-\delta)}\nonumber\\
&=&\delta^{z}\frac{h}{q(1-p)}\left(\frac{1-\delta^{{k}^{de}+1}}{\delta^{{k}^{de}}(1-\delta)}\right)-\frac{(1-\delta^{z})h}{p(1-q)(1-\delta)}\nonumber\\
&=&\frac{h}{q(1-p)}\left(\frac{1-\delta^{{k}^{de}+1}}{1-\delta}\right)\frac{\delta^{z}}{\delta^{{k}^{de}}}-\frac{h}{p(1-q)}\left(\frac{1-\delta^{z}}{1-\delta}\right)\nonumber\\
&=&\frac{h}{q(1-p)}\left(1+\delta+\delta^2+...+\delta^{{k}^{de}}\right)\frac{\delta^{z}}{\delta^{{k}^{de}}}-\frac{h}{q(1-p)}\left(\delta+\delta^2+...+\delta^{z}\right)\nonumber\\
&=&\frac{h}{q(1-p)}\left(\delta^{z-{k}^{de}}+\delta^{z-{k}^{de}+1}+\dots+1\right)\nonumber\\
&>&\frac{h}{q(1-p)} \ge 0. \label{proof.eq.ewz}
\end{eqnarray}
Therefore, $EW_{L}^{1}(z)>EW_{L}^{1}(0)>\alpha (r_{LH}-r_{LL})$ for all $0\le z\le {k}^{de}$.  That is, none of the L-type supply agents will choose to wait for H-type supply agents either.

\noindent $\bullet$ Subcase 2 with $\delta>1$,: in this case, we have $p<q$. 

Denote  $EW_L^{k}(z)$ as the expected waiting cost of the $k$-th L-type supply agent in the queue when there are $z$ H-type supply agents present. An L-type supply agent can only be matched if there are no H-type supply agents in the market. Therefore, the  $k$-th L-type supply agent has to wait for an expected time of $EW_L^{1}(z)$ due to the presence of $z$  H-type supply agents. In addition, the $k-1$ preceding L-type supply agents each have to wait as the first in line in the queue, which incurs a cost of $EW_L^{1}(0)$. That is,
\begin{eqnarray}\label{proof.eq.kz}
EW_L^{k}(z)&=&(k-1)EW_L^{1}(0)+EW_L^{1}(z)\\
&=&(k-1)EW_L^{1}(0)+\frac{1-\delta^{z+1}}{1-\delta}EW_L^{1}(0)-\frac{h}{p(1-q)(1-\delta)}\left(z+1-\frac{1-\delta^{z+1}}{1-\delta}\right)\nonumber\\
&=&\left(k-1+\frac{1-\delta^{z+1}}{1-\delta}\right)\frac{h}{q(1-p)}\left(\frac{1-\delta^{{k}^{de}+1}}{\delta^{{k}^{de}}(1-\delta)}\right)
-\frac{h}{p(1-q)(1-\delta)}\left(z+1-\frac{1-\delta^{z+1}}{1-\delta}\right).\nonumber
\end{eqnarray}
The $k$-th L-type supply agent will choose to wait if and only if $EW_L^{k}(z)\le \alpha (r_{LH}-r_{LL})$. Note that $EW_L^{k}(z)$ strictly increases in $k$. As a result, there exists a threshold for the queue length, denoted as $k_L(z)$, such that L-type supply agents will wait in the queue if they are positioned up to the $k_L(z)$-th L-type supply agent. The value of $k_L(z)$ is equal to the largest integer $k$ that satisfies $EW_L^{k}(z)\le \alpha(r_{LH}-r_{LL})$.  This threshold is unique due to the monotonicity of  $EW_L^{k}(z)$ with respect to $k$.

Thus, we can derive the expression of $k_L(z)$ as follows
\begin{eqnarray*}
	 k_L(z)=\bigg\lfloor\frac{\alpha(r_{LH}-r_{LL})+\frac{h}{p(1-q)(1-\delta)}\left(z+1-\frac{1-\delta^{z+1}}{1-\delta}\right)+\left(1-\frac{1-\delta^{z+1}}{1-\delta}\right)\frac{h}{q(1-p)}\left(\frac{1-\delta^{{k}^{de}+1}}{\delta^{{k}^{de}}(1-\delta)}\right)}
	{\left(\frac{1-\delta^{z+1}}{1-\delta}-1\right)\frac{h}{q(1-p)}\left(\frac{1-\delta^{{k}^{de}+1}}{\delta^{{k}^{de}}(1-\delta)}\right)}\bigg\rfloor.
\end{eqnarray*}

Since $\delta>1$, it is easy to verify that 
$$EW_L^{k}(z+1)-EW_L^{k}(z)=\delta^{z+1}\frac{h}{q(1-p)}\left(\frac{1-\delta^{{k}^{de}+1}}{\delta^{{k}^{de}}(1-\delta)}\right)-\frac{h}{p(1-q)(1-\delta)}(1-\delta^{z+1})>0.$$
Therefore, $EW_L^{k}(z)$ increases in $z$. As a consequence, $k_L(z)$ decreases in $z$. 

Now, we show that $ k_L(z)\leq {k}^{de}-z$. Note that
$$EW_L^{1}(0)=\frac{h}{q(1-p)}\left(\frac{1-\delta^{{k}^{de}+1}}{\delta^{{k}^{de}}(1-\delta)}\right)\geq \frac{h}{q(1-p)},$$
where the inequality is due to the fact that $$\frac{1-\delta^{{k}^{de}+1}}{\delta^{{k}^{de}}(1-\delta)}=\frac{1+\delta+\dots+\delta^{{k}^{de}}}{\delta^{{k}^{de}}}\geq 1.$$
Therefore, 
$$EW_L^{{k}^{de}}(0)={k}^{de}EW_L^{1}(0)\geq \frac{{k}^{de}h}{q(1-p)}=\frac{\alpha(r_{HH}-r_{HL})}{1-p}>\alpha(r_{LH}-r_{LL}),$$
which implies that $ k_L(0)\leq {k}^{de}$. 

By equations \eqref{proof.eq3} and (\ref{proof.eq.ewz}), we have $EW_{L}^{1}(z)-EW_{L}^{1}(z-1)\geq \frac{h}{q(1-p)}$ for all $z=1,\dots, {k}^{de}$. Thus,
$$EW_{L}^{1}(z)\geq EW_{L}^{1}(0)+z\frac{h}{q(1-p)}\geq (z+1)\frac{h}{q(1-p)}.$$
As a consequence, 
\begin{align*}
EW_L^{{k}^{de}-z}(z)&=({k}^{de}-z-1)EW_L^{1}(0)+EW_L^{1}(z)\\
&\geq ({k}^{de}-z-1)\frac{h}{q(1-p)}+(z+1)\frac{h}{q(1-p)}\\
&={k}^{de}\frac{h}{q(1-p)}\\
&=\frac{\alpha(r_{HH}-r_{HL})}{1-p}>\alpha(r_{LH}-r_{LL}).
\end{align*}
This implies that $ k_L(z)\leq {k}^{de}-z$. 

To summarize, in the equilibrium, if an H-type demand agent arrives, all supply agents choose to match with her. The H-type demand agent will select an H-type supply agent if one is available; otherwise, she will match with an L-type supply agent. If an L-type demand agent arrives, the first $k^{de}$ H-type supply agents in the waiting line will choose to wait and match with a future H-type demand agent, while the remaining H-type supply agents will match with her. Similarly, the first $k_L(x_H)$ L-type supply agents in the waiting line will opt to wait for a future H-type demand agent, while the other L-type supply agents will match with the L-type demand agent. The L-type demand agent will choose to match with an H-type supply agent if the number of H-type supply agents present exceeds $k^{de}$; otherwise, she will match with an L-type supply agent. Moreover, $k_L(z)$ is decreasing in $z$ and $ k_L(z)\leq {k}^{de}-z$ with $k_L(z)=0$ when $p\geq q$. The corresponding matching outcome in equilibrium is described in Theorem \ref{thm:stable matching of decentralized system}.

Finally, note that any welfare-maximizing matching policy would match an arrived H-type demand agent with an H-type supply agent if there is any. Otherwise, it would match it with an L-type supply agent. It is easy to see that the above equilibrium is the only equilibrium that satisfies these properties. Thus, it is welfare maximizing.

$\Box$

\subsection*{F. Proofs of Propositions \ref{prop: steady states in equilibrium} and \ref{prop:decentralized sw}}
We consider two cases: (1) $p\ge q$, and (2) $p<q$. 
\begin{figure}[ht]
	\centering
	\subfigure[$p\ge q$]{
		\label{fig: larger p}     \includegraphics[width=0.46\textwidth]{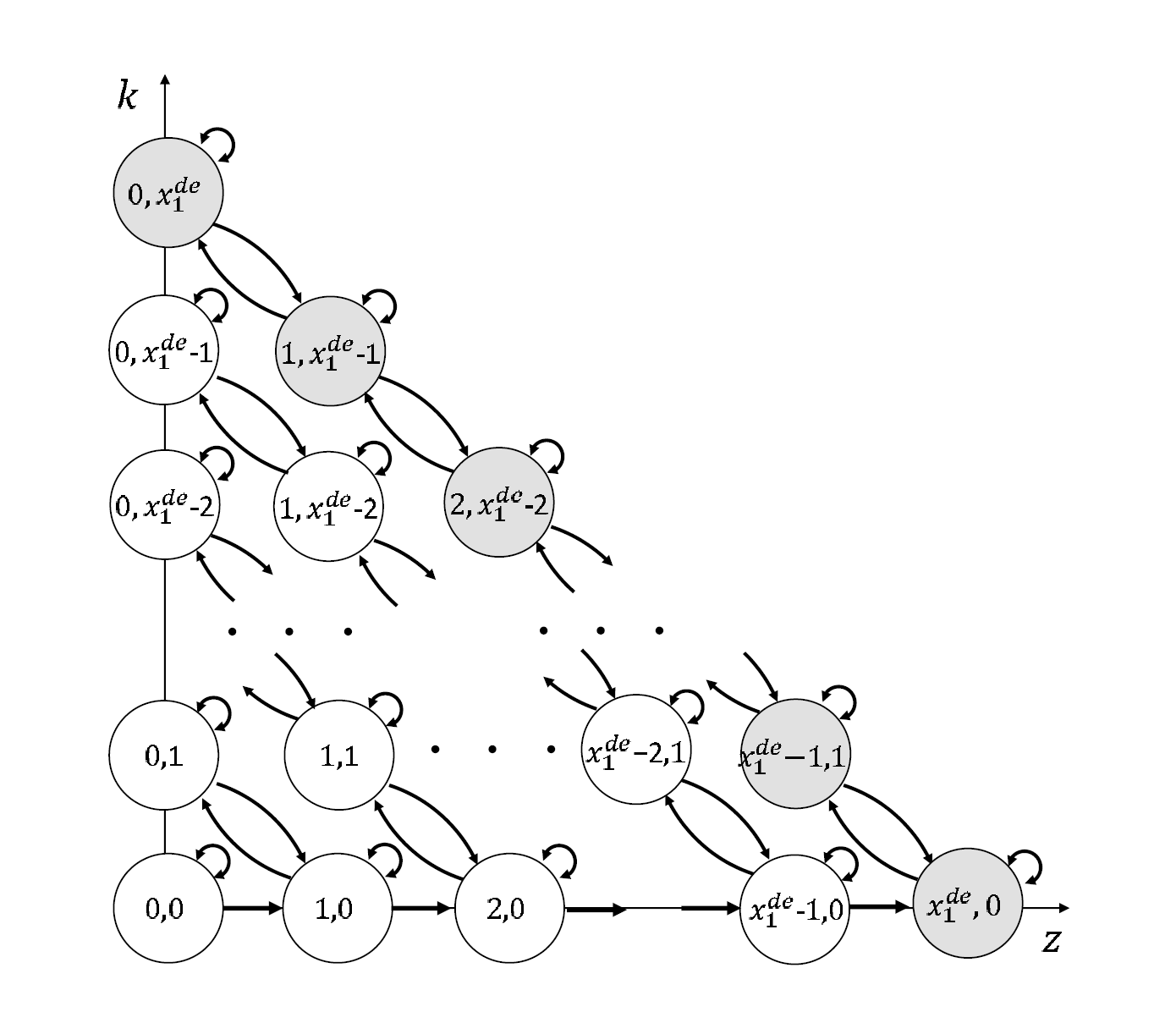}}
	\hspace{0.2in}
	\subfigure[$p<q$]{
		\label{fig: larger q}     \includegraphics[width=0.46\textwidth]{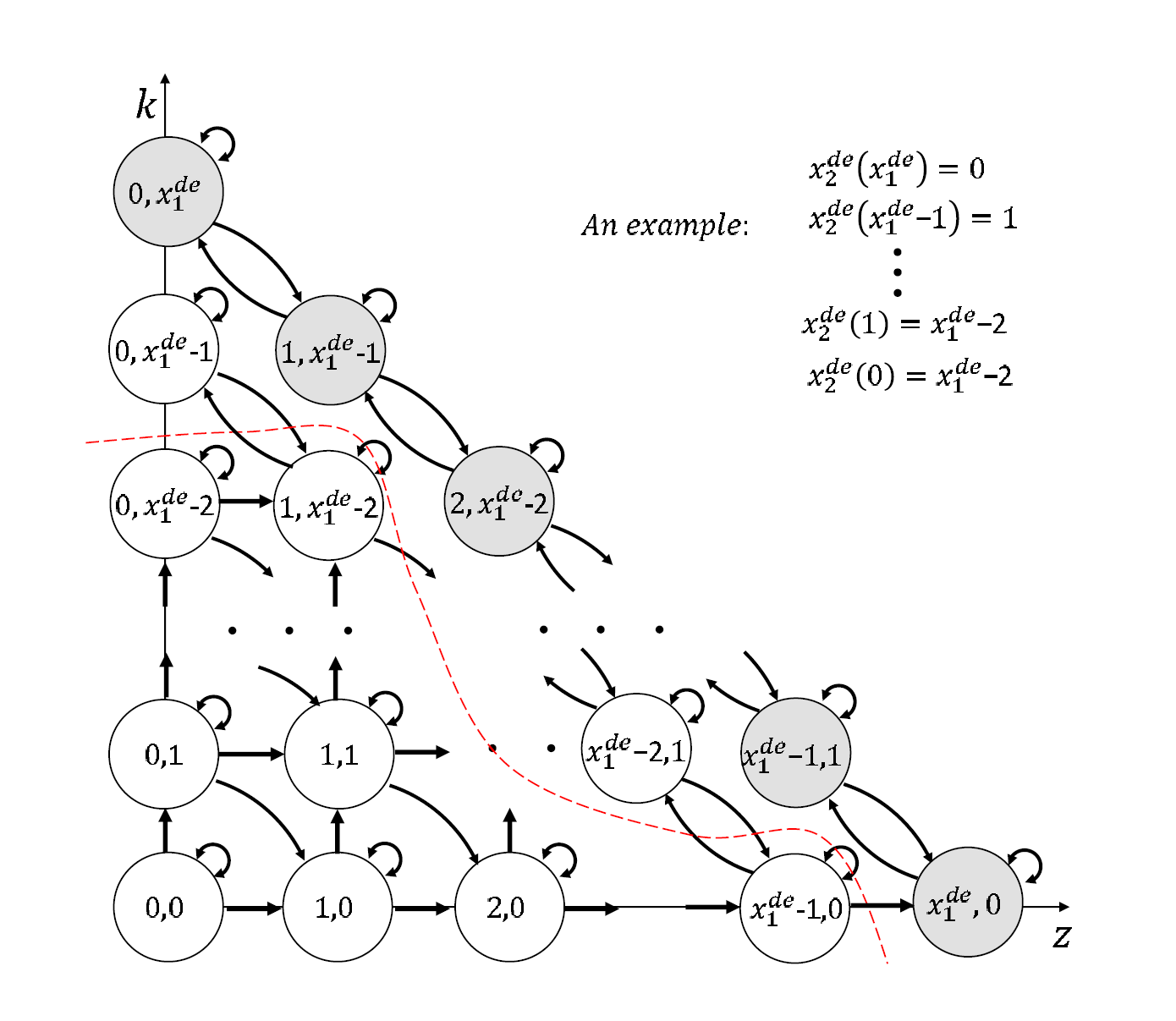}}
	\caption{State transition diagrams of L-type supply agents}
	\label{fig:state transit}
\end{figure}

In the first case, the L-type supply agent will choose to match with either an H-type or an L-type demand agent. Therefore, we can illustrate the transition diagram in Figure \ref{fig: larger p}. It is evident that the only recurrent states are $\{({k}^{de},0), ({k}^{de},1), ({k}^{de}-2,2), \dots, (0,{k}^{de})\}$.

In this second case, there exists a unique position threshold $ k_L(x_{H})$ for every $x_{H}\leq  k^{de}$ such that  L-type agents positioned at $k$ with $k\le  k_L(x_{H})$  will insist on waiting for H-type demand agents, while those positioned at $k$ with $k> k_L(x_{H})$ will accept either H-type or L-type demand agents. Again, we can represent this scenario with a transition diagram; see  Figure \ref{fig: larger q} for an example where $ k_L({k}^{de})=0$, $ k_L({k}^{de}-1)=1$, $\cdots$, $ k_L(1)={k}^{de}-2$, and $ k_L(0)={k}^{de}-2$. In this case, we again observe that the only recurrent states are  $\{({k}^{de},0), ({k}^{de},1), ({k}^{de}-2,2), \dots, (0,{k}^{de})\}$.

Therefore, in both cases, the Markov chain is reducible and there is only one closed communicating class  $\{({k}^{de},0), ({k}^{de},1), ({k}^{de}-2,2), \dots, (0,{k}^{de})\}$.

\begin{figure}[htpb]
	\centering
	\includegraphics[width=0.8\textwidth]{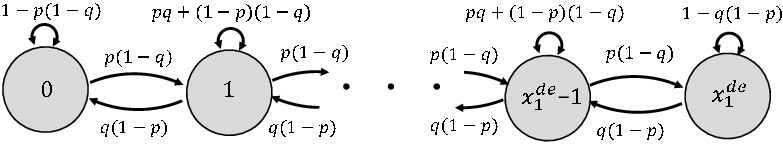}\\
	\caption{State transition diagram of H-type supply agents}\label{fig.dec.lost sales H type queue state transition}
\end{figure}

Now, we can compute the steady state distribution for those recurrent states. We again illustrate the transition diagram for these recurrent states in Figure \ref{fig.dec.lost sales H type queue state transition}. The balance equations are  
\begin{eqnarray*}
	\phi^{de}({k}^{de},0)&=&p(1-q)\phi^{de}({k}^{de}-1,1)+(1-q(1-p))\phi^{de}({k}^{de},0),\\
	\phi^{de}({k}^{de}-1,1)&=&p(1-q)\phi^{de}({k}^{de},0)+[p q+(1-p)(1-q)]({k}^{de}-1,1)+q(1-p)\phi^{de}({k}^{de}-2,2),\\
	\phi^{de}({k}^{de}-2,2)&=&p(1-q)\phi^{de}({k}^{de}-1,1)+[p q+(1-p)(1-q)]({k}^{de}-2,2)+q(1-p)\phi^{de}({k}^{de}-3,3),\\
	\dots\\
	\phi^{de}(1,{k}^{de}-1)&=&p(1-q)\phi^{de}(2,{k}^{de}-2)+[p q+(1-p)(1-q)](1,{k}^{de}-1)+q(1-p)\phi^{de}(0,{k}^{de}),\\
	\phi^{de}(0, {k}^{de})&=&[1-p(1-q)]\phi^{de}(0, {k}^{de}) +q(1-p)\phi^{de}(1, {k}^{de}-1),\\
	\phi^{de}({k}^{de},0)&+ &\phi^{de}({k}^{de},1)+ \dots+ \phi^{de}(0,{k}^{de})=1.
\end{eqnarray*}
To solve the above equations and obtain the steady state distribution, we again consider two cases: (1) $\delta\neq 1$ and (2) $\delta=1$. 

When $\delta\neq 1$, the solution to the above equations provides us with the steady state distribution
$$\phi^{de}({k}^{de}-i,i)=\frac{(1-\delta)\delta^i}{1-\delta^{ k^{de}+1}}, \quad \quad 0\le l\le  k^{de}$$
and social welfare
$$W^{de}=W({k}^{de})=qr_{HH}+(1-q)r_{LL}+p(1-q)(r_{HH}-r_{LL})(1-\delta)-p(1-q)r\frac{1-\delta}{1-\delta^{k^{de}+1}}-k^{de}h.$$

When $\delta= 1$, the solution to the above equations provides us with the steady state distribution
$$\phi^{de}({k}^{de}-i,i)=\frac{1}{k^{de}+1}, \quad \quad 0\le l\le  k^{de}$$
and social welfare
$$W^{de}=W({k}^{de})=pr_{HH}+(1-p)r_{LL}-\frac{p(1-p)r}{k^{de}+1}-k^{de}h.$$

$\Box$

\subsection*{F. Proofs of Propositions \ref{prop:threshold comparison} and \ref{prop: coordination}}

Note that ${k}^{de}=\left\lfloor \frac{q \alpha (r_{HH}-r_{HL})}{h}\right\rfloor$ and $k^{ce}$ are integers that are independent of $\alpha$. Then, $k^{de}<k^{ce}$ if and only if $\frac{q \alpha (r_{HH}-r_{HL})}{h}< k^{ce}$, i.e., $\alpha<\frac{hk^{ce}}{q(r_{HH}-r_{HL})}$. Similarly, $k^{de}>k^{ce}$ if and only if $\frac{q \alpha (r_{HH}-r_{HL})}{h}\geq k^{ce}+1$, i.e.,$\alpha\geq \frac{h(k^{ce}+1)}{q(r_{HH}-r_{HL})}$. When $\alpha\in[\frac{hk^{ce}}{q(r_{HH}-r_{HL})}, \frac{h(k^{ce}+1)}{q(r_{HH}-r_{HL})})$, $ \frac{q \alpha (r_{HH}-r_{HL})}{h}\in [k^{ce}, k^{ce}+1)$, i.e.,  $k^{de}=k^{ce}$. In this case, the centralized and decentralized settings are identical. 

$\Box$

\subsection*{G. Proofs of Proposition \ref{prop:comparison in centralized setting}}
Note that $k^{ce}_{FB}=\left\lfloor\sqrt{\frac{p(1-p)r}{2h}}\right\rfloor$ and $k^{ce}=\left\lfloor\frac{-h+\sqrt{h^2+4hp(1-p)r}}{2h}\right\rfloor$ when $p=q$. It is easy to see that when $p(1-p)r<2h$, $k^{ce}_{FB}=0$.  In this case, we have 
\begin{align*}
    p(1-p)r<2h & \iff 2h\sqrt{2h p(1-p)r}-2h p(1-p)r>0\\
     & \iff h^2+2h\sqrt{2h p(1-p)r}+2h p(1-p)r>h^2+4h p(1-p)r \\
     & \iff h+\sqrt{2h p(1-p)r}>\sqrt{h^2+4h p(1-p)r}\\
     &\iff \sqrt{\frac{p(1-p)r}{2h}}>\frac{-h+\sqrt{h^2+4hp(1-p)r}}{2h}.
\end{align*}
This implies that $k^{ce}\leq k^{ce}_{FB}=0$, i.e., $k^{ce}=0$.

Similarly, when $p(1-p)r\geq 2h$, we have $k^{ce}\geq k^{ce}_{FB}\geq 1$.  In this case, we also have 
\begin{align*}
    p(1-p)r\geq 2h & \iff 4(p(1-p)r)^2-8hp(1-p)r\geq 0\\
     & \iff 4(p(1-p)r)^2-4hp(1-p)r+h^2\geq h^2+4hp(1-p)r\\
     & \iff 2p(1-p)r-h\geq\sqrt{h^2+4h p(1-p)r}\\
     &\iff \frac{p(1-p)r}{h}\geq \frac{-h+\sqrt{h^2+4hp(1-p)r}}{2h}+1.
\end{align*}
This implies that $\frac{p(1-p)r}{h}\geq k^{ce}+1$. Moreover, 
\begin{align*}
  p(1-p)r\geq 2h & \iff -4h\sqrt{2h p(1-p)r}+4h p(1-p)r\geq 0\\
     & \iff h^2-4h\sqrt{2h p(1-p)r}+8h p(1-p)r\geq h^2+4h p(1-p)r \\
     & \iff 2\sqrt{2h p(1-p)r}-h\geq \sqrt{h^2+4h p(1-p)r}\\
     &\iff 2\sqrt{\frac{p(1-p)r}{2h}}\geq \frac{-h+\sqrt{h^2+4hp(1-p)r}}{2h}+1.
\end{align*}
This implies that $2k^{ce}_{FB}\geq k^{ce}$.

Note that 
\begin{align*}
    W^{ce}_{FB}&=pr_{HH}+(1-p)r_{LL}-\frac{p(1-p)r}{2k^{ce}_{FB}+1}-\frac{2k^{ce}_{PA}(k^{ce}_{FB}+1)}{2k^{ce}_{FB}+1}h,\\
    W^{ce}_{OB}&=pr_{HH}+(1-p)r_{LL}-\frac{p(1-p)r}{k^{ce}+1}-k^{ce}h, \\
    W_{IA}^{ce}&=pr_{HH}+(1-p)r_{LL}-p(1-p)r.
\end{align*}
Thus, 
\begin{align*}
W^{ce}_{FB}- W^{ce}_{OB}&=\frac{2k^{ce}_{FB}-k^{ce}}{(k^{ce}+1)(2k^{ce}_{FB}+1)}p(1-p)r-\frac{2(k^{ce}_{FB})^2-2k^{ce}_{FB}k^{ce}+2k^{ce}_{FB}-k^{ce}}{2k^{ce}_{FB}+1}h,\\
&=\frac{2k^{ce}_{FB}-k^{ce}}{2k^{ce}_{FB}+1}(\frac{p(1-p)r}{k^{ce}+1}-h)+\frac{2k^{ce}_{FB}(k^{ce}-k^{ce}_{FB})}{2k^{ce}_{FB}+1}h.
\end{align*}
When $p(1-p)r<2h$, we have $k^{ce}= k^{ce}_{FB}=0$ and $W^{ce}_{FB}- W^{ce}_{OB}=0$. When $p(1-p)r\geq 2h$, $k^{ce}\geq k^{ce}_{FB}$, $\frac{p(1-p)r}{h}\geq k^{ce}+1$, and $k^{ce}\leq 2k^{ce}_{FB}$, which implies $W^{ce}_{FB}- W^{ce}_{OB}\geq 0$. Thus, $W^{ce}_{FB}\geq  W^{ce}_{OB}$.

Similarly, 
\begin{align*}
W^{ce}_{OB}-W^{ce}_{NB}&=\frac{k^{ce}}{k^{ce}+1}p(1-p)r-k^{ce}h,
\end{align*}
which is 0 when $p(1-p)r< 2h$ and larger than or equal to 0 when $p(1-p)r\geq 2h$ since $\frac{p(1-p)r}{h}\geq k^{ce}+1$ in this case. Therefore, $W^{ce}_{OB}\geq W^{ce}_{NB}.$

Finally,
\begin{align*}
(W^{ce}_{OB}-W^{ce}_{NB})-(W^{ce}_{FB}- W^{ce}_{OB})&=(k^{ce}-\frac{2k^{ce}_{FB}-k^{ce}}{2k^{ce}_{FB}+1})(\frac{p(1-p)r}{k^{ce}+1}-h)-\frac{2k^{ce}_{FB}(k^{ce}-k^{ce}_{FB})}{2k^{ce}_{FB}+1}h.
\end{align*}
When $p(1-p)r<2h$, we have $k^{ce}= k^{ce}_{FB}=0$ and thus $(W^{ce}_{OB}-W^{ce}_{NB})-(W^{ce}_{FB}- W^{ce}_{OB})=0$.
When $2h\leq p(1-p)r<6h$, it is easy to verify that $k^{ce}=k^{de}=1$ and $(W^{ce}_{OB}-W^{ce}_{NB})-(W^{ce}_{FB}- W^{ce}_{OB})=0$.
When $p(1-p)r<2h$, we have
\begin{align*}
    p(1-p)r\geq 6h & \iff (p(1-p)r)^2-6hp(1-p)r\geq 0\\
     & \iff (p(1-p)r)^2-2hp(1-p)r+h^2\geq h^2+4hp(1-p)r\\
     & \iff p(1-p)r-h\geq\sqrt{h^2+4h p(1-p)r}\\
     &\iff \frac{p(1-p)r}{2h}\geq \frac{-h+\sqrt{h^2+4hp(1-p)r}}{2h}+1.
\end{align*}
This implies that $\frac{p(1-p)r}{2h}\geq k^{ce}+1$ and $$\frac{p(1-p)r}{k^{ce}+1}-h\geq h.$$ Therefore, 
\begin{align*}
(W^{ce}_{OB}-W^{ce}_{NB})-(W^{ce}_{FB}- W^{ce}_{OB})&=\frac{2k^{ce}_{FB}k^{ce}-2k^{ce}_{FB}+2k^{ce}}{2k^{ce}_{FB}+1}(\frac{p(1-p)r}{k^{ce}+1}-h)-\frac{2k^{ce}_{FB}(k^{ce}-k^{ce}_{FB})}{2k^{ce}_{FB}+1}h \\
&\geq \frac{2k^{ce}_{FB}k^{ce}-2k^{ce}_{FB}+2k^{ce}}{2k^{ce}_{FB}+1}h-\frac{2k^{ce}_{FB}(k^{ce}-k^{ce}_{FB})}{2k^{ce}_{FB}+1}h\\
&=\frac{2(k^{ce}_{FB})^2-2k^{ce}_{FB}+2k^{ce}}{2k^{ce}_{FB}+1}h \\
&\geq 0,
\end{align*}
where the first and last inequalities are due to the fact that $k^{ce}_{FB}\leq k^{ce}$ when $p(1-p)r\geq 2h$.

To sum up, we have $W^{ce}_{OB}-W^{ce}_{NB}\geq W^{ce}_{FB}- W^{ce}_{OB}$.

$\Box$

\subsection*{H. Proofs of Proposition \ref{prop:comparison in decentralized setting}}
Note that 
\begin{align*}
    W^{de}_{FB}&=pr_{HH}+(1-p)r_{LL}-\frac{p(1-p)r}{2k^{de}+1}-\frac{2k^{de}(k^{de}+1)}{2k^{de}+1}h,\\
    W^{de}_{OB}&=pr_{HH}+(1-p)r_{LL}-\frac{p(1-p)r}{k^{de}+1}-k^{de}h, \\
    W_{IA}^{de}&=pr_{HH}+(1-p)r_{LL}-p(1-p)r.
\end{align*}
Thus, 
\begin{align*}
W^{de}_{FB}- W^{de}_{OB}&=\frac{k^{de}}{(k^{de}+1)(2k^{de}+1)}p(1-p)r-\frac{k^{de}}{2k^{de}+1}h,\\
W^{de}_{OB}-W^{de}_{NB}&=\frac{k^{de}}{k^{de}+1}p(1-p)r-k^{de}h.
\end{align*}
It is easy to see that $W^{ce}_{FB}\geq W^{ce}_{OB}\geq W^{ce}_{NB}$ when $k^{de}\leq \max\{0, \frac{p(1-p)r}{h}-1\}$ and $W^{ce}_{FB}\leq W^{ce}_{OB}\leq W^{ce}_{NB}$ when $k^{de}\geq \max\{0, \frac{p(1-p)r}{h}-1\}$.

When  $\alpha\leq \frac{(1-p)r}{2(r_{HH}-r_{HL})}$, consider two cases. (1) If $h\geq p\alpha(r_{HH}-r_{HL})$, we have $\frac{p \alpha (r_{HH}-r_{HL})}{h}<1$ and thus ${k}^{de}=\left\lfloor \frac{q \alpha (r_{HH}-r_{HL})}{h}\right\rfloor=0\leq \max\{0, \frac{p(1-p)r}{h}-1\}$. (2) If $h< p\alpha(r_{HH}-r_{HL})$, we have 
$$p\alpha(r_{HH}-r_{HL})-p(1-p)r+h< 2 p\alpha(r_{HH}-r_{HL})-p(1-p)r\leq 0, $$
where the last inequality is due to the fact that $\alpha\leq \frac{(1-p)r}{2(r_{HH}-r_{HL})}$. This implies that $\frac{p \alpha (r_{HH}-r_{HL})}{h}\leq  \frac{p(1-p)r}{h}-1$ and thus ${k}^{de}=\left\lfloor \frac{q \alpha (r_{HH}-r_{HL})}{h}\right\rfloor\leq \frac{p(1-p)r}{h}-1$. In both cases, we have $k^{de}\leq \max\{0, \frac{p(1-p)r}{h}-1\}$  and therefore $W^{ce}_{FB}\geq W^{ce}_{OB}\geq W^{ce}_{NB}$.

When $\alpha\geq \frac{(1-p)r}{(r_{HH}-r_{HL})}$, we have $\frac{p \alpha (r_{HH}-r_{HL})}{h}> \frac{p(1-p)r}{h}$, which implies that ${k}^{de}=\left\lfloor \frac{q \alpha (r_{HH}-r_{HL})}{h}\right\rfloor\geq \frac{p(1-p)r}{h}-1$. As a consequence, when  $\alpha\geq \frac{(1-p)r}{(r_{HH}-r_{HL})}$, $W^{ce}_{FB}\leq W^{ce}_{OB}\leq W^{ce}_{NB}$.

When $\alpha\in (\frac{(1-p)r}{2(r_{HH}-r_{HL})}, \frac{(1-p)r}{(r_{HH}-r_{HL})})$, if $h\in\{h: \left\lfloor \frac{q \alpha (r_{HH}-r_{HL})}{h}\right\rfloor\leq \max\{0, \frac{p(1-p)r}{h}\}\}$, $W^{ce}_{FB}\geq W^{ce}_{OB}\geq W^{ce}_{NB}$. If $h\in\{h: \left\lfloor \frac{q \alpha (r_{HH}-r_{HL})}{h}\right\rfloor\leq \max\{0, \frac{p(1-p)r}{h}\}\}$, $W^{ce}_{FB}\leq W^{ce}_{OB}\leq W^{ce}_{NB}$.
\hfill$\Box$

\subsection*{I. Numerical Results for Varying $p$ and $q$}\label{appendix: numberical results}

 In this section, we report additional numerical examples to demonstrate the robustness of the results discussed in Section 5. The following numerical results are conducted under the parameters of $\alpha=0.2$ and $\mathbf{r}=(800, 50, 50, 0)$.

The impact of $h$ on social welfare in decentralized systems as $p$ and $q$ change: 
\begin{figure}[htbp]
    \centering
\subfigure[$p=0.5$ and $q=0.1$\label{subfig:a}]{
\includegraphics[width=0.28\textwidth]      {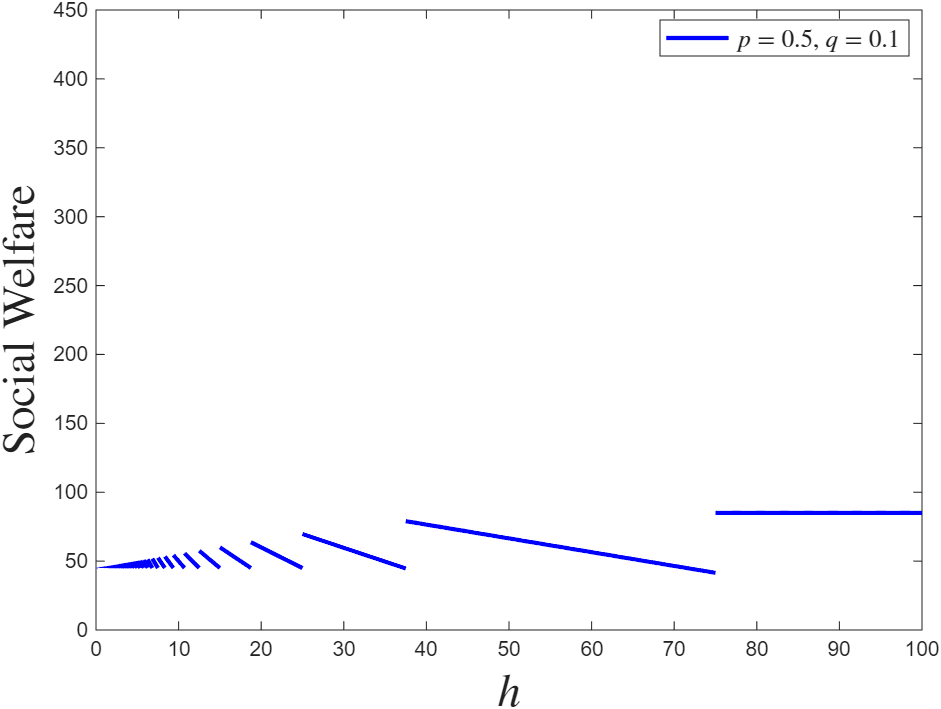}}
\hfill
\subfigure[$p=0.5$ and $q=0.2$\label{subfig:b}]{    \includegraphics[width=0.28\textwidth]{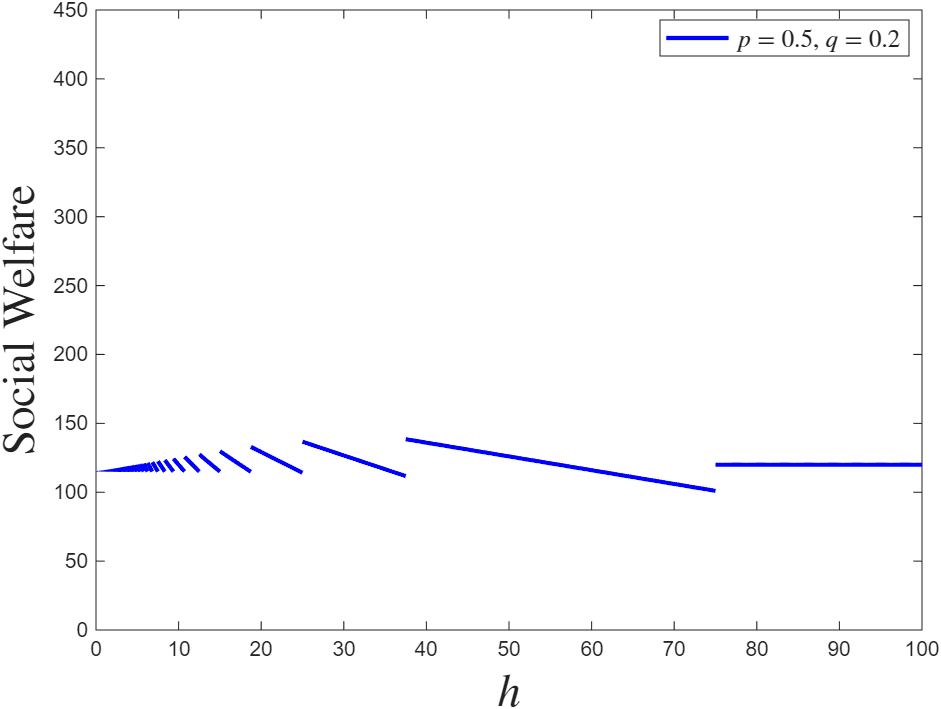}}
\hfill
\subfigure[$p=0.5$ and $q=0.3$\label{subfig:c}]{
\includegraphics[width=0.28\textwidth]{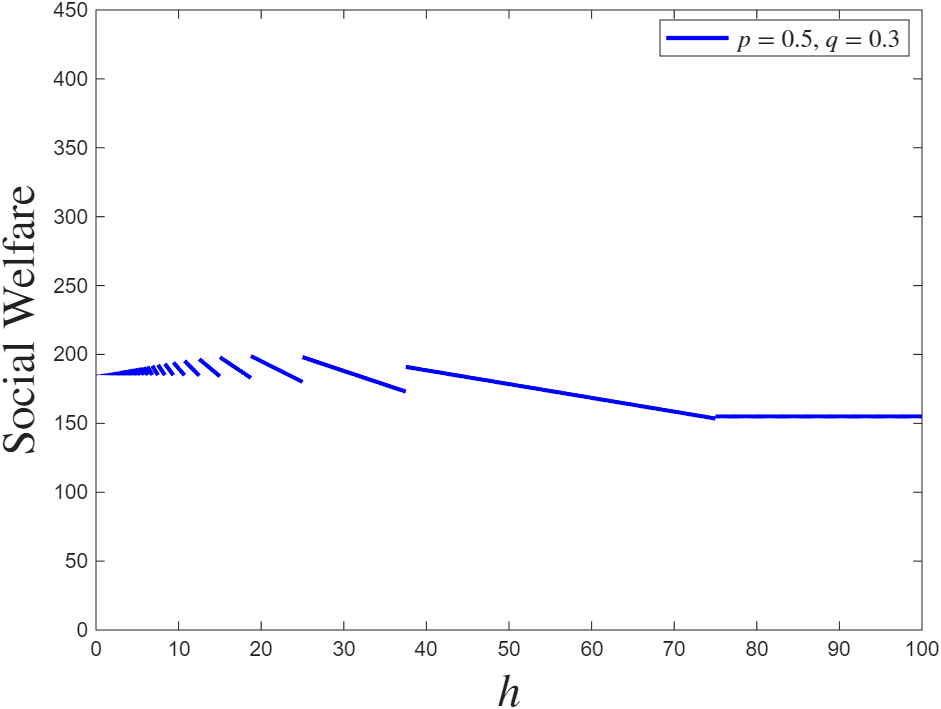}}
\vspace{8pt}
\subfigure[$p=0.5$ and $q=0.4$\label{subfig:a}]{
\includegraphics[width=0.28\textwidth]      {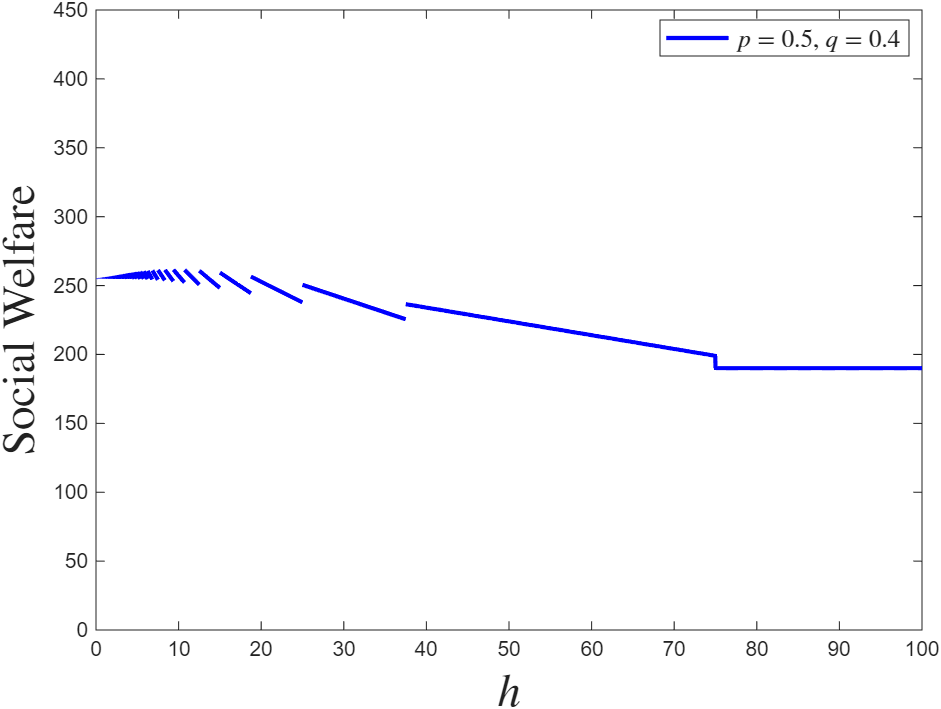}}
\hfill
\subfigure[$p=0.5$ and $q=0.5$\label{subfig:b}]{    \includegraphics[width=0.28\textwidth]{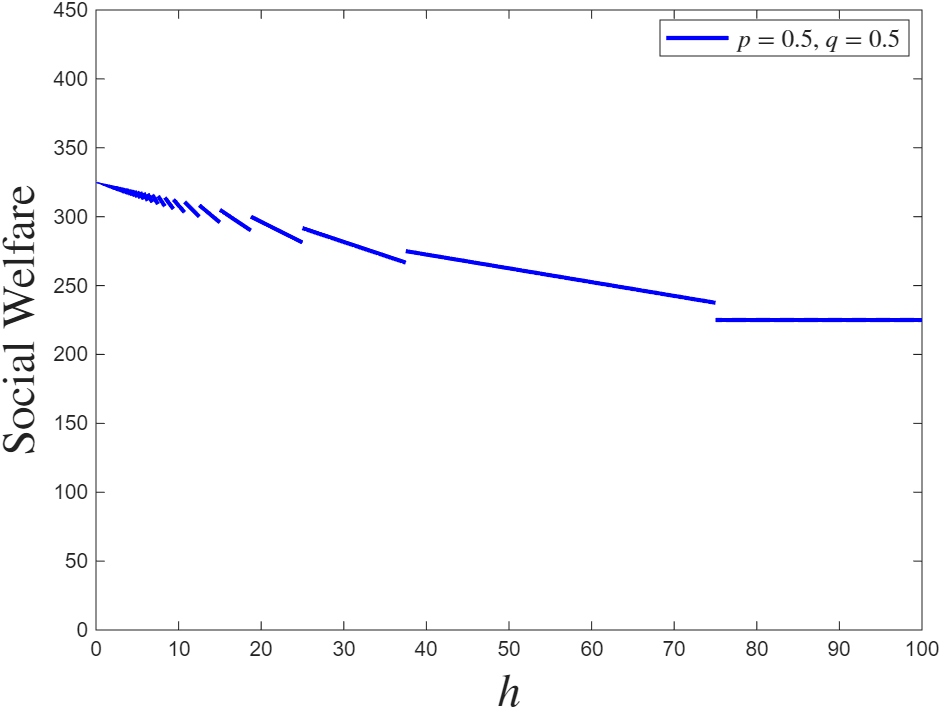}}
\hfill
\subfigure[$p=0.5$ and $q=0.6$\label{subfig:c}]{
\includegraphics[width=0.28\textwidth]{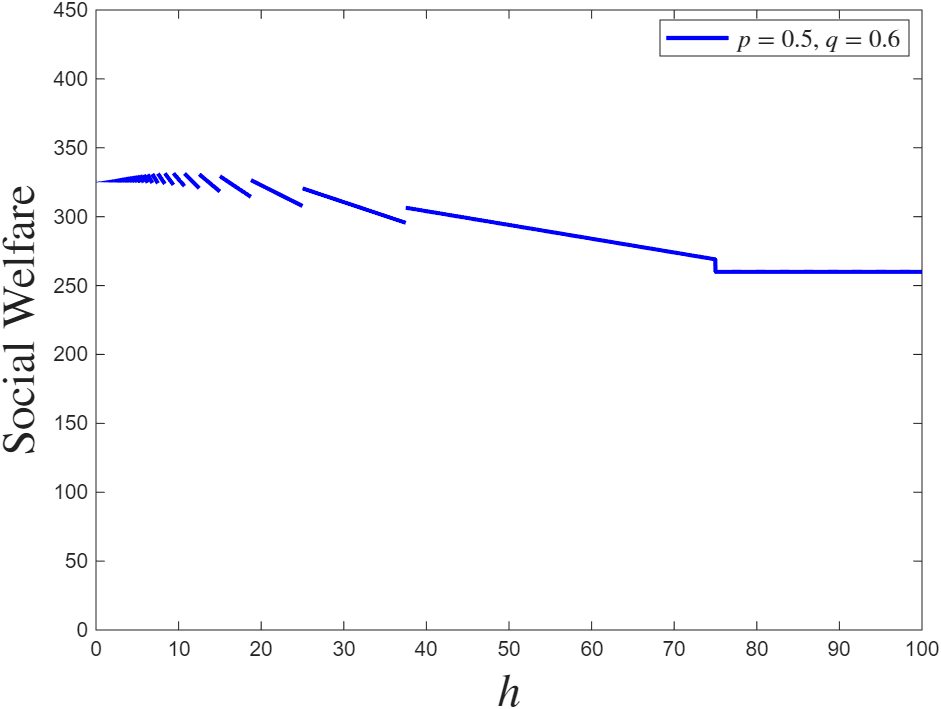}}
\vspace{8pt}
\subfigure[$p=0.5$ and $q=0.7$\label{subfig:a}]{
\includegraphics[width=0.28\textwidth]      {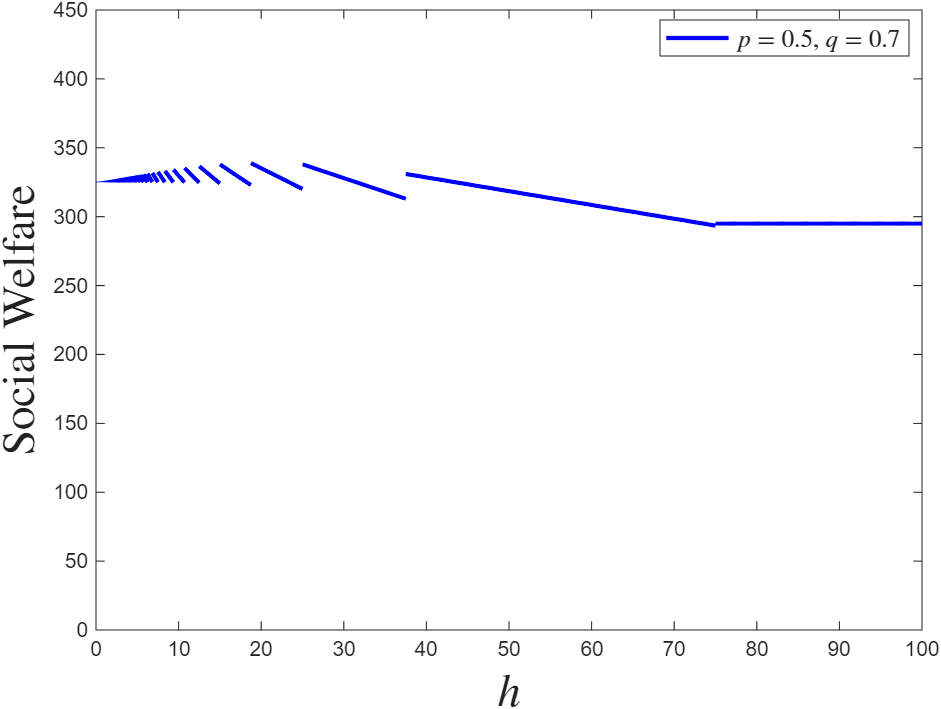}}
\hfill
\subfigure[$p=0.5$ and $q=0.8$\label{subfig:b}]{    \includegraphics[width=0.28\textwidth]{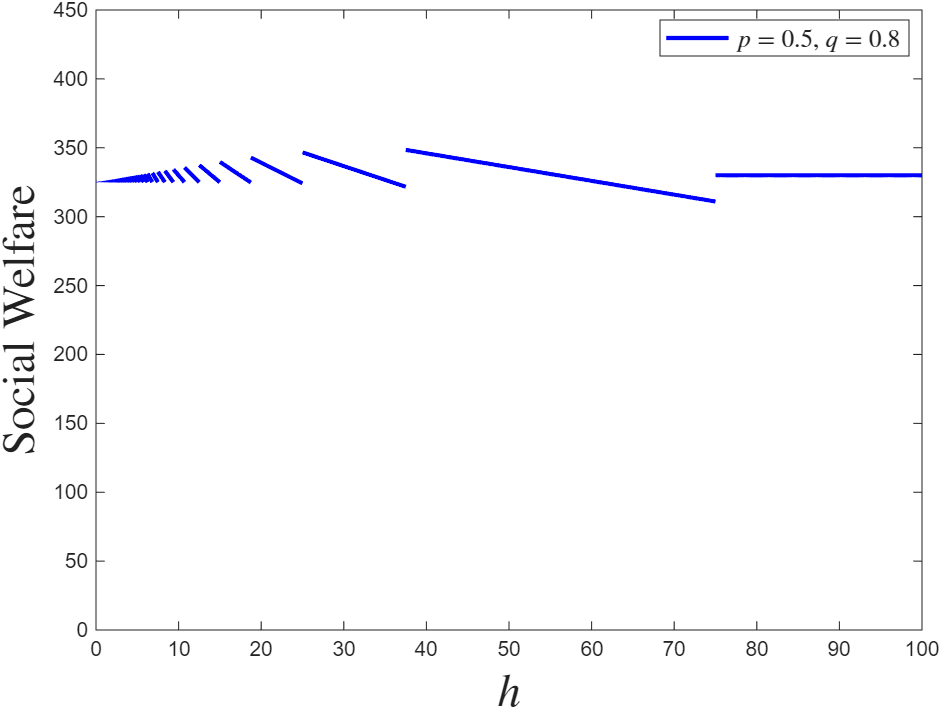}}
\hfill
\subfigure[$p=0.5$ and $q=0.9$\label{subfig:c}]{
\includegraphics[width=0.28\textwidth]{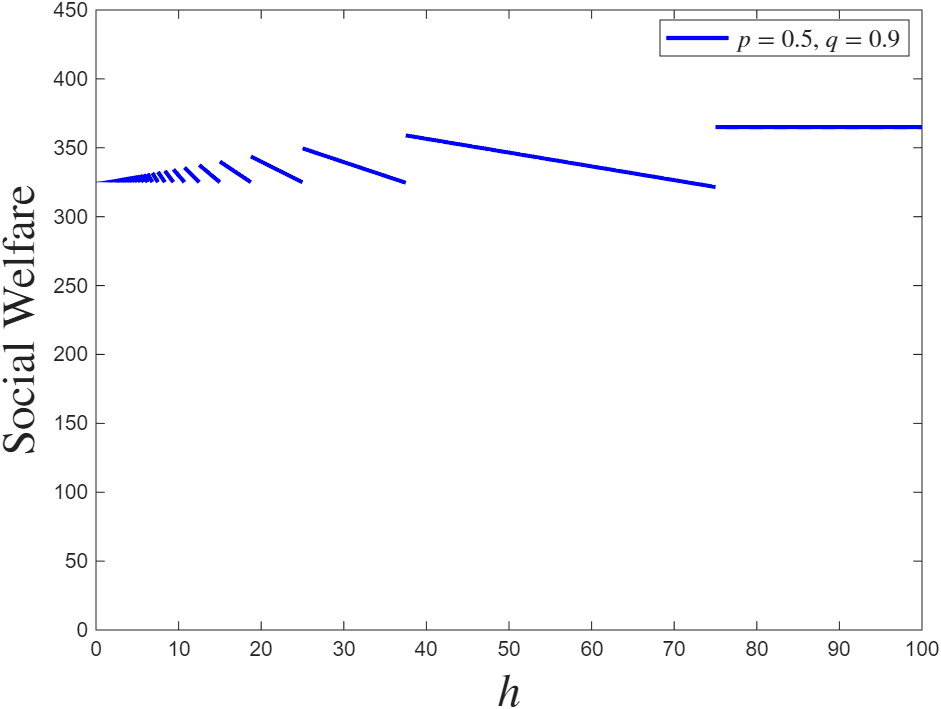}}
\caption{Impact of $h$ on social welfare in the decentralized system when $p=0.5$ and varying $q$}
 \end{figure}

 \begin{figure}[htbp]
 \centering
 \subfigure[$p=0.3$ and $q=0.7$\label{subfig:a}]{
\includegraphics[width=0.35\textwidth]   {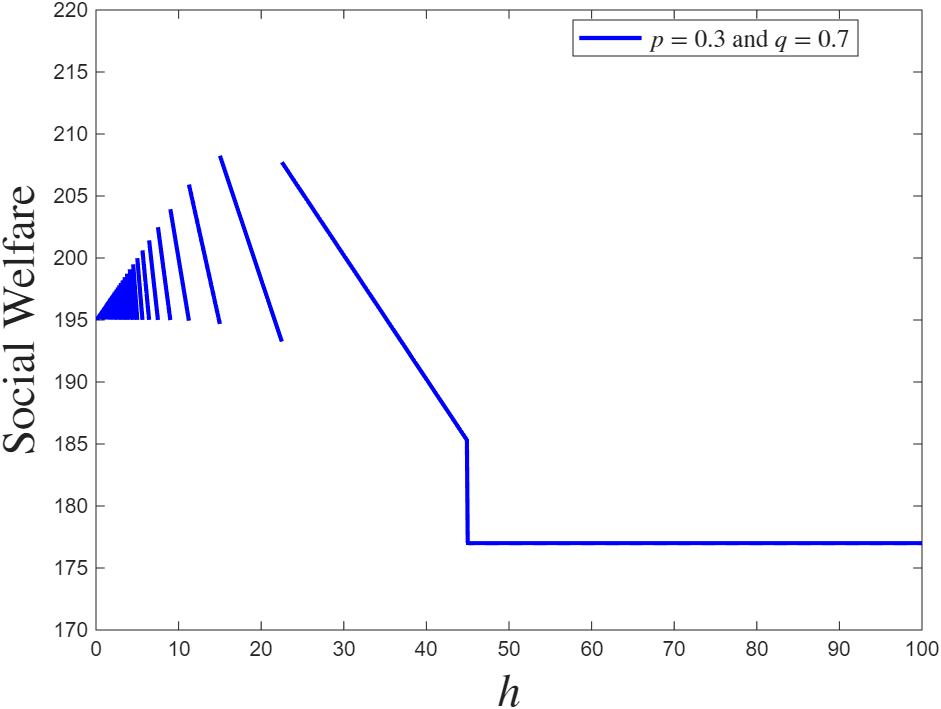}}
\subfigure[$p=0.7$ and $q=0.3$\label{subfig:c}]{
\includegraphics[width=0.35\textwidth]{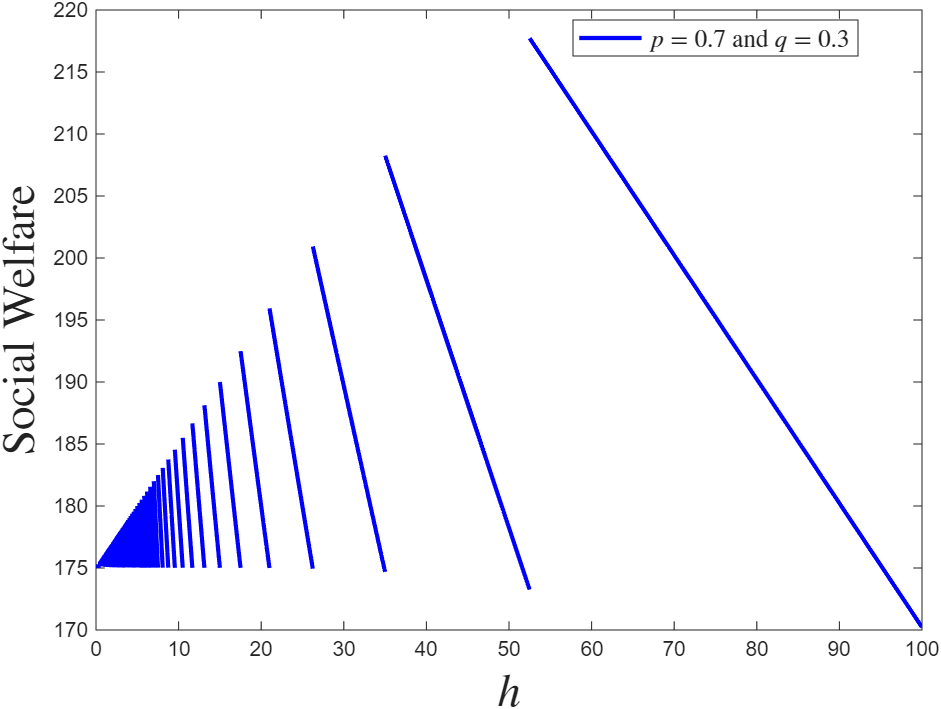}}\\
\vspace{8pt}
\subfigure[$p=0.2$ and $q=0.8$\label{subfig:d}]{
\includegraphics[width=0.35\textwidth]   {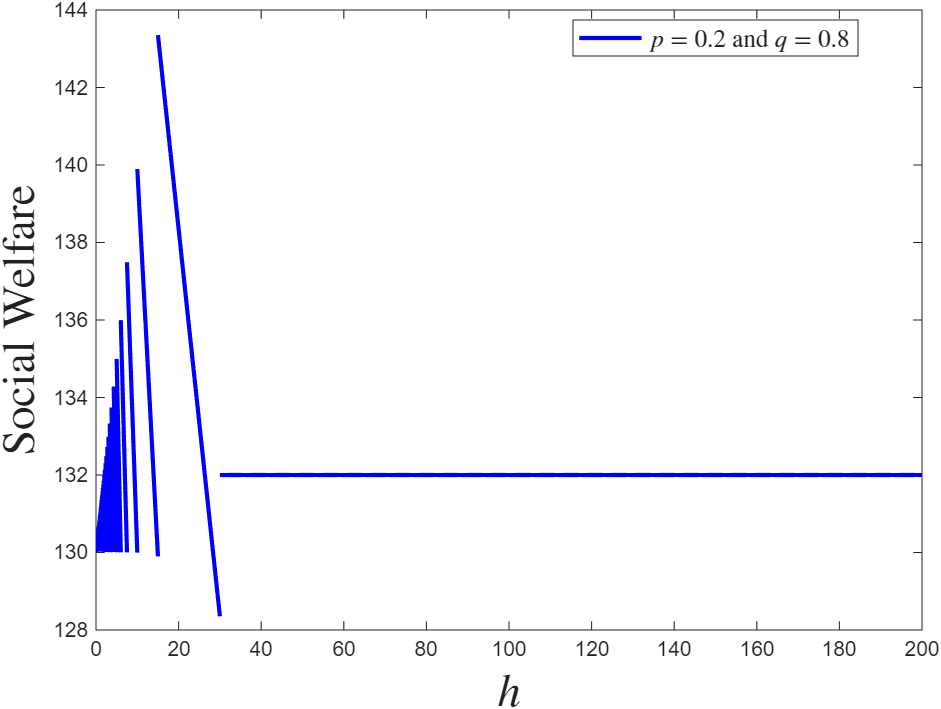}}
\subfigure[$p=0.8$ and $q=0.2$\label{subfig:e}]{
\includegraphics[width=0.35\textwidth]{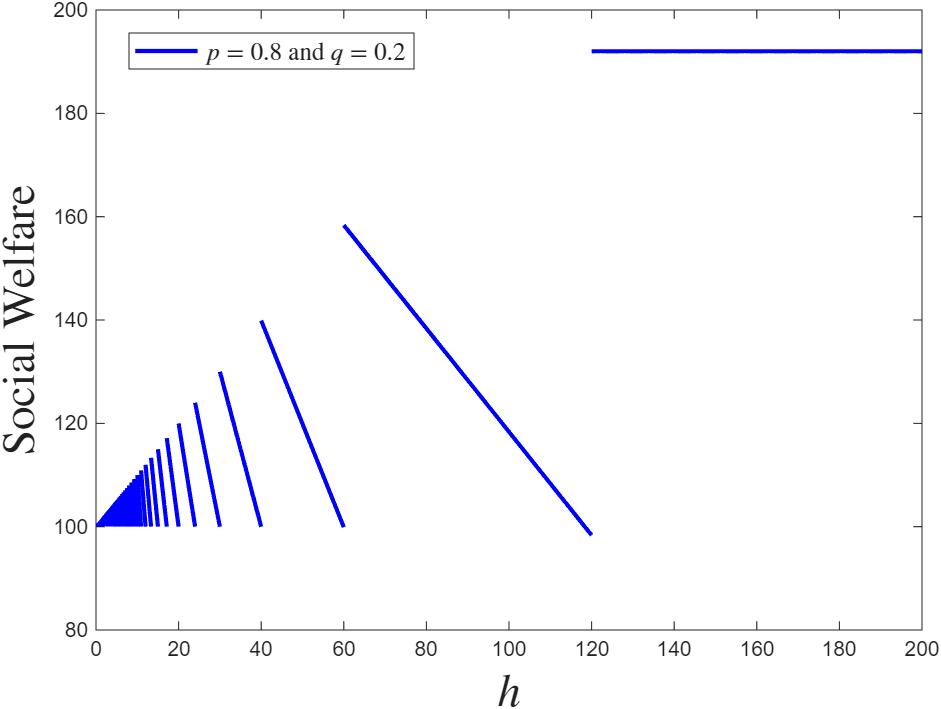}}\\
\vspace{8pt}
\subfigure[$p=0.1$ and $q=0.9$\label{subfig:d}]{
\includegraphics[width=0.35\textwidth]   {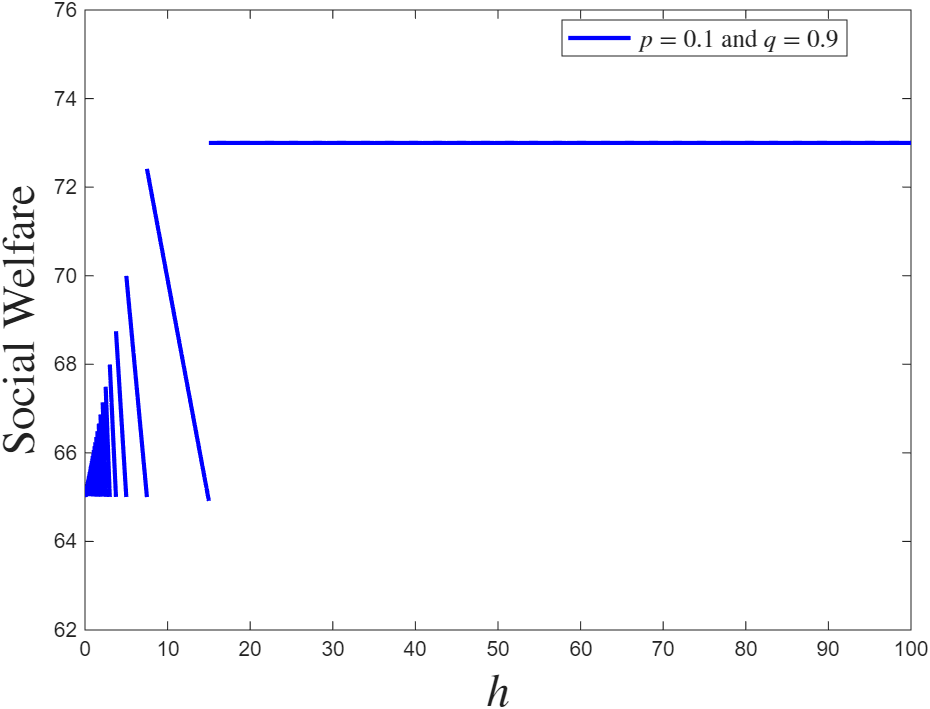}}
\subfigure[$p=0.9$ and $q=0.1$\label{subfig:e}]{
\includegraphics[width=0.35\textwidth]{figures/figures_appendix/decentralized_sw_p0.8_q0.2.png}}
\caption{Impact of $h$ on social welfare in the decentralized system when $p+q=1$}
 \end{figure}

\begin{figure}[htbp]
    \centering
\subfigure[$p=q=0.1$\label{subfig:a}]{
\includegraphics[width=0.28\textwidth]      {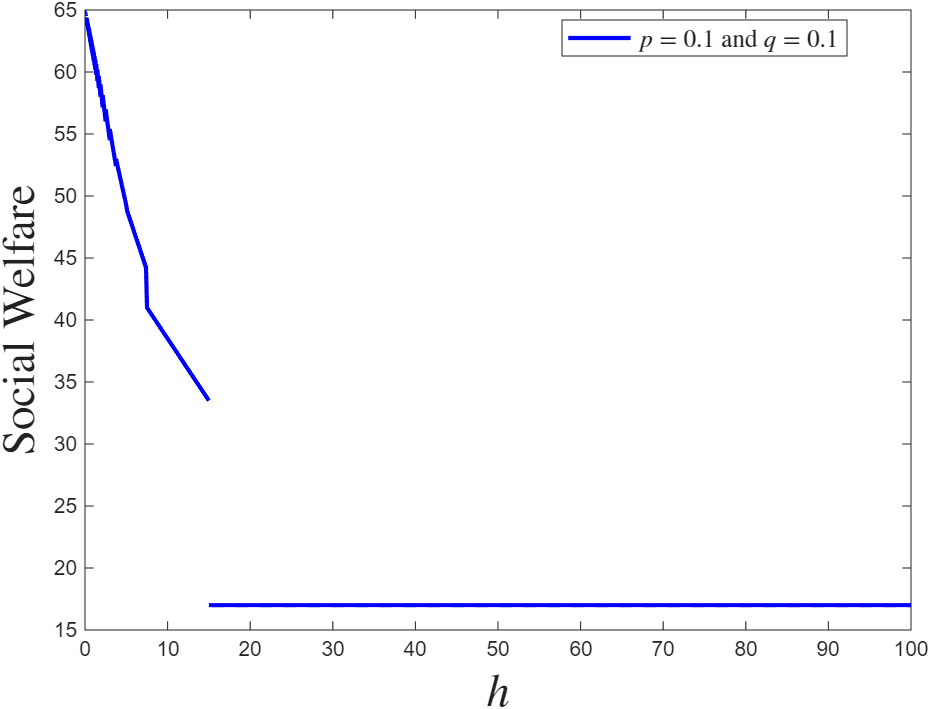}}
\hfill
\subfigure[$p=q=0.2$\label{subfig:b}]{    \includegraphics[width=0.28\textwidth]{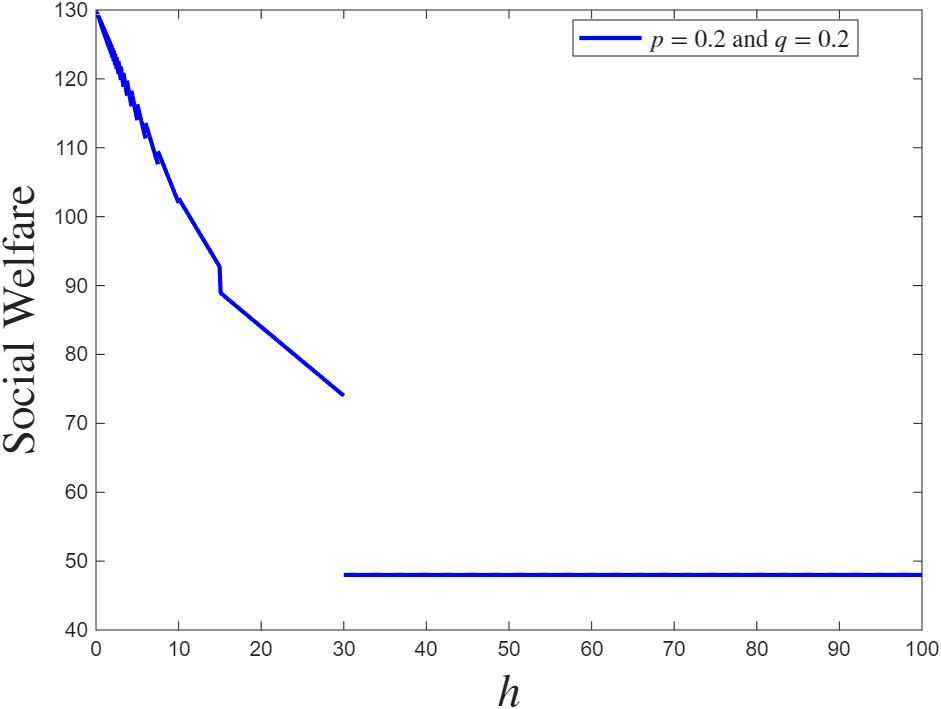}}
\hfill
\subfigure[$p=q=0.3$\label{subfig:c}]{
\includegraphics[width=0.28\textwidth]{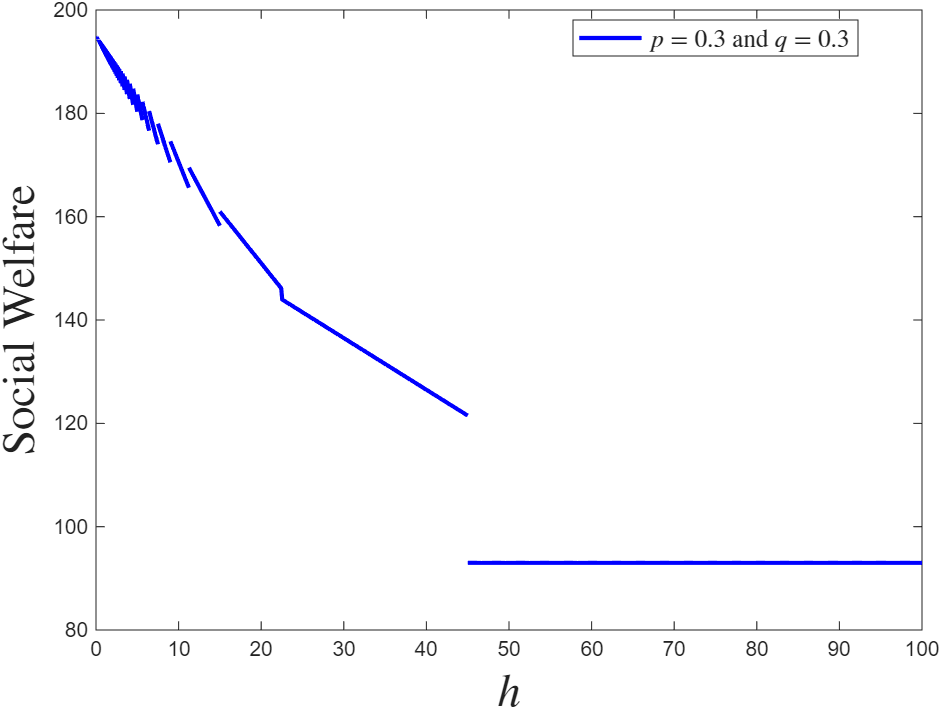}}
\vspace{8pt}
\subfigure[$p=q=0.4$\label{subfig:a}]{
\includegraphics[width=0.28\textwidth]      {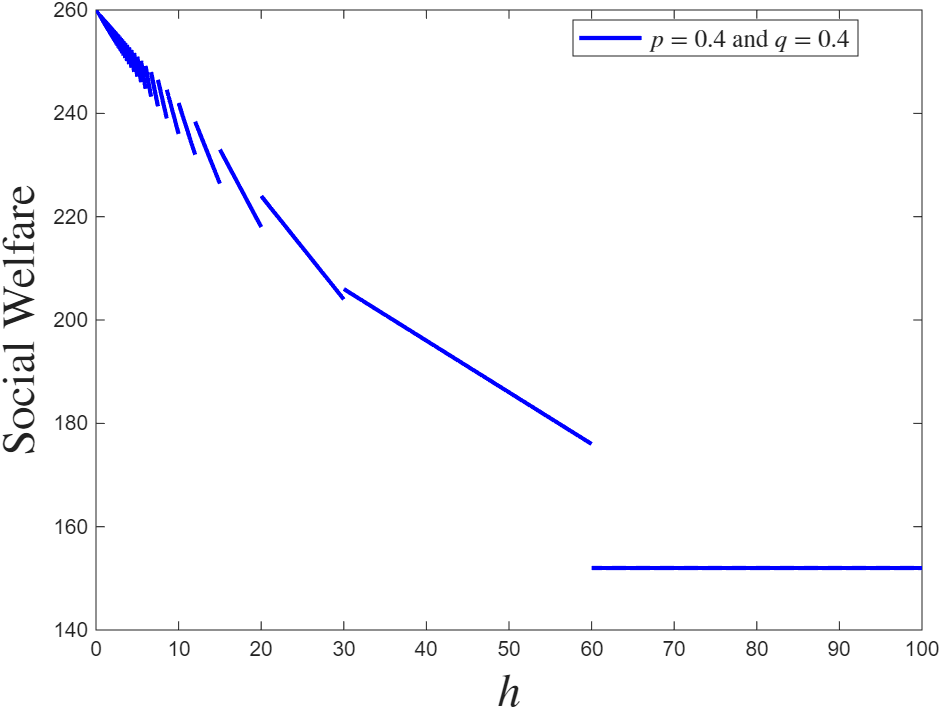}}
\hfill
\subfigure[$p=q=0.5$\label{subfig:b}]{    \includegraphics[width=0.28\textwidth]{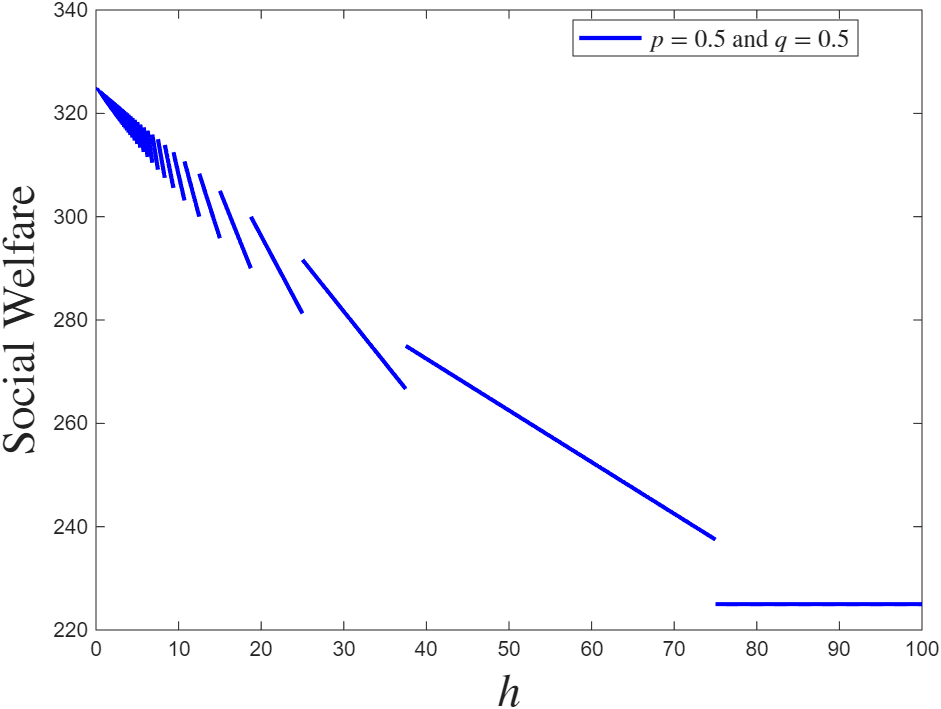}}
\hfill
\subfigure[$p=q=0.6$\label{subfig:c}]{
\includegraphics[width=0.28\textwidth]{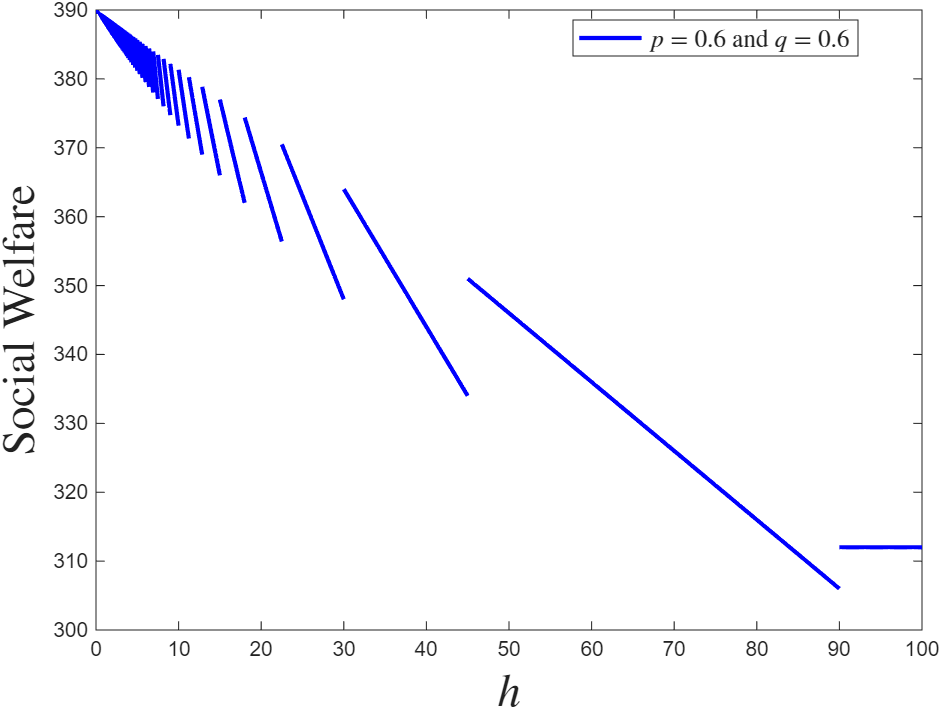}}
\vspace{8pt}
\subfigure[$p=q=0.7$\label{subfig:a}]{
\includegraphics[width=0.28\textwidth]      {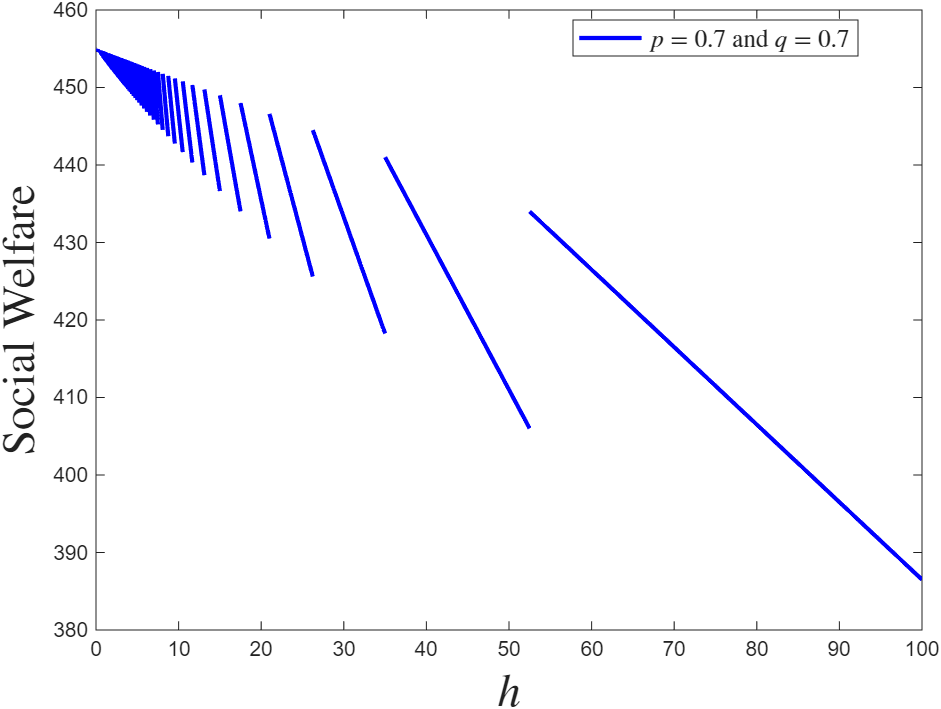}}
\hfill
\subfigure[$p=q=0.8$\label{subfig:b}]{    \includegraphics[width=0.28\textwidth]{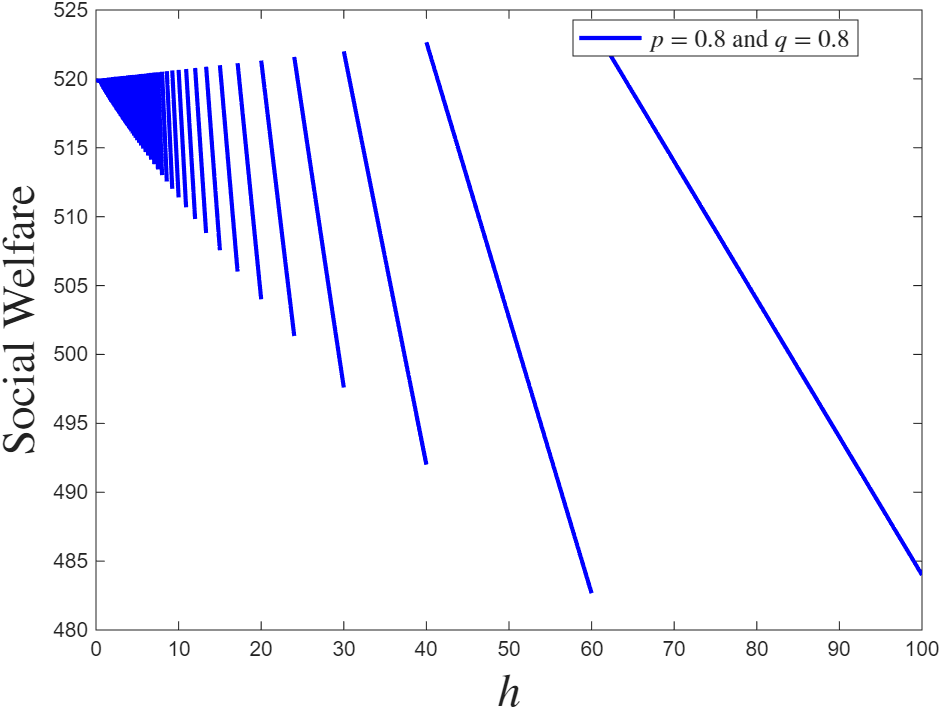}}
\hfill
\subfigure[$p=q=0.9$\label{subfig:c}]{
\includegraphics[width=0.28\textwidth]{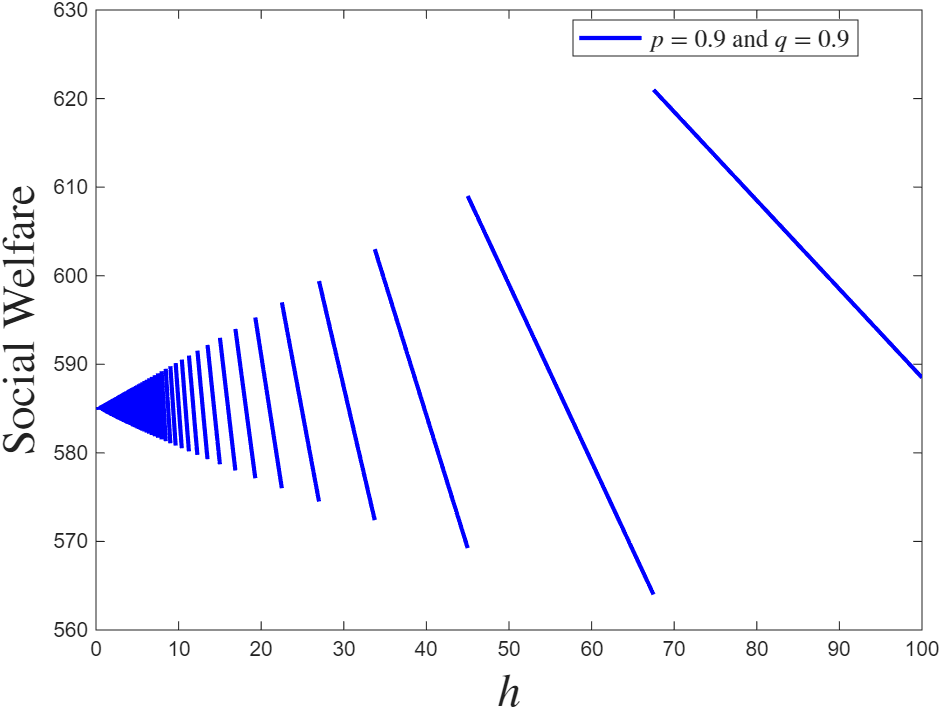}}
\caption{Impact of $h$ on social welfare in the decentralized system when $p=q$}
 \end{figure}

\end{document}